\documentclass[a4paper, UKenglish]{article}

\usepackage{latexsym}
\usepackage{amssymb}
\usepackage{amsmath,mathabx}
\usepackage{amsthm}
\usepackage{pstricks}
\usepackage{subfig}
\usepackage{algorithm}
\usepackage[noend]{algorithmic}
\usepackage{xspace,enumerate}
\usepackage{datetime}
\usepackage{tipa}
\usepackage{textcmds}

\usepackage{chngcntr}

\usepackage[normalem]{ulem}

\usepackage{tikz}
\usepackage{tikz-qtree}
\usepackage{chapterbib}
\usetikzlibrary{arrows,automata,backgrounds,calendar,chains,decorations,er,fadings,fit,matrix,mindmap,folding,patterns,petri,plothandlers,plotmarks,shadows,shapes,topaths,through,trees}

\theoremstyle{definition}
\newtheorem{theorem}{Theorem}
\newtheorem{definition}[theorem]{Definition}
\newtheorem{proposition}[theorem]{Proposition}

\newtheorem{example}[theorem]{Example}

\newtheorem{lemma}[theorem]{Lemma}
\newtheorem{corollary}[theorem]{Corollary}

\graphicspath{{./}}

\begin{document}
	\newcommand{\nc}{\newcommand}


\newcommand  {\subs}     {\subseteq}
\renewcommand{\sups}{\supseteq}
\newcommand{\covers}{\sqsupseteq}
\newcommand{\covered}{\sqsubseteq}
\newcommand{\mydef}{\ensuremath{\mathrel{\smash{\stackrel{\scriptscriptstyle{
    \text{def}}}{=}}}}}

\newcommand{\ohne}{\ensuremath{\setminus}}
\newcommand{\Pot}{{\ensuremath{\mathcal{P}}}}

\newcommand{\Dom}{\mathcal{D}}
\newcommand{\Rng}{\mathcal{R}}

\newcommand{\eps}{{\ensuremath{\epsilon}}}
\newcommand{\pprime}{{\ensuremath{\prime \prime}}}

\newcommand {\nat}      {\mathbb{N}}

\newcommand{\Err}{\text{err}}
\newcommand{\Null}{\bot}

\newcommand{\calA}{{\ensuremath{\mathcal{A}}}\xspace}
\newcommand{\calB}{{\ensuremath{\mathcal{B}}}\xspace}
\newcommand{\calC}{\ensuremath{\mathcal{C}}\xspace}
\newcommand{\calD}{\ensuremath{\mathcal{D}}\xspace}
\newcommand{\calE}{\ensuremath{\mathcal{E}}\xspace}
\newcommand{\calG}{\ensuremath{\mathcal{G}}\xspace}
\newcommand{\calI}{\ensuremath{\mathcal{I}}\xspace}
\newcommand{\calP}{\ensuremath{\mathcal{P}}\xspace}
\newcommand{\calQ}{\ensuremath{\mathcal{Q}}\xspace}
\newcommand{\calS}{\ensuremath{\mathcal{S}}\xspace}

\newcommand{\bigO}{\ensuremath{\mathcal{O}}}

\newcommand{\blank}{\ensuremath{\sqcup}}

\newcommand{\leftidx}[3]{{\vphantom{#2}}#1#2#3}
\newcommand{\lequiv}[3][]{\ensuremath{\tensor*[_{#2}]{\equiv}{^{#1}_{#3}}}}

\newcommand{\id}{\ensuremath{\text{id}}}

\newcommand{\hatE}{\ensuremath{\hat{E}}}
\newcommand{\hate}{\ensuremath{\hat{e}}}
\newcommand{\hatP}{\ensuremath{\hat{P}}}
\newcommand{\hatQ}{\ensuremath{\hat{Q}}}
\newcommand{\hatS}{\ensuremath{\hat{S}}}
\newcommand{\hatG}{\ensuremath{\hat{G}}}
\newcommand{\ol}[1]{\overline{#1}}

\newcommand{\posbool}{\ensuremath{\mathcal{B}^+}}

\newcommand{\Sigmat}{\tilde{\Sigma}}

\newcommand{\smin}[1]{\ensuremath{\left[#1\right]_{\text{min}}}}

\newcommand{\Exp}{\text{Exp}}


\def\pone{\textsc{Juliet}\xspace}
\def\ptwo{\textsc{Romeo}\xspace}
\def\pones{\textsc{Juliet}'s\xspace} 
\def\ptwos{\textsc{Romeo}'s\xspace}
\def\ponea{\textsc{J}\xspace} 
\def\ptwoa{\textsc{R}\xspace}

\def\playerA{\textsc{Adam}\xspace}
\def\playerE{\textsc{Eve}\xspace}
\def\playerAs{\textsc{Adam's}\xspace}
\def\playerEs{\textsc{Eve's}\xspace}

\newcommand{\funcsymb}{\Gamma}
\newcommand{\ponesymb}{\funcsymb_{\ponea}}
\newcommand{\ptwosymb}{\funcsymb_{\ptwoa}}

\newcommand{\movesingle}[1]{\ensuremath{\stackrel{#1}{\longrightarrow}}}
\newcommand{\movesingles}{\movesingle{\sigma}}
\newcommand{\move}[2]{\ensuremath{\stackrel{#1,#2}{\longrightarrow}}}
\newcommand{\moves}[1]{\move{\sigma}{#1}}
\newcommand{\movest}{\move{\sigma}{\tau}}

\newcommand{\tmovesingle}[1]{\ensuremath{\stackrel{#1}{{\longrightarrow}^*}}}
\newcommand{\tmovesingles}{\tmovesingle{\sigma}}
\newcommand{\tmove}[2]{\ensuremath{\stackrel{#1,#2}{\longrightarrow}^*}}
\newcommand{\tmoves}[1]{\tmove{\sigma}{#1}}
\newcommand{\tmovest}{\tmove{\sigma}{\tau}}

\newcommand {\play}      {{\ensuremath{\Pi}}\xspace}

\newcommand{\passes}{\ensuremath{\text{LS}}\xspace}

\newcommand{\selG}{\ensuremath{\tilde{G}}\xspace}
\newcommand{\Glr}{\ensuremath{G_{\text{L2R}}}\xspace}

\newcommand{\Depth}[1][G]{\ensuremath{\text{Depth}^{#1}}}


\newcommand{\algprobname}[1]{\ensuremath{\text{\sc{#1}}}\xspace}
\newcommand{\lrall}{\algprobname{L2RAll}\xspace}
\newcommand{\corridor}{\algprobname{TwoPlayer2ExpCorridorTiling}\xspace}
\newcommand{\AIT}{\algprobname{AIT($\epsilon$-NWT)}\xspace}
\newcommand{\coiiiSAT}{\algprobname{co-3-SAT}\xspace}

\providecommand  {\myclass} [1]  {\ensuremath{\mbox{\sc #1}}\xspace}

\newcommand{\PTIME}{\myclass{PTIME}}
\newcommand{\NP}{\myclass{NP}}
\newcommand{\coNP}{\myclass{co-NP}}
\newcommand{\PSPACE}{\myclass{PSPACE}}
\newcommand{\NPSPACE}{\myclass{NPSPACE}}
\newcommand{\APSPACE}{\myclass{APSPACE}}
\newcommand{\LOGSPACE}{\myclass{LOGSPACE}}
\newcommand     {\EXPTIME}  {\myclass{EXPTIME}}
\newcommand     {\iiEXPTIME}  {\myclass{2-EXPTIME}}
\newcommand     {\kEXPTIME}  {\myclass{$k$-EXPTIME}}
\newcommand     {\EXPSPACE}  {\myclass{EXPSPACE}}
\newcommand     {\NEXPSPACE}  {\myclass{NEXPSPACE}}
\newcommand     {\AEXPSPACE}  {\myclass{AEXPSPACE}}
\newcommand     {\NEXPTIME}  {\myclass{NEXPTIME}}
\newcommand     {\kNEXPTIME}  {\myclass{$k$-NEXPTIME}}
\newcommand     {\coNEXPTIME}  {\myclass{co-NEXPTIME}}
\newcommand     {\cokNEXPTIME}  {\myclass{co-$k$-NEXPTIME}}
\newcommand     {\cokpNEXPTIME}  {\myclass{co-$(k+1)$-NEXPTIME}}
\newcommand     {\iiiEXPTIME}  {\myclass{3-EXPTIME}}

\newcommand{\QBF}{\algprobname{QBF}}

\newcommand{\icup}{\ensuremath{\oplus}} 
\newcommand{\bigicup}{\ensuremath{\bigoplus}} 


\nc{\win}{JWin}

\newcommand{\safeu}{\ensuremath{\text{\win}}}
\newcommand{\safelrrf}{\ensuremath{\text{\win}_{L2R}^{\text{rf}}}}
\newcommand{\safelrol}[1]{\mbox{\ensuremath{\Sigma^*\setminus\safelr(#1)}}}
\newcommand{\safelrp}{\ensuremath{\text{\win}_{\text{L2R}^+}}}
\newcommand{\safelsi}{\ensuremath{\text{\win}_{LS1}}}
\newcommand{\safeop}[1][3]{\ensuremath{\text{\win}_{1P(#1)}}}
\newcommand{\safeR}{\safeop[4]}


\newcommand{\ltrp}{\ensuremath{\text{L2R}^+}}
\newcommand{\stratu}{\ensuremath{\text{STRAT}}}
\newcommand{\stratlr}{\ensuremath{\text{STRAT}}} 
\newcommand{\stratlrrf}{\ensuremath{\text{STRAT}_{L2R}^{rf}}}
\newcommand{\stratlrpc}{\ensuremath{\text{STRAT}_{L2R}^{pc}}}
\newcommand{\stratrf}{\ensuremath{\text{STRAT}^\text{rf}_{L2R}}}
\newcommand{\stratls}{\ensuremath{\text{STRAT}_{LS1}}}
\newcommand{\stratp}[1][3]{\ensuremath{\text{STRAT}_{1P(#1)}}}
\newcommand{\strata}{\ensuremath{\text{STRAT}_{\ponea}}}
\newcommand{\stratac}{\ensuremath{\text{STRAT}_{\ponea,\Call}}}
\newcommand{\stratb}{\ensuremath{\text{STRAT}_{\ptwoa}}}
\newcommand{\stratr}[1][G]{\ensuremath{\stratb(#1)}\xspace}

\newcommand{\Astrat}{{\ensuremath{\sigma}}}
\newcommand{\Bstrat}{{\ensuremath{{\tau}}}}

\newcommand{\ov}[1]{\ensuremath{\overline{#1}}}
\newcommand{\ovS}{\ensuremath{\ov{S}}}
\newcommand{\ovAstrat}{\ensuremath{\ov{\Astrat}}}
\newcommand{\ovBstrat}{\ensuremath{\ov{\Bstrat}}}

\newcommand{\func}[1][]{\ensuremath{\text{func}^{#1}}}
\newcommand{\funcs}[1][]{\ensuremath{\text{funcs}^{#1}}}
\newcommand{\funcslr}[1][G]{\ensuremath{s_\text{L2R}^{#1}}}
\newcommand{\word}[1][G]{\ensuremath{\text{word}_{#1}}}
\newcommand{\words}[1][G]{\ensuremath{\text{words}_{#1}}}
\newcommand{\state}[1][G]{\ensuremath{\text{state}_{#1}}}
\newcommand{\states}[1][G]{\ensuremath{\text{states}_{#1}}}
\newcommand{\suf}[1][G]{\ensuremath{\text{suf}^{#1}}}
\newcommand{\sufstates}[1][G]{\ensuremath{\text{sufstates}^{#1}}}
\newcommand{\Bstates}[1][]{\ensuremath{\text{2-states}^{#1}}}
\newcommand{\transit}{\ensuremath{\text{transit}}}

\newcommand{\cstate}[1][C]{\ensuremath{\text{state}_{\game{#1}}}}
\newcommand{\cstates}[1][C]{\ensuremath{\text{states}_{\game{#1}}}}


\newcommand{\wset}{\ensuremath{\mathcal{L}}}
\newcommand{\wsetRead}{\ensuremath{\mathcal{L}_{Read}}}
\newcommand{\wsetCall}{\ensuremath{\mathcal{L}_{Call}}}

\nc{\Call}{\ensuremath{\text{Call}}\xspace}
\nc{\Read}{\ensuremath{\text{Read}}\xspace}
\nc{\LS}{\ensuremath{\text{LS}}\xspace}
\nc{\Stop}{\ensuremath{\text{Stop}}\xspace}

\newcommand{\deriv}{\ensuremath{\Rightarrow}}


\newcommand{\effect}[2][]{\ensuremath{\calE^{#1}[#2]}\xspace} 
\newcommand{\geffect}[2][]{\ensuremath{\calE^{#1}[G,#2]}\xspace} 
\newcommand{\ceffect}[1][G]{\ensuremath{\calC[#1]}\xspace} 
\newcommand{\cneffect}[2][G]{\ensuremath{\calC^{#2}[#1]}\xspace} 
\newcommand{\game}[1]{\ensuremath{\calG(#1)}\xspace} 

\newcommand{\safelr}[1][]{\ensuremath{\text{\win}^{#1}}}
\newcommand{\Safelr}[1][]{\ensuremath{\textsc{JWin}^{#1}}}
\newcommand{\Safelrfin}{\Safelr[\text{fin}]}

\newcommand{\Gloc}{\ensuremath{\calG_{\text{loc}}}}
\newcommand{\Gnloc}{\ensuremath{\calG_{\text{nd,loc}}}}
\newcommand{\Gxsc}{\ensuremath{\calG_{\text{xsc}}}}
\newcommand{\Gall}{\ensuremath{\calG_{\text{all}}}}
\newcommand{\Gnall}{\ensuremath{\calG_{\text{nd,all}}}}
\newcommand{\Gfin}{\ensuremath{\calG_{\text{fin}}}}
\newcommand{\Gnd}{\ensuremath{\calG_{\text{nd}}}}
\newcommand{\Gndf}{\ensuremath{\calG_{\text{ndf}}}}
\newcommand{\Gins}{\ensuremath{\calG_{\text{ins}}}}
\newcommand{\Gval}[1][]{\ensuremath{\calG_{\text{val}}^{#1}}}
\newcommand{\Gflat}{\ensuremath{\calG_{\text{flat}}}}

\newcommand{\Floc}{\ensuremath{F_{\text{loc}}}}

\newcommand    {\algproblem} [3] 
{\fbox{\parbox[t]{.95\textwidth}{\centerline{#1}\begin{tabular}{lp{.75\textwidth}} Given: & #2\\ Question: & #3\end{tabular}}}}


\newcommand{\anRC}[1]{\ensuremath{{#1}_{\text{RC}}}}
\newcommand{\anR}[1]{\ensuremath{{#1}_{\text{R}}}}
\newcommand{\anC}[1]{\ensuremath{{#1}_{\text{C}}}}
\newcommand{\anNull}[1]{\ensuremath{{#1}_{\text{0}}}}
\newcommand{\anRandC}[1]{\ensuremath{{#1}_{\text{R$\land$C}}}}
\newcommand{\anRorC}[1]{\ensuremath{{#1}_{\text{R$\lor$C}}}}

\newcommand{\unan}[1]{\ensuremath{\natural({#1})}}

\newcommand{\Contract}{\ensuremath{\text{Contract}}}
\newcommand{\Root}{\ensuremath{\text{root}}}

\newcommand{\linstate}[1]{\ensuremath{\overline{#1}}}


\newcommand{\calls}[1][G]{\ensuremath{S^{#1}_{\text{call}}}}
\newcommand{\summary}{s}
\newcommand{\seffect}[2][1]{\ensuremath{E^{#1}(#2)}} 
\newcommand{\teffect}{\ensuremath{\tilde{E}}} 
\newcommand{\releffect}[2][]{\ensuremath{e^{#1}(#2)}\xspace} 
\newcommand{\deffect}[2][]{\ensuremath{\hat{E}^{#1}[#2]}\xspace}
\newcommand{\fctE}{\ensuremath{E_{\text{fct}}}}
\newcommand{\rveffect}{\ensuremath{e_{v}}}
\newcommand{\veffect}{\ensuremath{E_{v}}}
\newcommand{\finE}{\ensuremath{E^{<\infty}}}

\newcommand{\dreleffect}[1]{\ensuremath{\hat{e}(#1)}\xspace} 
\newcommand{\ehat}{\ensuremath{\hat{e}}\xspace} 

\newcommand{\uhat}{\ensuremath{\hat{u}}\xspace}

\newcommand{\effectset}[2][2]{\ensuremath{\mathcal{E}_{#1}(#2)}\xspace}

\newcommand{\Norm}{\ensuremath{\textsc{Norm}}\xspace}
\newcommand{\Mix}{\ensuremath{\textsc{Mix}}\xspace}
\newcommand{\SMix}{\ensuremath{\textsc{SMix}}\xspace}
\newcommand{\Pick}{\ensuremath{\textsc{Shuffle}}\xspace}
\newcommand{\Shuffle}{\ensuremath{\textsc{Shuffle}}\xspace}

\newcommand{\comp}{\ensuremath{\text{comp}}}
\newcommand{\decomp}{\ensuremath{\text{decomp}}}


\newcommand{\Alr}{{\ensuremath{A_{\text{L2R}}}}}
\newcommand{\Als}{{\ensuremath{A_{\text{LS1}}}}}

\newcommand{\deltalr}{{\ensuremath{\delta_{\text{L2R}}}}\xspace}
\newcommand{\deltaslr}{{\ensuremath{\delta^*_{\text{L2R}}}}\xspace}
\newcommand{\deltalrhat}{{\ensuremath{\hat{\delta}_{\text{L2R}}}}\xspace}
\newcommand{\deltaslrhat}{{\ensuremath{\hat{\delta}^*_{\text{L2R}}}}\xspace}

\newcommand{\dAlr}{{\ensuremath{\widehat{A}_{\text{L2R}}}}}

\newcommand{\Asep}{{\ensuremath{A_{\text{sep}}}}}
\newcommand{\Aut}[2][]{{\ensuremath{A_{#1 #2}}}}

\newcommand{\ovdelta}{{\ensuremath{\overline{\delta}}}}
\newcommand{\tdelta}{{\ensuremath{\tilde{\delta}}}}

\newcommand{\reach}[1][]{\ensuremath{\stackrel{#1}{\leadsto}}}

\newcommand{\Yield}{\ensuremath{\text{Yield}}}

\newcommand{\layer}[1][]{\ensuremath{#1\text{-layer}}}


\newcommand{\Last}{\text{LAST}}

\newcommand{\wf}[1][\Sigma]{\ensuremath{\text{NW}(#1)}}
\newcommand{\rwf}[1][\Sigma]{\ensuremath{\text{rNW}(#1)}}
\newcommand{\swf}[1][\Sigma]{\ensuremath{\text{sNW}(#1)}}
\newcommand{\op}[1]{\ensuremath{\text{\small\textlangle}{#1}\text{\small\textrangle}}}
\newcommand{\cl}[1]{\ensuremath{\text{\small\textlangle}/{#1}\text{\small\textrangle}}}

\newcommand{\nw}[1]{\ensuremath{\widehat{#1}}}

\newcommand{\true}{\ensuremath{\text{true}}}
\newcommand{\false}{\ensuremath{\text{false}}}

\newcommand{\Aug}{\ensuremath{\text{Aug}}}
\newcommand{\Enc}{\ensuremath{\text{Enc}}}

\newcommand{\Strip}{\ensuremath{\text{Strip}}}
\newcommand{\Val}{\ensuremath{\text{Val}}}


\newcommand{\We}{\ensuremath{\mathcal{W}}}
\newcommand{\Se}{\ensuremath{\mathcal{S}}}
\newcommand{\Te}{\ensuremath{\mathcal{T}}}
\newcommand{\Ge}{\ensuremath{\mathcal{G}}}
\newcommand{\He}{\ensuremath{\mathcal{H}}}

\newcommand{\Hh}{\ensuremath{\mathcal{H}}}
\newcommand{\Vv}{\ensuremath{\mathcal{V}}}

\newcommand{\ovWe}{\ensuremath{\overline{\mathcal{W}}}}
\newcommand{\ovTe}{\ensuremath{\overline{\mathcal{T}}}}
\newcommand{\ovGe}{\ensuremath{\overline{\mathcal{G}}}}

\newcommand{\TT}{\ensuremath{\mathbb{T}}}
\newcommand{\GG}{\ensuremath{\mathbb{G}}}
\newcommand{\NN}{\ensuremath{\mathbb{N}}}


\newcommand{\added}[1]{\color{blue}#1\color{black}}

\newcommand{\mcomment}[2]{{\footnotesize\color{blue}(#1)}\footnote{\color{blue}#1: #2}} 

\newcommand{\tsm}[1]{\mcomment{TS}{#1}}
\newcommand{\msm}[1]{\mcomment{MS}{#1}}

\newcommand{\thomas}[1]{\ \\ \fbox{\parbox{\linewidth}{{\sc Thomas}:\\
      #1}}}
\nc{\thomasm}{\tsm}

\newcommand{\martin}[1]{\ \\ \fbox{\parbox{\linewidth}{{\sc Martin}:\\ #1}}}
\nc{\martinm}{\msm}


\newcommand{\bild}[4]{        
\begin{figure}[h]
	\begin{center}
		\scalebox{#1}{\input{Bilder/#2}} \caption{#3} \label{#4}
	\end{center}          
\end{figure}           
}

\newcommand{\ignore}[1]{}

\newcommand{\skipproof}[1]{}

\newenvironment{restate}[1]
{\noindent \textbf{#1 \itshape (restated).} \itshape} 
{}

\newenvironment{proofof}[1]
{\noindent \textbf{Proof of #1.}} 
{\qed}

	\title{Transducer-based Rewriting Games for Active XML}
	\author{Martin Schuster\\ \and \small TU Dortmund University}
	\maketitle

	\begin{abstract}
		\emph{Context-free games} are two-player rewriting games that are played on nested strings representing XML documents with embedded function symbols. 
These games were introduced to model rewriting processes for intensional documents in the \emph{Active XML} framework, where input documents are to be rewritten into a given target schema by calls to external services.

This paper studies 
the setting where dependencies between inputs and outputs of service calls are modelled by transducers, which has not been examined previously. It defines transducer models operating on nested words and studies their properties, as well as the computational complexity of the winning problem for transducer-based context-free games in several scenarios. While the complexity of this problem is quite high in most settings (ranging from $\NP$-complete to undecidable), some tractable restrictions are also identified.
	\end{abstract}

		\section{Introduction}\label{sec:intro}

\paragraph*{Scientific context}

\emph{Context-free games} on strings 
are two-player games extending context-free grammars,
 with the first player (called \pone) choosing the non-terminal to be replaced and the second player (called \ptwo) choosing a replacement for that non-terminal. The winning condition for \pone is reaching, at some point during the game, some string in a given target language over the combined alphabet of non-terminals and terminals.

These games were first introduced in \cite{MuschollSS06} to model
the rewriting process of \emph{Active XML (AXML)} \cite{AbiteboulBM08} documents. 
The intention of AXML is modelling \emph{intensional} documents, i.e. documents that do not store all required information explicitly but instead contain references to external services, from which current information may be materialised on demand, as illustrated in the example below. To this end, AXML extends standard
XML with \emph{function nodes} referring to external web services that may be called to insert data into the AXML document when the document is requested. Context-free games abstract from AXML to model the uncertainty inherent in using external data.

A standard example (cf. \cite{MiloAABN05,MuschollSS06}) of an application for AXML is depicted in Figure \ref{fig:rewriting}. In this example, we consider (part of) an AXML document retained by a local online news site providing information about current weather and events. Initially, the server-side document looks like the one in Figure \ref{subfig:before}. The nodes labelled @weather\_svc and @events\_svc are function nodes referring to external weather and event services.

\begin{figure}[h]
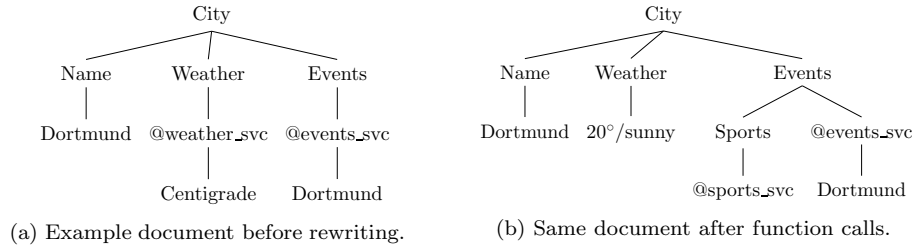

	\centering
	\subfloat[c][Example document before rewriting.]{
		\label{subfig:before}
		\resizebox{.4\textwidth}{!}{		
			\Tree [.City [.Name Dortmund ] [.Weather [.@weather\_svc Centigrade ] ] [.Events [.@events\_svc Dortmund ] ] ]
		}
	}
	\qquad
	\subfloat[c][Same document after function calls.]{
		\label{subfig:after}
		\resizebox{.49\textwidth}{!}{		
			\Tree [.City [.Name Dortmund ] [.Weather 20$^\circ$/sunny ] [.Events [.Sports @sports\_svc ] [.@events\_svc Dortmund ] ] ]	
		}	
	}\\
	\caption{Example of Active XML rewriting. After calls to function nodes @weather\_svc and @events\_svc in Fig. \ref{subfig:before}, external data is materialised to yield the document in Fig. \ref{subfig:after}. }
	\label{fig:rewriting}
\end{figure}

Figure \ref{subfig:after} shows the document from Figure \ref{subfig:before} after both function nodes have been called, replacing them by the external services' results. As exemplified by the function node labelled @weather\_svc, call results replace the entire subtree rooted at the called function node, with that subtree being passed to the external service as a parameter. Concretely, the @weather\_svc node's child tells the weather service that temperatures returned should be in centigrade. As the call result of the @event\_svc node shows, returns of external services may contain further function nodes, even copies of the called function node.

The \emph{safe rewriting problem} \cite{MiloAABN05} of determining whether a given AXML document can always be rewritten into a target schema 
was abstracted in \cite{MuschollSS06} into the problem of determining whether \pone has a winning strategy in a given context-free game on strings, with \pone representing a rewriting algorithm and \ptwo representing the uncertainty inherent in function calls. This research assumed DTDs
as schema formalisms. Allowing more expressive schema languages such as XML Schema \cite{SchusterS15} then led to research into context-free games on \emph{nested strings} (i.e. XML-like linearisations of trees) \cite{AlurM09}. Even though none of these previous works modelled dependencies between parameters and outputs of function calls, they already showed that the winning problem for \pone can be undecidable or of a very high complexity, unless strategies for \pone and allowed schema languages are seriously restricted.

The impact of service call parameters
has so far only been studied in a limited fashion. In \cite{SchusterS15} (and in \cite{MiloAABN05}, for AXML rewriting with DTDs), external services were modelled by \emph{input} (or \emph{validation}) and \emph{output} (or \emph{replacement}) schemas, the semantics being that a function node could only be called if its parameter subtree was valid with regard to its corresponding input schema, which would then yield a return conforming to its output schema. This is a purely syntactic handling of input parameters which models a rather simple relationship between input and output of function calls; for instance, an @event\_svc call reproducing its input parameters in its output as shown in Figure \ref{fig:rewriting} cannot be enforced in this model.

While dependencies between parameters and outputs of service calls have always been implicit in the AXML model, they have not been studied in detail so far. Therefore, we extend the context-free games on nested words from \cite{SchusterS15} by (generally non-deterministic) transformations relating function parameters to possible outputs. We define \emph{nested word transducers} (NWT) as a comparatively simple finite representation for transformations on nested words that naturally extends the nested word automata used in \cite{SchusterS15}. We then study the complexity of the winning problem for \pone in various restrictions of context-free games with transducer-based replacement. As auxiliary results of potentially independent interest, we also examine closure properties and basic algorithmic problems of NWT.

\paragraph*{Contributions}

In light of prior undecidability and complexity results, the main objective of this paper is finding suitable restrictions to transducer-based context-free games that render the winning problem for \pone decidable with as low complexity as possible. The two basic types of restrictions we examine are strategy restrictions (i.e. restrictions to \pones capabilities of calling function symbols) and restrictions to the type of NWT used for rewriting. To avoid undecidability, we only allow left-to-right strategies, i.e. once a function node has been called, no function calls to  nodes preceding it (in post-order)
 are possible (cf. \cite{MuschollSS06}). 

The most important class of strategy restrictions considered here are restrictions to games with limited \emph{replay}. In general context-free games, after a function call, \pone continues her rewriting on that function call's result; we call this the \emph{unbounded replay} case, as \pone may continue rewriting function call results for as long as new function nodes are returned.
In the \emph{replay-free} case, \pone is instead forbidden to call any function nodes inside results of function calls. As an intermediate case between unbounded replay and replay-free games, we also consider \emph{bounded replay} games, where \pone may only call functions returned by function calls up to a fixed maximum recursion depth. For instance, in a replay-free game, \pone could call neither of the function nodes labelled @sports\_svc and @events\_svc in the situation of Figure \ref{subfig:after}; in a bounded replay game of depth 2, on the other hand, she could call these nodes, but not any function nodes returned by those secondary calls.

The second type of restrictions comes from limiting expressiveness of the transducer used in games. Generally, transducers are allowed to be \emph{non-functional}, i.e. any input string, may have several transducts (to model the fact that function call results are dependent upon, but not uniquely determined by, input parameters). The main types of transducers examined here are the following:
\begin{itemize}
\item \emph{Nested word transducers (NWT)} allow for transforming input strings into output strings that are arbitrarily long, regardless of the input string's size. 
\item \emph{Nested word transducers without $\epsilon$-transitions ($\epsilon$-free NWT)} may only increase the size of an input string by no more than a linear factor.
\item \emph{Relabelling transducers} may only change labels of input strings, not their structure.
\item As a special case, \emph{functional relabelling transducers} are relabelling transducers whose output string is uniquely determined by their input.
\end{itemize}

This paper's main complexity results are summarised in Table \ref{tab:results}. The central insight here is that the least restricted settings yield an undecidable or non-elementary winning problem, and even for strong restrictions, the complexity of the winning problem is generally quite high, with no tractable case among the standard settings.
For this reason, we also study several limitations of these settings, derived from our lower bound proofs, in order to reduce complexity:
\begin{itemize}
\item \emph{Depth-bounded NWT} lower the complexity of the replay-free case to $\EXPSPACE$-complete (in comparison to $\iiEXPTIME$ for general NWT).
\item Strategies with \emph{bounded \Call width} lower the complexity of the bounded-replay case for $\epsilon$-free NWT from non-elementary to $\coNEXPTIME$-complete or $\coNP$-complete, depending on the precise formalisation of bounded \Call width.
\item \emph{Write-once} strategies yield a tractable case for functional relabelling transducers in a setting even more restrictive than the replay-free one.
\end{itemize}

\begin{table}[t]
\footnotesize
  \centering
          \begin{tabular}{|p{4cm}|c|c|c|}
            \hline
           			& No replay & Bounded & Unbounded\\\hline
NWT 		& $\iiEXPTIME$ & undecidable & undecidable \\\hline
$\epsilon$-free NWT 					& $\coNEXPTIME$&  non-elementary & undecidable\\\hline
Relabelling 			& $\PSPACE$ &  $\PSPACE$ & $\EXPTIME$\\\hline
Functional relabelling & $\NP$&  $\NP$ & $\PSPACE$\\\hline
          \end{tabular}
        \caption{Summary of complexity results. All results are completeness results.}
    \label{tab:results}
  \end{table}

\paragraph*{Related Work}

Beyond the work already discussed, further complexity and decidability results for context-free games can be found in \cite{AbiteboulMB05,BjorklundSSK13}, and \cite{MuschollSS06} contains further references to related work.

Our concept of nested word transducers is based on Visibly Pushdown Transducers \cite{RaskinS08,ThomoVY08}, specifically the well-nested VPT of \cite{FiliotRRST10}. The original definitions of VPT in \cite{RaskinS08,ThomoVY08} included $\epsilon$-transitions, which were dropped from later definitions, as they caused several algorithmic problems, such as functionality and equivalence, to become undecidable (cf. \cite{Servais11}). Different from these approaches, this paper combines $\epsilon$-transitions with the restriction to well-nested words, which is (to the best of the author's knowledge) new research. 

\paragraph*{Organisation}

Section \ref{sec:prelim} gives basic notation and definitions. Section \ref{sec:epsilon} defines nested word transducers and examines their structural and algorithmic properties.
The next three sections give results on the complexity of the winning problem for games with transducer-based replacement, from most to least expressive -- general nested word transducers (Section \ref{sec:unbounded}), nested word transducers without $\epsilon$-transitions (Section \ref{sec:bounded}), and relabelling transducers (Section \ref{sec:relabel}). Each of these sections also discusses one of the restrictions with reduced complexity mentioned above.
Section \ref{sec:conclusion} concludes the paper. Due to space limitations, proofs and technical definitions are deferred to the appendix. The author is grateful to Gaetano Geck and Thomas Schwentick for careful proof-reading and valuable suggestions, and to the anonymous reviewers of MFCS 2016 for their insightful and constructive comments.
			
		\section{Preliminaries}\label{sec:prelim}

For any natural number $n \in \mathbb{N}$, we denote by $[n]$ the set $\{1, \ldots, n\}$.
For finite sets $M$, $\Pot(M)$ denotes the powerset of $M$, i.e. the set of all subsets of $M$.
For an alphabet $\Sigma$, we denote the set of finite strings over $\Sigma$ by $\Sigma^*$ and  $\epsilon$ denotes the empty string.

\paragraph*{Nested words}

For a finite alphabet $\Sigma$, $\op{\Sigma} \mydef \{\op{a} \mid a \in \Sigma\}$ denotes the set of all \emph{opening $\Sigma$-tags} and $\cl{\Sigma} \mydef \{\cl{a} \mid a \in \Sigma\}$ the set of all \emph{closing $\Sigma$-tags}. We denote by $\hat{\Sigma} \mydef \op{\Sigma} \cup \cl{\Sigma}$ the set of all $\Sigma$-tags. The set $\wf \subs \hat{\Sigma}^*$ of \emph{(well-)nested words} (or \emph{(well-)nested strings}) over $\Sigma$ is the smallest set such that $\epsilon \in \wf$, and if $u,v \in \wf$ and $a \in \Sigma$, then also $u \op{a}v\cl{a} \in \wf$. We (informally) associate with every nested word $w$ its \emph{canonical forest representation}, such that words  $\op{a}\cl{a}$, $\op{a}v\cl{a}$ and $uv$ correspond to an $a$-labelled leaf, a tree with root $a$ (and subforest corresponding to $v$), and the forest of $u$ followed by the forest of $v$, respectively.  
A nested string $w$ is \emph{rooted} if its corresponding forest is a tree. We denote the set of rooted nested strings over $\Sigma$ by $\rwf$. In a string $w = w_1 \ldots w_n \in \hat{\Sigma}^*$, two tags $w_i \in \op{\Sigma}$ and $w_j \in \cl{\Sigma}$ with $i<j$ are \emph{associated} if the substring $w_i \ldots w_j$ of $w$ is a rooted nested string. An opening (closing) tag $w_i$ in $w$ is \emph{unmatched}, if it has no associated closing (opening) tag in $w$.
To stress the distinction from nested strings in $\wf$, we refer to strings in $\Sigma^*$ as \emph{flat strings}. 

\paragraph*{Nested word automata}

A \emph{nested word automaton (NWA)} $A = (Q, P, \Sigma, \delta, q_0, F)$ \cite{AlurM09} is basically a pushdown automaton which performs a push operation on every opening tag and a pop operation on every closing tag, and in which the pushdown symbols are just states.
More formally, $A$ consists of a set $Q$ of \emph{linear states}, a set $P$ of \emph{hierarchical states}, an alphabet $\Sigma$, a \emph{transition relation} $\delta$, an \emph{initial state} $q_0 \in Q$, and a set  $F \subs Q$ of \emph{accepting (linear) states}.
The relation $\delta$ is a subset of the union of sets  $(Q \times \op{\Sigma} \times Q \times P)$ and $(Q \times P \times \cl{\Sigma} \times Q)$. We sometimes interpret $\delta$ as the union of two functions from $(Q \times \op{\Sigma})$ to $\Pot(Q \times P)$ and from $(Q \times P \times \cl{\Sigma})$ to $\Pot(Q)$ and write accordingly $(q',p) \in \delta(q,\op{a})$ for $(q,\op{a},q',p) \in \delta$ and $q' \in \delta(q,p,\cl{a})$ for $(q,p,\cl{a},q') \in \delta$. The semantics of NWA as well as the language $L(A)$ decided by a NWA $A$ are defined in the natural way, with a NWA accepting if it reaches a configuration with an accepting state and empty stack. If $A$ is a NWA, we call $L(A)$ a \emph{regular language} (of nested words).
A NWA is \emph{deterministic} (or DNWA) if $|\delta(q, \op{a})| = 1 = |\delta(q,p,\cl{a})|$ for all $q \in Q$, $p \in P$ and $a \in \Sigma$. In this case, we simply write $\delta(q, \op{a}) = (q', p')$ instead of $\delta(q, \op{a}) = \{(q', p')\}$ (and accordingly for $\delta(q,p,\cl{a})$).

	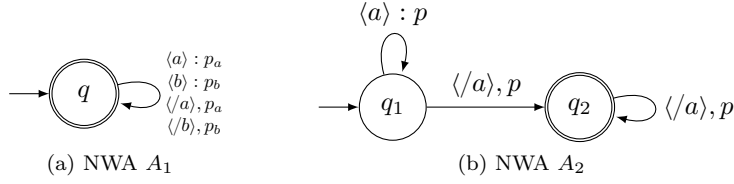
\begin{figure}[t]
	\centering
	\subfloat[c][NWA $A_1$]{
		\label{subfig:wf}
	\begin{tikzpicture}[scale=0.5,node distance=5em,initial text={},>=latex,initial distance=3em,auto]
		\node [state,initial,accepting] (q) {$q$};
		
		\path[->] (q) edge	 [loop right]	node	[align=center,scale=0.7]{$\op{a}:p_a$\\$\op{b}:p_b$\\$\cl{a},p_a$\\$\cl{b},p_b$}	();
	\end{tikzpicture}
	}
	\qquad
	\subfloat[c][NWA $A_2$]{
		\label{subfig:path}
				\begin{tikzpicture}[scale=0.5,node distance=7em,initial text={},>=latex,initial distance=3em,auto]
		\node [state,initial] (q) {$q_1$};
		\node [state,accepting] (qq) [right of=q] {$q_2$};
		
		\path[->] (q)	edge	 []		node	{$\cl{a},p$}	(qq)
						edge	 [loop above]	node	[align=center]{$\op{a}:p$}	();
		\path[->] (qq)	edge	 [loop right]	node	[align=center]{$\cl{a},p$}	();
	\end{tikzpicture}
	}\\
	\caption{NWAs $A_1$ and $A_2$ from Example \ref{ex:nwa}}
	\label{fig:nwa}
\end{figure}

\begin{example}\label{ex:nwa}
The NWA $A_1$ (Fig. \ref{subfig:wf}) checks that its input string is well-nested by pushing hierarchical state $p_a$ (resp. $p_b$) to the stack on each opening $\op{a}$ (resp. $\op{b}$) tag and popping an according hierarchical state with each matching closing tag. In this manner, $A_1$ decides the set of all well-nested strings over $\{a,b\}$. The NWA $A_2$ (Fig. \ref{subfig:path}) initially pushes a hierarchical state $p$ each time it reads $\op{a}$ in linear state $q_1$, then changes linear state to $q_2$ on reading the first $\cl{a}$ and accepts iff each initial $\op{a}$ is matched by a $\cl{a}$. In this manner, it decides the language $\{\op{a}^n \cl{a}^n \mid n \geq 1\}$. 
\end{example}

\paragraph*{Context-free games}

A \emph{context-free game (with transduction) on nested words (cfG)} $G=(\Sigma,\funcsymb,R,T)$ consists
of a finite alphabet $\Sigma$, a set $\Gamma\subs\Sigma$ of \emph{function symbols}, a \emph{(replacement) rule set} $R\subs \rwf \times \wf$ and a \emph{target language} $T\subs\wf$. 
We will only consider the case where $T$ is a non-empty regular nested word language and replacement rules are given by nested word transducers, to be defined in Section \ref{sec:epsilon}.
A play of $G$ is played by two players, \pone\/ and \ptwo, on a word $w\in \wf$. In a nutshell, \pone moves the focus along $w$ from left to right and decides for each closing tag $\cl{a}$, whether she plays a \emph{Read} or, in case $a\in\Gamma$,  a \emph{Call} move. In the latter case, \ptwo then replaces the rooted word $u$ ending at the position of $\cl{a}$ with some word $v$ with $(u,v) \in R$ and the focus is set on the first symbol of $v$. If no such word $v$ exists, \ptwo immediately wins the play.
 In case of a \Read move,
 the focus just moves further on. \pone wins a play if the word obtained at its end is in $T$.

\paragraph*{Strategies}

A \emph{strategy} for player $p\in\{\ponea,\ptwoa\}$ maps 
game states 
where player $p$ is to move into allowed moves for player $p$, i.e. strategies $\Astrat$ for \pone return moves in $\{\Read,\Call\}$ while strategies $\Bstrat$ for \ptwo return replacement strings in $\wf$. Given an initial word $w$ and strategies  $\Astrat,\Bstrat$ the play $\play(\Astrat,\Bstrat,w)$ according to $\Astrat$ and $\Bstrat$ on $w$ is uniquely determined. A \emph{winning strategy} for \pone is a strategy $\Astrat$ such that \pone wins the play $\play(\Astrat,\Bstrat,w)$, for every $\Bstrat$ of \ptwo.
By $\safelr(G)$ we denote the set of all words for which \pone has a winning strategy in $G$.

The \emph{\Call depth} of a play $\play$ is the maximum nesting depth of \Call moves in $\play$, 
if this maximum exists. That is, the \Call depth of a play is zero, if no \Call is played at all, and one, if no \Call is played inside a string yielded by a replacement move.
For a  strategy $\Astrat$ of \pone and a string $w\in\wf$, the \emph{\Call depth} $\Depth(\Astrat,w)$ of $\Astrat$ on $w$ is the maximum \Call depth in any play $\play(\Astrat,\Bstrat,w)$. A strategy $\Astrat$ has \emph{$k$-bounded \Call depth} if $\Depth(\Astrat, w) \leq k$ for all $w \in \wf$.
As a more intuitive formulation, we use the concept of \emph{replay}:
Strategies for \pone  of \Call depth one are called \emph{replay-free}, and strategies of $k$-bounded \Call depth, for any $k$, have \emph{bounded replay}. 

\paragraph*{Algorithmic problems}

In this paper, we study the following algorithmic problem $\Safelr(\calG)$ for various
classes $\calG$ of context-free games with replacement transducers. 

\begin{centering}
  \algproblem{$\Safelr(\calG)$}{A context-free game $G\in\calG$ and a string
    $w$.}{Is $w \in \safelr(G)$?}\\
\end{centering}

A class $\calG$ of context-free games in $\Safelr(\calG)$ generally comes with three parameters:
\begin{itemize}
\item the representation of the target language $T$,
\item the representation of the replacement relation $R$, and
\item to which extent replay is restricted.
\end{itemize}

We generally assume target languages to be represented by DNWAs, because the complexity of the winning problem is already quite high under that assumption, and our main interest is in finding classes $\calG$ for which $\Safelr(\calG)$ is tractable.
Replacement relations will be given as different types of nested word transducers (defined in Section \ref{sec:epsilon}). 
By a slight abuse of notation, the replacement transducer implementing a replacement relation $R$ will also be referred to as $R$.
 
In each setting, we consider the cases of unrestricted replay, bounded replay (\Call depth $k$, for some $k$), and no  replay (\Call depth $1$). We note that replay depth is formally not an actual game parameter, but the algorithmic problem can be restricted to strategies of \pone of the stated kind. If the class $\calG$  of games is clear from the context, we often simply write $\Safelr$ instead of $\Safelr(\calG)$.
		
		\section{Nested Word Transducers}\label{sec:epsilon}

In this section, 
we define nested word transducers and examine their closure properties and complexities of algorithmic problems. Thanks to our definition putting some rather severe restrictions on the use of $\epsilon$-transitions and the allowed output of transducers, we obtain  advantageous closure properties and comparatively low complexities.

Intuitively, a NWT $T$ works much like a NWA with output and additional $\epsilon$-transitions -- $T$ reads its input from left to right and decides nondeterministically which available transition to use; on an opening (resp. closing) transition, it reads an opening (closing) input tag, changes its linear state and pushes (pops) a hierarchical state while producing an output. Opening (closing, internal) $\epsilon$-transitions do not consume input symbols but induce state changes and outputs. $T$ only produces an output string if it accepts the input string.

\begin{definition}
	A \emph{nested word transducer} (or \emph{NWT}) is a tuple \mbox{$T = (Q,P,P_{\epsilon},\Sigma,\delta,q_0,F)$} consisting of a set $Q$ of \emph{linear states}, a set $P$ of \emph{hierarchical states}, a set $P_{\epsilon} \subs P$ of \emph{hierarchical $\epsilon$-states}, an alphabet $\Sigma$, a \emph{transition relation} $\delta$, which is the union of three relations from \mbox{$(Q \times (\op{\Sigma} \cup \{\op{\epsilon}\}) \times Q \times P \times \hat{\Sigma}^*)$} (called \emph{opening} transitions), $(Q \times \{\epsilon\} \times Q \times \wf)$ (called \emph{internal} transitions) and \mbox{$(Q \times P \times (\cl{\Sigma} \cup \{\cl{\epsilon}\}) \times Q \times \hat{\Sigma}^*)$} (called \emph{closing} transitions), an \emph{initial state} $q_0 \in Q$, and a set of \emph{accepting states} $F \subs Q$, such that for all $q,q',r,r' \in Q$, $p \in P$, $a \in \Sigma \cup \{\epsilon\}$ and $u,v \in \hat{\Sigma}^*$ it holds that\footnote{These three conditions make NWT roughly correspond to \emph{synchronized visibly pushdown transducers} \cite{RaskinS08}; we mainly require them to ensure closure of regular nested word languages under NWT transduction.}
	\begin{itemize}
		\item $(q,\op{\epsilon}, q',p,u) \in \delta$ or $(q,p,\cl{\epsilon},q',u) \in \delta$ if and only if $p \in P_\epsilon$ (\emph{$\epsilon$-consistency}),
		\item if $(q,\op{a}, q',p,u) \in \delta$ and $(r,p,\cl{a},r',v) \in \delta$, then $uv \in \wf$ (\emph{well-formedness}), and
		\item for each $(q,\op{a}, q',p,u) \in \delta$ (resp. $(r,p,\cl{a},r',u) \in \delta$) with $u \neq \epsilon$, $u$ contains at least one unmatched opening (resp. closing) tag (\emph{synchronisation}).
	\end{itemize}
	As for standard NWA, we also write $(q',p,u) \in \delta(q, \op{a})$ (resp. $(q',u) \in \delta(q,p,\cl{a})$, $(q',u) \in \delta(q,\epsilon)$) instead of $(q,\op{a},q',p,u) \in \delta$ (resp. $(q,p,\cl{a},q',u), (q,\epsilon,q',u) \in \delta$).
\end{definition}

A detailed semantics definition can be found in the appendix.

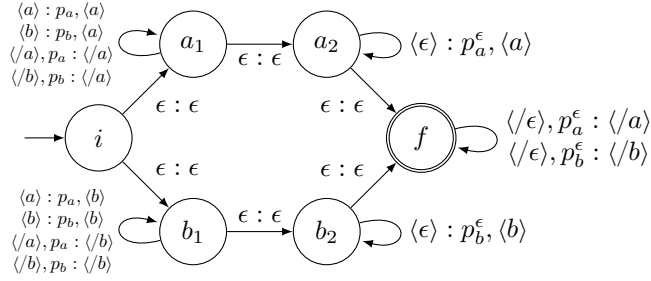
\begin{figure}[t]
	\centering
	\begin{tikzpicture}[scale=0.5,node distance=5em,initial text={},>=latex,initial distance=3em,auto]
		\node [state,initial] (i) {$i$};
		\node [state] (a) [above right of=i] {$a_1$};
		\node [state] (aa) [right of=a] {$a_2$};
		\node [state,accepting] (f) [below right of=aa] {$f$};
		\node [state] (b) [below right of=i] {$b_1$};
		\node [state] (bb) [right of=b] {$b_2$};
		
		\path[->] (i)	edge	 [swap]		node	{$\epsilon:\epsilon$}	(a)
						edge	 []	node	{$\epsilon:\epsilon$}	(b);
		\path[->] (a)	edge	 [swap]		node	{$\epsilon:\epsilon$}	(aa)
						edge	 [loop left]	node	[align=center,scale=0.7]{$\op{a}:p_a,\op{a}$\\$\op{b}:p_b,\op{a}$\\$\cl{a},p_a:\cl{a}$\\$\cl{b},p_b:\cl{a}$}	();
		\path[->] (aa)	edge	 [swap]		node	{$\epsilon:\epsilon$}	(f)
						edge	 [loop right]	node	[align=center]{$\op{\epsilon}:p^{\epsilon}_a,\op{a}$}	();
		\path[->] (b)	edge	 []	node	{$\epsilon:\epsilon$}	(bb)
						edge	 [loop left]	node	[align=center,scale=0.7]{$\op{a}:p_a,\op{b}$\\$\op{b}:p_b,\op{b}$\\$\cl{a},p_a:\cl{b}$\\$\cl{b},p_b:\cl{b}$}	();
		\path[->] (bb)	edge	 []	node	{$\epsilon:\epsilon$}	(f)
						edge	 [loop right]	node	[align=center]{$\op{\epsilon}:p^{\epsilon}_b,\op{b}$}	();
		\path[->] (f)	edge	 [loop right]	node	[align=center]{$\cl{\epsilon},p^{\epsilon}_a:\cl{a}$\\$\cl{\epsilon},p^{\epsilon}_b:\cl{b}$}	();
	\end{tikzpicture}
	\caption{Nested Word Transducer $T_{ab}$ from Example \ref{ex:nwt}.}
	\label{fig:nwt}
\end{figure}

\begin{example}\label{ex:nwt}
	Figure \ref{fig:nwt} shows a NWT $T_{ab}$, with linear states displayed as circles and transitions as arrows. From the initial state $i$, $T_{ab}$ branches nondeterministically into either state $a_1$ or $b_1$. In state $a_1$, $T_{ab}$ checks that the input string is well-nested just as the NWA $A_1$ from Example \ref{ex:nwa}. During this check, $T_{ab}$ outputs $\op{a}$ (resp. $\cl{a}$) for each opening (resp. closing) input tag, effectively relabelling the input string to consist exclusively of $a$-labelled tags. In state $a_2$, $T_{ab}$ inserts into the output string an arbitrary number of opening $\op{a}$ tags, for which a matching number of $\cl{a}$ tags are inserted in state $f$ before $T_{ab}$ accepts. The behaviour of $T_{ab}$ in states $b_1$ and $b_2$ is analogous, but outputs consist only of $b$-labelled tags. Altogether, $T_{ab}$ chooses nondeterministically some $x \in \{a,b\}$, relabels all tags of a well-nested input string into $x$-labelled tags and then appends a string of the form $\op{x}^n\cl{x}^n$.
\end{example}

	The \emph{image} $T(w)$ of a well-nested string $w \in \wf$ under $T$ is the set of all outputs of $T$ on $w$ according to some accepting run 
	of $T$ on $w$. 
	This definition extends to sets of input strings in the natural way: For a set $S \subseteq \wf$, we define $T(S) = \bigcup_{w \in S} T(w)$.
	The \emph{domain} $\Dom(T)$ of $T$ is the set of all strings $w$ such that $T(w) \neq \emptyset$, 
	and the \emph{range} $\Rng(T)$ of $T$ is the set of all strings $u$ such that there exists a $w \in \wf$ with $u \in T(w)$, i.e. the set of all possible outputs of $T$.

We next define several restrictions on the expressiveness of NWT.

\begin{definition}
	Let $T = (Q,P,P_{\epsilon},\Sigma,\delta,q_0,F)$ be a NWT. We call $T$
		\begin{itemize}
		\item \emph{$\epsilon$-free} if $P_{\epsilon} = \emptyset$ and $\delta$ contains no $\epsilon$-transitions.
		\item \emph{non-deleting} if the output component of every non-internal transition in $\delta$ is a non-empty string;
		\item \emph{deterministic} (or a DNWT) if for every $q \in Q, p \in P$ and $a \in \Sigma$, it holds that $|\delta(q,\op{a})| = |\delta(q,p,\cl{a})| = 1$;
		\item a \emph{relabelling} transducer if it is $\epsilon$-free and for every $q,q' \in Q, p \in P$, $a \in \Sigma$ and $u \in \Sigma^*$, if $(q',p,u) \in \delta(q, \op{a})$, then $u \in \op{\Sigma}$, and if $(q',u) \in \delta(q,p,\cl{a})$, then $u \in \cl{\Sigma}$;
		\item \emph{functional}, if for every $w \in \wf$, it holds that $|T(w)| = 1$.
	\end{itemize}
\end{definition}

It is easy to see that 
the length of any output of an $\epsilon$-free NWT is at most linear in the length of the input string, while outputs of general NWT may grow to an arbitrary length. 
We note that functionality, unlike the other restrictions defined here, is a \emph{semantic} condition. We do not investigate in this paper the decidability or complexity of determining whether a NWT is functional; likely, techniques for Visibly Pushdown Transducers in \cite{FiliotRRST10} could be adapted for this purpose. Also, while determinism implies functionality, the converse does not hold.

The following lemma shows that we can assume without loss of generality that each transition of a NWT involves at most one input and at most one output symbol, i.e. each opening (closing) transition outputs at most one opening (closing) tag 
and each internal $\epsilon$-transition outputs nothing. 

\begin{lemma}\label{lemma:normalform}
	Each NWT $T = (Q,P,P_\epsilon, \Sigma, \delta, q_0, F)$ can be transformed in polynomial time into an NWT $T' = (Q',P',P'_\epsilon, \Sigma, \delta', q_0, F)$ with $T(w) = T'(w)$ for each $w \in \wf$, such that for any transition in $\delta'$ with output $u$, it holds that $|u| \leq 1$.
	
	We say that a NWT of this shape is in \emph{normal form}.
\end{lemma}

In most of this paper, we restrict our attention to non-deleting transducers. This is because regular nested word languages are closed under transduction by non-deleting NWT, which does not hold in the presence of deletions (consider, for instance, a NWT deleting all matched opening and closing $c$-labelled tags on the regular input language $\{(\op{a}\cl{a}\op{c})^n(\cl{c}\op{b}\cl{b})^n \mid n \geq 0\}$). The practical motivation for desiring this property is the fact that the AXML setting assumes that function call results can be specified by standard XML schema languages, which are subclasses of regular nested word languages.

Moreover, for most of the transducer models examined here, non-deleting transducers are not a significant restriction when it comes to context-free games, as the following result states.

\begin{lemma}\label{lemma:nondeleting}
	Any context-free game $G = (\Sigma, \funcsymb, R, T)$ with NWT $R$ can be transformed in polynomial time into a game $G' = (\Sigma', \funcsymb, R', T')$ such that $R'$ is non-deleting and it holds that $\safelr(G') \cap \wf = \safelr(G)$. 
\end{lemma}

Using Lemma \ref{lemma:normalform}, it is comparatively easy (if tedious) to prove that non-deleting NWTs are closed under composition. This proof, like most proofs for properties of NWT in this section, follows proof ideas used in \cite{FiliotRRST10,RaskinS08} adapted to the specifics of NWT.

\begin{proposition}\label{prop:nwtcomposition}
	Let $T_1$, $T_2$ be non-deleting NWT. Then there exists a non-deleting NWT $T$ such that for all $w \in \wf$, it holds that $T(w) = (T_2 \circ T_1)(w) \mydef T_2(T_1(w))$. This NWT $T$ can be computed from $T_1$ and $T_2$ in polynomial time and is of size $\bigO(|T_1| \cdot |T_2|)$.
\end{proposition}

Since we are solely interested in NWTs operating on well-nested strings, we restrict our attention to NWTs with well-nested domains. The following corollary to Proposition \ref{prop:nwtcomposition} justifies this restriction.

\begin{corollary}\label{cor:wellformeddomain}
	Let $T$ be a non-deleting NWT and $A$ a NWA over alphabet $\Sigma$. Then, there exists a non-deleting NWT $T'$ of size $\bigO(|T| \cdot |A|)$ such that $\Dom(T') = \Dom(T) \cap L(A)$ and $T'(w) = T(w)$ for each $w \in \Dom(T) \cap L(A)$.
\end{corollary}

In order to prove closure of regular nested word languages under transduction by non-deleting NWT, we observe another helpful property of these transducers.

\begin{lemma}\label{lemma:regularrange}
	Let $T$ be a non-deleting NWT with $\Dom(T) \subs \wf$. Then $\Rng(T)$ is a regular language of nested words.
\end{lemma}

\begin{corollary}\label{cor:regulartransduction}
	Regular nested word languages are closed under transduction by non-deleting NWT, i.e. if $L \subs \wf$ is regular and $T$ an NWT, then $T(L)$ is regular. 
\end{corollary}

We now turn to the complexity of standard decision problems for NWT. The upper bounds use relatively simple constructions based on Proposition \ref{prop:nwtcomposition}, while lower bounds follow from comparable results for NWA.

\begin{theorem}\label{thm:nwtmembership}
	The membership problem for non-deleting NWT (Given a non-deleting NWT $T$ and strings $w,u \in \wf$, is $u \in T(w)$?) is in $\PTIME$.
\end{theorem}

\begin{theorem}\label{thm:nwtemptiness}
	The nonemptiness problem for non-deleting NWT (Given a non-deleting NWT $T$, is there a string $w \in \wf$ with $T(w)\neq \emptyset$?) is $\PTIME$-complete with regard to logspace reductions.
\end{theorem}

\begin{theorem}\label{thm:nwttypechecking}
	The type checking problem for non-deleting NWT (Given a non-deleting NWT $T$ and NWA $A_1, A_2$, is $T(L(A_1)) \subs L(A_2)$?) is 
	\begin{enumerate}[(a)]
		\item $\EXPTIME$-complete in general, and
		\item $\PTIME$-complete (w.r.t. logspace reductions) if $A_2$ is a DNWA.
	\end{enumerate}
\end{theorem}

		\section{Games with general NWT replacement}\label{sec:unbounded}

Having laid the foundation with basic results on NWT, we now examine context-free games with NWT-based replacement. The main characteristic distinguishing general NWT from $\epsilon$-free NWT is the fact that, for any input string $w$ and NWT $T$, transducts in $T(w)$ may be arbitrarily large in the size of $w$. This behaviour is necessary if we want to simulate games with regular replacement languages (in the sense of \cite{SchusterS15}) by transducer-based games. As it turns out, however, NWT-based replacement is much more complex than that: the winning problem in games with replay becomes undecidable (as opposed to $\iiEXPTIME$ with regular replacement languages), which may be proven by a rather straightforward reduction from the complement of the halting problem for Turing machines.

\begin{theorem}\label{theo:unboundednotre}
	For the class of games with NWT and \Call depth $k \geq 2$, $\Safelr$ is not recursively enumerable.
\end{theorem}

Even the replay-free winning problem for \pone is quite hard when using NWT for replacement -- we show that this problem is complete for doubly exponential time. The lower bound uses a rather intricate reduction from a two-player tiling problem, while the upper bound is proven by reduction to the purely NWT-based problem of \emph{alternating iterated transduction}, which can be proven to be in $\iiEXPTIME$.

\begin{theorem}\label{thm:unboundediiexptime}
	For the class of replay-free games with NWT, $\Safelr$ is $\iiEXPTIME$-complete.
\end{theorem}

The lower bound proofs for both of these results
require replacement transducers to output nested words of arbitrary depth. Considering our practical motivation, it is rarely required that function calls in Active XML documents return arbitrarily deep trees. Therefore, we now investigate the impact of limiting replacement transducers' output depth.

For simplicity's sake, we assume depth-boundedness as a \emph{semantic} restriction, i.e. we assert that all outputs in $R(w)$ produced by a depth-bounded replacement transducer $R$ on a string $w$ obey a given upper bound on their depth, without examining the decidability and complexity of determining whether or not a given transducer is depth-bounded.

We note that NWT with an output depth linear in the size of the input string are already strictly more expressive than $\epsilon$-free NWT, so the lower bounds from Section \ref{sec:bounded} also hold for NWT with linear output depth. As these lower bounds are already quite high, we focus only on transducers whose output depth is bounded by a constant.

\begin{definition}
	An NWT $R$ is called \emph{depth-bounded} if there is some constant $d \geq 0$ such that for any $w \in \wf$ and any $w' \in R(w)$, the depth of $w'$ is at most $d$.
\end{definition}

Using depth-bounded NWT as replacement transducers places the complexity of the winning problem between those for general NWT and for $\epsilon$-free NWT. The upper and lower bounds are proven similarly to those of Theorem \ref{thm:unboundediiexptime}, but use the fact that the stack size of a depth-bounded NWT on a fixed input is bounded by a constant.

\begin{theorem}\label{thm:depthboundedexpspace}
	For the class of replay-free games with depth-bounded NWT, $\Safelr$ is $\EXPSPACE$-complete.
\end{theorem}

		\section{Games with $\boldsymbol{\epsilon}$-free NWT replacement}\label{sec:bounded}

In this section, we examine context-free games with replacement relations given by $\epsilon$-free NWT. As we shall see,
this leads to a decidable winning problem for games with bounded replay, but non-elementary complexity in all but the easiest case. For the unbounded replay case, we can construct a rather straightforward reduction from the halting problem for TMs.

\begin{theorem}\label{thm:boundedundecidable}
	For the class of games with $\epsilon$-free NWT and unbounded replay, $\Safelr$ is \mbox{undecidable}.
\end{theorem}

Different from games with general NWT, the winning problem for \pone in games with $\epsilon$-free NWT and fixed \Call depth is decidable; however, the complexity of deciding $\Safelr$ is already non-elementary for \Call depth 2.

\begin{theorem}\label{thm:boundednonelementary}
	For the class of games with $\epsilon$-free NWT and \Call depth bounded by $d \geq 2$, $\Safelr$ is decidable, but not decidable in elementary time.
\end{theorem}

Even for replay-free games with $\epsilon$-free NWT, the complexity of deciding the winning problem for \pone is still rather high. The lower bound is proven by reduction from a tiling problem, and the $\coNEXPTIME$ algorithm uses non-determinism to guess moves for \ptwo while trying out all possible strategies for \pone by backtracking.

\begin{theorem} \label{thm:conexptime}
	For the class of replay-free games with $\epsilon$-free NWT, $\Safelr$ is complete for $\coNEXPTIME$.
\end{theorem}

The non-elementary lower bound in Theorem \ref{thm:boundednonelementary} follows from the fact that, in each string returned by \ptwo, \pone may play \Call arbitrarily often. On a return string corresponding to a path of length $n$, \pone may play \Call on all $n$ nodes bottom-up, with each such \Call doubling the number of nodes below the called node, inducing a non-elementary blow-up.

To avoid this, we now examine games with bounded \emph{\Call width}, where, intuitively, \pone may only play \Call for a bounded number of times in each replacement string given by \ptwo. Note that \Call width is counted within each individual replacement string -- so, in a game of \Call depth 3 and \Call width $c$, if \pone plays \Call on some position of the input string, she may then place up to $c$ calls within the string returned by \ptwo, and again up to $c$ calls in \emph{each} of the depth-2 replacement strings resulting from those calls.

More formally, the \Call width of a play $\play$ is the maximum number of times \pone plays \Call in any replacement string given by \ptwo in $\play$. This definition extends naturally into that of \Call width of a strategy.
Note that \Call width only applies to \emph{replacement} strings, so \pone may still call arbitrarily many positions of the \emph{input} string, even for games with \Call width 0. For this reason, replay-free strategies always have bounded \Call width.

The proof of Theorem \ref{thm:boundedundecidable} shows that $\Safelr$ remains undecidable for games with unbounded \Call depth, even with \Call width bounded by $1$. For bounded replay, though, the complexity of $\Safelr$ collapses to that of the replay-free case if $\Call$ width is bounded.

\begin{theorem}\label{thm:boundedcallwidth}
	For the class of games with $\epsilon$-free NWT, \Call depth bounded by $d \geq 1$ and \Call width bounded by $k \geq 1$, $\Safelr$ is $\coNEXPTIME$-complete.
\end{theorem}

As mentioned above, bounded \Call width does not affect \pones options for \Call moves on the input string, as we generally want \pone to be able to at least process \emph{all} function symbols in the input. Dropping this requirement (i.e. bounding \emph{\Call width including input}) yields at least an exponential improvement in complexity.

\begin{theorem}\label{thm:boundedinputcallwidth}
	For the class of games with $\epsilon$-free NWT, \Call depth bounded by $d$ and \Call width including input bounded by $k$, $\Safelr$ is 
	\begin{enumerate}[(a)]
	\item $\coNP$-complete for $d \geq 1$ and $k \geq 2$,
	\item $\coNP$-complete for $d \geq 2$ and $k \geq 1$, and
	\item in $\PTIME$ for $d=k=1$.
	\end{enumerate}
\end{theorem}

The upper bounds in Theorems \ref{thm:boundedcallwidth} and \ref{thm:boundedinputcallwidth} use a backtracking algorithm like the one for Theorem \ref{thm:conexptime}; in this case, however, bounded \Call width reduces the size of both the decision tree for \pone and the occurring replacement strings.

		\section{Games with relabelling replacement}\label{sec:relabel}

As seen before, even the limited amount of insertion allowed by $\epsilon$-free NWT renders the winning problem for \pone quite complex. We now examine how this changes if we disallow insertion entirely.
First, we show that the winning problem is greatly simplified by the fact that transducts of relabelling transducers do not require any additional space beyond that provided by the input. In fact, the upper bounds of Theorems \ref{thm:relabelexptime} to \ref{thm:functionalnp} all use a nigh-trivial (alternating or nondeterministic) algorithm that simply simulates the game. Lower bounds, on the other hand, are proven by reduction from the word problem for linearly bounded (alternating) Turing machines (Theorems \ref{thm:relabelexptime} and \ref{thm:functionalpspace}) and from standard logic-based problems (Theorems \ref{thm:relabelpspace} and \ref{thm:functionalnp}).

\begin{theorem}\label{thm:relabelexptime}
	For the class of games with relabelling transducers and unbounded replay, $\Safelr$ is $\EXPTIME$-complete.
\end{theorem}

With limited or no replay, the complexity decreases even further.

\begin{theorem}\label{thm:relabelpspace}
	For any $k \geq 1$, for the class of games with relabelling transducers and bounded \Call depth $k$, $\Safelr$ is $\PSPACE$-complete.
\end{theorem}

As the winning problem for \pone remains intractable (assuming $\PTIME \neq \PSPACE$) for replay-free games with relabelling transducers, we now turn to the even more limited class of functional relabellings. Note that games with functional transducers are essentially ``solitaire games'' for \pone, as they do not allow for any choice of transducts by \ptwo.

\begin{theorem}\label{thm:functionalpspace}
	For the class of games with functional relabelling transducers and unbounded replay, $\Safelr$ is $\PSPACE$-complete.
\end{theorem}

As for general relabelling transducers, the complexity of $\Safelr$ is the same for games with bounded replay and no replay when restricted to functional relabelling transducers.

\begin{theorem}\label{thm:functionalnp}
	For any $k \geq 1$, for the class of games with functional relabelling transducers and bounded \Call depth $k$, $\Safelr$ is $\NP$-complete.
\end{theorem}

We see that even in this very simple class of games, we still fail to obtain a $\PTIME$ upper bound. Careful examination of lower bound proofs shows that our semantics for replay-free games still allows for a sort of ``hidden replay'': On a string of the form $\op{a}\op{b}v\cl{b}\cl{a}$, if \pone plays \Call first on $\cl{b}$ then on $\cl{a}$, the substring $v$ undergoes \emph{two} transductions -- one from the \Call to $\cl{b}$, another from the \Call to $\cl{a}$. This allows us to perform any number $d$ of transductions on a given string by enclosing it inside $d$ nested function symbols.

Excluding this hidden replay yields a very narrow restriction of context-free games, which we call \emph{write-once} games. In these, no substring may be transduced more than once, i.e. \pone may only play \Call on any closing tag $\cl{a}$ if the substring enclosed in it does not contain a substring on which \pone has played \Call before. 
Note that write-once games are always replay-free, but even weaker as far as \pones rewriting capabilities are concerned.

A slight adaptation of the proof of Theorem \ref{thm:relabelpspace} shows that $\Safelr$ remains $\PSPACE$-hard for write-once games with arbitrary relabelling transducers; for functional (and deterministic) relabelling transducers, however, we can prove tractability. 
The proof constructs from a given game $G$ a NWT $R_\ponea$ such that for each $w \in \wf$, the set of all strings into which \pone way rewrite $w$ in $G$ is given by $R_\ponea(w)$.

\begin{theorem}\label{thm:writeonce}
	For the class of write-once games with functional relabelling transducers, $\Safelr$ is in $\PTIME$.
\end{theorem}

		\section{Conclusion}\label{sec:conclusion}

The research presented in this paper shows that a major challenge in using transducers
for context-free games is finding sensible transducer models and strategy restrictions that do not cause a prohibitive increase in the complexity of the winning problem compared to context-free games without parameter transformation. This paper has made a first step towards identifying what suitable restrictions may look like; however, the few tractable cases identified here are still so restricted that they may be only of limited practical interest.

It is possible that 
the complexity of the winning problem in games with replacement transducers may be further reduced by restricting  relevant schemas to be closer to practical schema specifications for XML (such as DTDs or XML Schema). However, since research in \cite{SchusterS15} indicates that specifications of input schemas for external services influence the complexity of the safe rewriting problem, further research might be necessary to find transducer models whose input and output schemas
can be described by DTDs or XML Schema.

\bibliographystyle{plain}
\bibliography{nwgames}

\counterwithin{theorem}{section}

\renewcommand{\thesection}{\Alph{section}}
\setcounter{section}{0}
\newpage

		\section{Appendix}
For easier reference, we restate the results that were already stated
in the body of the paper. Definitions and results not stated in the body can be
identified by their number of the type A.xxx. 

		\section*{Additional preliminaries}

This section gives additional notations and definitions needed only for proofs in the following sections.

\paragraph*{Nested words}

For the purpose of reductions, we will sometimes need to encode flat strings as nested strings; for a string $w = w_1 \cdots w_n \in \Sigma^*$ with $w_1, \ldots, w_n \in \Sigma$, the \emph{standard nested string encoding $\nw{w}$ of $w$} is $\nw{w} = \op{w_1}\cl{w_1} \ldots \op{w_n}\cl{w_n} \in \wf$.

\paragraph*{Nested word automata}
A \emph{configuration} $\kappa$ of $A$ is a tuple $(q, \alpha) \in Q \times P^*$, with a linear state $q$ and a sequence $\alpha$ of hierarchical states, reflecting the pushdown store. A \emph{run of $A$ on $w=w_1 \ldots w_n \in \wf$} is a sequence $\kappa_0, \ldots, \kappa_n$ of configurations $\kappa_i = (q_i, \alpha_i)$ of $A$ such that for each $i \in [n]$ and $a\in\Sigma$ it holds that $\kappa_i$ is a \emph{successor configuration of $\kappa_{i-1}$ with $w_i$}, i.e. that
\begin{itemize}
	\item if $w_i = \op{a}$, then $(q_{i-1}, \op{a}, q_i, p) \in \delta$ (for some $p \in P$), and $\alpha_i = p \alpha_{i-1}$, or
	\item if $w_i = \cl{a}$, then $(q_{i-1}, p, \cl{a}, q_i) \in \delta$ (for some $p \in Q$), and $p \alpha_i = \alpha_{i-1}$. 
\end{itemize}
	In this case, we also write $\kappa_0 \reach[w]_{A} \kappa_n$.
We say that $A$ \emph{accepts} $w$ if $(q_0, \epsilon) \reach[w]_A (q', \epsilon)$ for some $q' \in F$. The language $L(A) \subs \wf$ is defined as the set of all strings accepted by $A$.

\paragraph*{Context-free games}
Towards a  formal definition of play, a \emph{configuration} is a tuple $\kappa=(p,u,v)\in\{\ponea,\ptwoa\}\times \hat{\Sigma}^* \times \hat{\Sigma}^*$ where $p$ is the player to move, $uv \in \wf$ is the \emph{current word}, and the first symbol of $v$ is \emph{the current position}. A \emph{winning configuration for \pone} is a configuration $\kappa_\ponea=(\ponea,u,\epsilon)$ with $u\in T$, and a \emph{winning configuration for \ptwo} is a configuration $\kappa_\ptwoa=(\ptwoa,u_1\op{a}v,\cl{a}u_2)$, with $\op{a}v\cl{a} \in \rwf$ such that there is no $v' \in \wf$ with $(\op{a}v\cl{a},v') \in R$. 
The configuration $\kappa'=(p',u',v')$ is a \emph{successor configuration} of $\kappa=(p,u,v)$ (Notation: $\kappa\to \kappa'$) if one of the following holds:
 \begin{enumerate}[(1)]
 \item $p'=p=\ponea$, $u' = us$, and $sv' = v$ for some $s\in \hat{\Sigma}$ (\pone plays \Read);
  \item $p=\ponea$, $p'=\ptwoa$, $u = u'$, $v=v' = \cl{a} z$ for $z \in \hat{\Sigma}^*$, $a\in\Gamma$, (\pone plays \Call);
  \item $p=\ptwoa$, $p'=\ponea$, $u=x \op{a} y$, $v = \cl{a} z$ for $x,z \in \hat{\Sigma}^*$, $\op{a}y\cl{a} \in \rwf$, $u'=x$ and $v'=y'z$ for some $y' \in \wf$ with $(\op{a}y\cl{a}, y') \in R$ (\ptwo plays $y'$).
 \end{enumerate}

The \emph{initial configuration} of game $G$ for string $w$ is $\kappa_0(w)\mydef (\ponea,\epsilon,w)$. 
A \emph{play} of $G$ is either an infinite sequence $\play=\kappa_0,\kappa_1,\ldots$ 
or a finite sequence $\play=\kappa_0,\kappa_1,\ldots,\kappa_k$ of
configurations, where, for each $i>0$, $\kappa_{i-1} \to \kappa_{i}$ and, in the finite case, $\kappa_k$ has no successor configuration. In the latter case, \pone
\emph{wins} the play if $\kappa_k$ is a winning position for \pone, in all other cases, \ptwo wins.

\paragraph*{Strategies}

A \emph{strategy} for player $p\in\{\ponea,\ptwoa\}$ maps prefixes $\kappa_0,\kappa_1,\ldots,\kappa_k$ of plays, where $\kappa_k$ is a $p$-configuration, to allowed moves. We denote strategies for \pone by $\Astrat,\Astrat',\Astrat_1,\ldots$ and strategies for \ptwo by $\Bstrat,\Bstrat',\Bstrat_1,\ldots$. 

A strategy $\Astrat$ is \emph{memoryless} if, for every prefix $\kappa_0,\kappa_1,\ldots,\kappa_k$ of a play, the selected move $\Astrat(\kappa_0,\kappa_1,\ldots,\kappa_k)$ only depends on $\kappa_k$. As context-free games are reachability games we only need to consider memoryless strategies; see, e.g., \cite{GraedelTW02}.
\begin{proposition}
  Let $G$ be a context-free game, and $w$ a string. Then either \pone or \ptwo has a winning strategy on $w$, which is actually memoryless.
\end{proposition}

Therefore, in the following, strategies $\Astrat$ for \pone map configurations $\kappa$ to moves $\Astrat(\kappa)\in\{\Call,\Read\}$ and strategies $\Bstrat$ for \ptwo map configurations $\kappa$ to moves $\Bstrat(\kappa)\in \wf$. 

For configurations $\kappa,\kappa'$ and strategies $\Astrat,\Bstrat$  we write $\kappa\movest
\kappa'$ if $\kappa'$ is the unique successor configuration of $\kappa$ determined by strategies $\Astrat$ and $\Bstrat$. Given an initial word $w$ and
strategies  $\Astrat,\Bstrat$ the play\footnote{As the underlying game $G$
will always be clear from the context, our notation does not
mention $G$ explicitly.}
 $\play(\Astrat,\Bstrat,w)\mydef \kappa_0(w)\movest
\kappa_1\movest 
\cdots$ is uniquely determined. If $\play(\Astrat,\Bstrat,w)$ is finite, we denote the word represented by its final configuration by $\word(w,\Astrat,\Bstrat)$.

A strategy $\Astrat$ for \pone is \emph{finite} on string $w$ if the play $\play(\Astrat,\Bstrat,w)$ is 
finite for every strategy $\Bstrat$ of \ptwo. It is a \emph{winning strategy} on
$w$ if \pone wins the play $\play(\Astrat,\Bstrat,w)$, for every $\Bstrat$ of \ptwo. A strategy $\Bstrat$ 
for \ptwo is a \emph{winning strategy} for 
$w$ if \ptwo wins $\play(\Astrat,\Bstrat,w)$, for every strategy $\Astrat$ of
\pone. We only consider finite strategies for \pone, due to \pones  winning condition.
We denote the set of all finite strategies for \pone in the game $G$ by $\strata(G)$, and the set of all strategies for \ptwo by $\stratb(G)$.

\paragraph*{Tilings}

Several of our lower bound proofs use algorithmic problems involving tilings. We will define the individual algorithmic problems in later sections as needed and only give basic common definitions here.

A \emph{tiling of height $m$ and width $\ell$} (or \emph{$m \times \ell$-tiling}) over a \emph{tile set} $U$ with \emph{vertical constraints} $V \subs U \times U$, \emph{horizontal constraints} $H \subs U \times U$, \emph{initial tile} $u_i \in U$ and \emph{final tile} $u_f \in U$ is a mapping $t:[m] \times [\ell] \to U$ such that
\begin{itemize}
\item $(t(i,j),t(i+1,j)) \in V$ for each $i \in [m-1]$ and $j \in [\ell]$,
\item $(t(i,j),t(i,j+1)) \in H$ for each $i \in [m]$ and $j \in [\ell-1]$,
\item $t(1,1) = u_i$, and
\item $t(m,\ell) = u_f$.
\end{itemize}
Intuitively, a tiling arranges $m \cdot \ell$ tiles from $U$ in $m$ rows and $\ell$ columns such that the first row starts with the initial tile, the last row ends with the final tile and horizontally or vertically adjacent tiles match (as per the horizontal and vertical constraints). Accordingly, for a tiling $t$ of width $\ell$ and height $m$, we refer to the string $t(i,1)t(i,2) \cdots t(i,\ell) \in U^*$  as the \emph{$i$-th row} and to $t(1,j)t(2,j) \cdots t(m,j) \in U^*$ as the \emph{$j$-th column} of $t$. 

We will often encode an $m \times \ell$-tiling $t$ as a string $w_t$ of the form $(U^\ell\#)^m$ using a special \emph{line divider symbol} $\# \notin U$, with the interpretation that $t(i,j)$ is the $(i-1) \cdot (\ell+1) + m$-th symbol of $w_t$.
 
		\section*{Proofs and additional results for Section \ref{sec:epsilon}}

The following result is used in several complexity proofs for Section \ref{sec:epsilon}.

\begin{theorem}\label{thm:dnwaptime}
	The emptiness problem for DNWA (given a DNWA $A$, is $L(A) = \emptyset$?) is complete for $\PTIME$ with regard to logspace reductions.
\end{theorem}

\begin{proof}
	The upper bound was proven in \cite{AlurM09}. The lower bound can be proven by a straightforward (if technical) reduction from the emptiness problem for deterministic top-down tree automata (cf. Section 7.2 in \cite{AlurM09} and Theorem 1 in \cite{Veanes97}).
\end{proof}

In the next subsection, we extend nested word automata by (internal) $\epsilon$-transitions. These $\epsilon$-NWA will be of use in the proofs for section \ref{sec:epsilon}.

\subsection*{Nested Word Automata with $\epsilon$-transitions}

It is well known (see, for instance, \cite{HopcroftMU01}) that extending finite-state automata with $\epsilon$-transitions does not change their expressive power; (nondeterministic) finite-state automata with $\epsilon$-transitions still decide exactly the class of regular languages of flat strings.

Even though nested word automata and the class of regular nested word languages strongly parallel finite-state automata and flat regular languages, a similar investigation has so far not been performed for nested word automata. We now define and examine nested word automata with $\epsilon$-transitions, mainly as a tool for the analysis of nested word transducers with $\epsilon$-transitions.

\begin{definition}
	A \emph{Nested Word Automaton with $\epsilon$-transitions ($\epsilon$-NWA)}  $A = (Q,P,\Sigma,\delta,q_0,F)$ consists of
	\begin{itemize}
		\item a set $Q$ of \emph{linear states},
		\item a set $P$ of \emph{hierarchical states},
		\item an alphabet $\Sigma$,
		\item a \emph{transition relation} $\delta \subs (Q \times \op{\Sigma} \times Q \times P) \cup (Q \times P \times \cl{\Sigma} \times Q) \cup (Q \times \{\epsilon\} \times Q)$, 
		\item an \emph{initial state} $q_0 \in Q$, and
		\item a set of \emph{accepting states} $F \subs Q$.
	\end{itemize}
	As for standard NWA, we also write $(q',p) \in \delta(q, \op{a})$ (resp. $q' \in \delta(q,p,\cl{a})$, $q' \in \delta(q,\epsilon)$) instead of $(q,\op{a},q',p) \in \delta$ (resp. $(q,p,\cl{a},q'), (q,\epsilon,q') \in \delta$).
\end{definition}

Note that $\epsilon$-transitions are always \emph{internal} transitions that merely change the current linear state, not the hierarchical stack; allowing $\epsilon$-transitions to manipulate the stack as well would yield a strictly more expressive automaton model.

The semantics of $\epsilon$-NWA is defined almost exactly like that of NWA, by way of (accepting) runs. The only difference is that for a run $\kappa_0, \ldots, \kappa_n$ of an $\epsilon$-NWA $A = (Q,P,\Sigma,\delta,q_0,F)$ on a nested string $w \in \wf$, it merely holds that $n \geq |w|$, and each configuration $\kappa_i$ is either a successor configuration of $\kappa_{i-1}$ with the next unread symbol of $w$ (as defined in Section \ref{sec:prelim}), or a successor configuration with $\epsilon$, i.e. $\kappa_{i-1} = (q,\alpha)$, $\kappa_i = (q', \alpha)$ and $(q,\epsilon,q') \in \delta$.

The following properties of $\epsilon$-NWA follow easily from the proofs for corresponding properties of NWA in \cite{AlurM09}.

\begin{lemma}\label{lemma:epsnwadeterminisation}
	For each $\epsilon$-NWA $A$, there exists a DNWA $A'$ of size at most exponential in $|A|$ such that $L(A) = L(A')$.
\end{lemma}

\begin{proof}
	(Sketch) The proof of this statement for NWA without $\epsilon$-transitions in Theorem 3.3 of \cite{AlurM09} uses a modified subset construction, where states of the exponential-sized DNWA $A'$ correspond to sets of pairs of states of $A$ (so-called \emph{summaries}) such that if a (not necessarily well-nested) string $w$ induces a (partial) run from the initial state of $A'$ to some summary state $S \in \Pot(Q \times Q)$, then there are (partial) runs with $w$ from each $q \in Q$ to all $q' \in \{q' \mid (q,q') \in S\}$ in $A'$.
	
	To account for $\epsilon$-transitions, we modify these summaries to include \emph{$\epsilon$-closures} of target states, i.e. for each pair $(q,q')$ of states contained in a summary $S$ as constructed in \cite{AlurM09}, we add to $S$ all pairs $(q, q'')$, where $q''$ is reachable from $q$ by a series of $\epsilon$-transitions in $A$. Otherwise, the construction (and correctness proof) is the same as in \cite{AlurM09}.
\end{proof}

\begin{lemma}\label{lemma:epsnwabool}
	For all $\epsilon$-NWA $A_1$ and $A_2$, it is possible to construct in polynomial time $\epsilon$-NWA deciding $L(A_1) \cup L(A_2)$ and \mbox{$L(A_1) \cap L(A_2)$}.
\end{lemma}

\begin{proof}
	(Sketch) This proof, like the one for Theorem 3.5 in \cite{AlurM09}, uses a standard product construction simulating $A_1$ and $A_2$ simultaneously on the input. The product automaton $A' = (Q_1 \times Q_2,P_1 \times P_2,\Sigma,\delta',(q_{0,1},q_{0,2}),F')$ is constructed as in \cite{AlurM09} and simply extended by $\epsilon$-transitions. Note that, unlike reading transitions, $\epsilon$-transitions do \emph{not} have to be synchronised between the two automata, i.e. an $\epsilon$-transition of $A'$ simulates an $\epsilon$-transition of only \emph{one} of the component automata $A_1$ or $A_2$. The transition relation $\delta'$ of $A'$ is therefore extended by the sets $\{((q_1, q_2),\epsilon, (q'_1, q_2)) \mid (q_1, \eps, q'_1) \in \delta_1, q_2 \in Q_2\}$ and $\{((q_1, q_2),\epsilon, (q_1, q'_2)) \mid (q_2, \eps, q'_2) \in \delta_2, q_1 \in Q_1\}$.
\end{proof}

\begin{theorem}\label{thm:epsnwacomplexity}

	\begin{enumerate}[(a)]
	\item The membership and emptiness problem for $\epsilon$-NWA are in $\PTIME$.
	\item The universality, equivalence and inclusion problem for $\epsilon$-NWA are $\EXPTIME$-complete.
	\item Deciding, given an $\epsilon$-NWA $A$ and a DNWA $B$, whether $L(A) \subs L(B)$ is $\PTIME$-complete with respect to logspace reductions.
	\end{enumerate}
\end{theorem}

\begin{proof}
	These complexity properties mostly follow from the corresponding results for NWA without $\epsilon$-transitions, proven in \cite{AlurM09} (Proposition 6.1 and Theorem 6.2). 
	
	For (a), $\epsilon$-NWA may also be interpreted as pushdown automata, whose membership and emptiness problem are decidable in polynomial time.
	
	Lower bounds for (b) directly follow from the corresponding lower bounds for NWA without $\epsilon$-transitions. Corresponding upper bounds can be proven using Lemmas \ref{lemma:epsnwadeterminisation} and \ref{lemma:epsnwabool} as well as the complement construction for DNWA (\cite{AlurM09}, Theorem 3.5) and set-theoretic formulations for universality (is $\overline{L(A)} = \emptyset$?), equivalence (is $(\overline{L(A_1)} \cap L(A_2)) \cup (L(A_1) \cap \overline{L(A_2)}) = \emptyset$?) and inclusion (is $L(A_1) \cap \overline{L(A_2)} = \emptyset$?).
	
	The upper bound in (c) also follows from the fact that DNWA can be efficiently complemented and the fact that $L(A) \subs L(B)$ holds if and only if $L(A) \cap \overline{L(B)} = \emptyset$. By Lemma \ref{lemma:epsnwabool} and part (a) of this theorem, this can be checked in polynomial time.
	
	The lower bound in (c) is proven by reduction from the emptiness problem for DNWA (cf. Theorem \ref{thm:dnwaptime}). Let $B$ be a DNWA to be checked for emptiness. We can construct in logarithmic space an $\epsilon$-NWA $A$ deciding $\wf$, and (since $B$ is deterministic) a DNWA $B'$ deciding the complement of $L(B)$. It then holds that $L(A) \subs L(B')$ if and only if $L(B') = \wf$, which is the case if and only if $L(B)$ is empty.
\end{proof}

\subsection*{Proofs for Section \ref{sec:epsilon}}

We begin by defining the formal semantics of NWT.

\begin{definition}
	Let $T = (Q,P,P_{\epsilon},\Sigma,\delta,q_0,F)$ be an NWT. A \emph{configuration} $\kappa = (q,\alpha)$ consists of a linear state $q \in Q$ and a stack $\alpha \in P^*$ of hierarchical states. 

	For any string $w \in \wf$, 	an \emph{$\epsilon$-extension of $w$} is a string $\hat{w}$ obtained by inserting symbols $\op{\epsilon}, \cl{\epsilon}$ and $\epsilon$ into $w$ such that the maximal subsequence of $\hat{w}$ consisting only of symbols from $\hat{\Sigma} \cup \{\op{\epsilon},\cl{\epsilon}\}$ is a well-nested word over $\Sigma \cup \{\epsilon\}$.
	
	A \emph{run} of $T$ on an $\epsilon$-extension $\hat{w} = \hat{w}_1 \ldots \hat{w}_n$ of a string $w \in \wf$ starting at configuration $(r_0, \alpha_0)$ is a finite sequence $\rho = (r_0, \alpha_0 )(r_1,\alpha_1) \ldots (r_n, \alpha_n)$ such that for each $i \in [n]$, one of the following holds:
	\begin{itemize}
		\item $\hat{w}_i = \op{a}$ for some $a \in \Sigma \cup \{\epsilon\}$, $(r_{i},p_i,u_i) \in \delta(r_{i-1},\op{a})$ and $\alpha_i = \alpha_{i-1} p$,
		\item $\hat{w}_i = \epsilon$, $(r_i,u_i) \in \delta(r_{i-1},\epsilon)$ and $\alpha_i = \alpha_{i-1}$, or
		\item $\hat{w}_i = \cl{a}$ for some $a \in \Sigma \cup \{\epsilon\}$, $(r_{i},u_i) \in \delta(r_{i-1},p,\cl{a})$ and $\alpha_i p = \alpha_{i-1}$.
	\end{itemize}
	The run $\rho$ is \emph{accepting} if $r_0=q_0$, $r_n \in F$, and $\alpha_0 = \alpha_n = \epsilon$; in this case, the string $u_1 u_2 \ldots u_n$ is considered the \emph{output} of $T$ on $w$ according to $\rho$.\footnote{Note that the $\epsilon$-extension $\hat{w}$ on which the output $u_1 u_2 \ldots u_n$ is produced is already implicit in the run $\rho$, so we do not specify it explicitly.}
	
\end{definition}

Note that the semantics for NWT directly carry over to $\epsilon$-free NWT, with the only $\epsilon$-extension of any $w \in \Dom(T)$ on which there are accepting runs of $T$ being $w$ itself.

\begin{restate}{Lemma \ref{lemma:normalform}}
	Each NWT $T = (Q,P,P_\epsilon, \Sigma, \delta, q_0, F)$ can be transformed in polynomial time into an NWT $T' = (Q',P',P'_\epsilon, \Sigma, \delta', q_0, F)$ with $T(w) = T'(w)$ for each $w \in \wf$, such that for any transition in $\delta'$ with output $u$, it holds that $|u| \leq 1$.
	
	We say that a NWT of this shape is in \emph{normal form}.
\end{restate}

\begin{proof}
	An arbitrary NWT $T$ is transformed into an NWT $T'$ in normal form by successively replacing each transition that is not of the required form by a sequence of new states and transitions. We only describe this procedure for an opening transition; closing and internal transitions can be handled in a similar manner.
	
	Assume for some $q \in Q$ and $a \in \Sigma$ that $(q',p',v) \in \delta(q, \op{a})$ with $v = v_1 \cdots v_n \in \hat{\Sigma}^*$ for some $n>1$. We add new states $q_1, \ldots, q_{n-1}$ to $Q$ and $p_{p',a}$ to $P_\epsilon$. Let $k \leq n$ be the position of the last unmatched opening tag in $v$, i.e. $v_1 \cdots v_{k-1} \in \hat{\Sigma}^*$, $v_k \in \op{\Sigma}$ and $v_{k+1} \cdots v_n \in \wf$. We add a transition $(q_1,p_{p',a},v_1)$ to $\delta(q,\op{\epsilon})$\footnote{Note that $v_1\in \op{\Sigma}$ due to well-formedness.}, a transition $(q_{k+1}, p', v_k)$ to $\delta(q_k,\op{a})$ and, for each $i \in [n-1]$ with $i \neq k$, a transition $(q_{i+1}, p_{p',a}, v_i)$ to $\delta(q_i, \op{\epsilon})$ if $v_i \in \op{\Sigma}$, or $(q_{i+1},v_i)$ to $\delta(q_i, p_{p',a}, \cl{\epsilon})$ if $v_i \in \cl{\Sigma}$, identifying $q_{n}$ with $q'$. Finally, we remove the original transition. This takes care of (reading or $\epsilon$-)transitions producing more than one output symbol.
		
	The resulting NWT $T'$ is obviously in normal form. Its equivalence to $T$ is relatively simple (if tedious) to prove by an induction over the structure of input strings, with the main argument using the fact that $T$ fulfils the $\epsilon$-consistency and well-formedness condition. Notably, these conditions also justify the above simplification that newly added $\epsilon$ transitions obtained from a transition reading $a$ and pushing the hierarchical state $p'$ use solely the new hierarchical $\epsilon$-state $p_{p',a}$.
\end{proof}

\begin{restate}{Lemma \ref{lemma:nondeleting}}
	Any context-free game $G = (\Sigma, \funcsymb, R, T)$ with NWT $R$ can be transformed in polynomial time into a game $G' = (\Sigma', \funcsymb, R', T')$ such that $R'$ is non-deleting and it holds that $\safelr(G') \cap \wf = \safelr(G)$. 
\end{restate}

\begin{proof}(Sketch)
	The idea behind this proof is modifying $R$ into $R'$ in such a way that, whenever $R$ would delete some tag $\op{a}$ (or $\cl{a}$), $R'$ instead replaces tag by a special ``strike-out'' version $\op{\text{\sout{a}}}$ (or $\cl{\text{\sout{a}}}$) with $\text{\sout{a}} \notin \Sigma$, instead; we therefore set $\Sigma' = \Sigma \uplus \{\text{\sout{a}} \mid a \in \Sigma \}$. To ensure that iterated transductions respect deleted tags, we add transitions to each state of $R'$ that only replace ``strike-out'' tags by themselves without changing the state of $R'$. Finally, we similarly modify the target DNWA for $T$ in order to ignore all tags with labels not in $\Sigma$ (i.e. only check them for their nesting structure without changing states). Since \pone may only play \Call on symbols from  $\Sigma$, \ptwo is unable to rewrite ``strike-out'' symbols, and the DNWA for $T'$ ignores symbols outside of $\Sigma$, it is clear that \pone has a winning strategy on any string $w \in \wf$ in $G'$ if and only if she has a winning strategy on $w$ in $G$.
\end{proof}

\begin{restate}{Proposition \ref{prop:nwtcomposition}}
	Let $T_1$, $T_2$ be non-deleting NWT. Then, there exists a non-deleting NWT $T$ such that for all $w \in \wf$, it holds that $T(w) = (T_2 \circ T_1)(w) \mydef T_2(T_1(w))$. This NWT $T$ can be computed from $T_1$ and $T_2$ in polynomial time and is of size $\bigO(|T_1| \cdot |T_2|)$.
\end{restate}

\begin{proof}
	The basic idea behind this construction is simple: The transducer $T$ simulates $T_1$ on its input and directly feeds the output of $T_1$ into $T_2$. This is done by a sort of product construction which is, for the most part, quite straightforward. However, there are some subtleties that need to be addressed, which stem from the fact that both $T_1$ and $T_2$ are NWT, so the various possibilities how one of the two transducers might perform a step while the other is idle have to be dealt with. We note that this proof could be extended to general (not necessarily non-deleting) NWT, but since the proof details are quite technical already, we restrict our attention to non-deleting NWT here.
	
	Let $T_1 = (Q^1,P^1,P_{\epsilon}^1,\Sigma,\delta^1,q^1_0,F^1)$ and $T_1 = (Q^2,P^2,P_{\epsilon}^2,\Sigma,\delta^2,q^2_0,F^2)$ be two NWT in normal form. We construct from $T_1$ and $T_2$ the NWT $T = (Q,P,P_{\epsilon},\Sigma,\delta,q_0,F)$ as follows:
	\begin{itemize}
	\item $Q = Q^1 \times Q^2$,
	\item $P = (P^1 \cup \{\Null\}) \times P^2 $ for a special symbol $\Null$ not used in $T_1$ or $T_2$,
	\item $P_\epsilon = (P^1_\epsilon \times P^2) \cup (\{\Null\} \times P^2_\epsilon)$
	\item $q_0 = (q_0^1, q_0^2)$,
	\item $F  = F^1 \times F^2$, and
	\item $\delta$ is constructed as detailed below.
	\end{itemize}
	
	The construction of $\delta$ is quite straightforward in the case where $T_1$ produces an output on which $T_2$ is simulated: If $(q_1', p_1', \op{b}) \in \delta^1(q_1, \op{a})$ (or $(q_1', p_1', \op{b}) \in \delta^1(q_1, \op{\epsilon})$, respectively) and $(q_2', p_2', \op{c}) \in \delta^2(q_2, \op{b})$, then $((q_1', q_2'), (p_1', p_2'), \op{c}) \in \delta((q_1,q_2), \op{a})$ (or $((q_1', q_2'), (p_1', p_2'), \op{c}) \in \delta((q_1,q_2), \op{\epsilon})$, respectively), and analogously for the corresponding closing ($\epsilon$-) transitions.
	
	Internal $\epsilon$-transitions for $T_1$ and $T_2$ are also easily handled, as for these transitions, neither transducer consumes an input or produces an output. If $(q_1', \epsilon) \in \delta^1(q_1, \epsilon)$, then $((q_1', q_2), \epsilon) \in \delta((q_1,q_2), \epsilon)$ for each $q_2 \in Q^2$, and if $(q_2', \epsilon) \in \delta^2(q_2, \epsilon)$, then $((q_1, q'_2), \epsilon) \in \delta((q_1,q_2), \epsilon)$ for each $q_1 \in Q^1$.
	
	The only case requiring special attention is the one where only one of the two transducers modifies the input. This happens when 
	$T_2$ produces an output symbol by an $\epsilon$-transition, without $T_1$ consuming an input symbol. In this case, $T_2$ produces a hierarchical state while $T_1$ doesn't; to this end, some states in the set $P_\epsilon$ contain a component $\Null$ to indicate a ``null transition'' for $T_1$. 
	
	More formally, 
	if for some $q_2, q'_2 \in Q^2$, $p'_2 \in P^2$ and $a \in \Sigma$, $(q'_2, p'_2, \op{a}) \in \delta^2(q_2, \op{\epsilon})$, then $((q_1, q'_2), (\Null^1, p'_2), \op{a}) \in \delta((q_1,q_2), \op{\epsilon})$ for all $q_1 \in Q^1$ (and analogously for closing transitions).
	
	It follows directly from the construction that $T$ is $\epsilon$-consistent; the well-formedness condition for $T$ follows from well-formedness of $T_1$ and $T_2$ by a simple but lengthy case distinction over all the sorts of transitions introduced here. Furthermore, since $T$ is obviously in normal form, it automatically fulfils synchronisation. It remains to be proven that $T(w) = T_2(T_1(w))$ indeed holds for all $w \in \wf$.
	
	To this end, let $w_2 \in T_2(T_1(w))$; let further $\rho_1$ be an accepting run of $T_1$ with output $w_1$ on an $\epsilon$-extension $\hat{w}$ of $w$, and let $\rho_2$ be an accepting run of $T_2$ with output $w_2$ on an $\epsilon$-extension $\hat{w}_1$ of $w_1$. We construct an $\epsilon$-extension $\hat{w}'$ of $w$ and an accepting run $\rho$ of $T$ on $\hat{w}'$ with output $w_2$.
	
	We denote all positions of $\hat{w}$ in which $T_1$ outputs some symbol as \emph{1-producing}. 	Note that there is a bijective correspondence between positions of $w_1$ and 1-producing positions of $\hat{w}$ and that all $\op{\epsilon}$- and $\cl{\epsilon}$-positions of $\hat{w}$ are 1-producing. Next, we examine all $\op{\epsilon}$- and $\cl{\epsilon}$-positions of $\hat{w_1}$; these, we call \emph{2-producing}. Due to the well-formedness and $\epsilon$-consistency restrictions on $T_1$ and $T_2$, it is possible to insert all $\op{\epsilon}$-, $\cl{\epsilon}$- and $\epsilon$-positions of $\hat{w}_1$ into $\hat{w}$ in such a way that we obtain an $\epsilon$-extension $\hat{w}'$ of $w$ that has both $\hat{w}$ and $\hat{w}_1$ as subsequences. 
	
	It is now easy to see that an accepting run $\rho$ of $T$ on $\hat{w}'$ with output $w_2$ can be obtained by combining the transitions used in $\rho_1$ and $\rho_2$ -- positions in $\hat{w}'$ that are 1-producing correspond to ``standard'' transitions of $T$, $\epsilon$-positions correspond to internal $\epsilon$-transitions, 
	and 2-producing positions correspond to $\epsilon$-transitions with a $\Null$ component in their hierarchical state. This shows that $T_2(T_1(w)) \subseteq T(w)$.
	
	For the other direction, let $\rho$ be a run of $T$ on an $\epsilon$-extension $\hat{w}'$ of $w$ with output $w_2$. We label the positions of $\hat{w}'$ according to the transitions taken by $T$ in $\rho$ -- positions where ``standard'' transitions are used are labelled as 1-producing, $\epsilon$-positions are labelled as internal,
	and positions with transitions whose hierarchical stack contains a $\Null$ component are labelled as 2-producing. Similar to the previous part of the proof, we can then use these labels to separate $\rho$ into an accepting run $\rho_1$ of $T_1$ on an $\epsilon$-extension $\hat{w}$ of $w$ with output $w_1$ and an accepting run $\rho_2$ of $T_2$ on an $\epsilon$-extension $\hat{w}_1$ of $w_1$ with output $w_2$, thus proving that $w_2 \in T_2(T_1(w))$ holds.
\end{proof}

\begin{restate}{Corollary \ref{cor:wellformeddomain}}
	Let $T$ be a non-deleting NWT and $A$ a NWA over alphabet $\Sigma$. Then, there exists a non-deleting NWT $T'$ of size $\bigO(|T| \cdot |A|)$ such that $\Dom(T') = \Dom(T) \cap L(A)$ and $T'(w) = T(w)$ for each $w \in \Dom(T) \cap L(A)$.
\end{restate}

\begin{proof}
	Let $T_{A}$ be a NWT with $\Dom(T_{A}) = L(A)$ and $T_{A}(w) = \{w\}$ for each $w \in \wf$, i.e. $T_{A}$ accepts exactly the strings in $L(A)$ and simply outputs its input string. Such a NWT is easy to construct with a size in $\bigO(|A|)$. 
	
	Then, we set $T'$ to be the NWT for $T \circ T_{A}$ as constructed in Proposition \ref{prop:nwtcomposition}. For each $w \in L(A) \cap \Dom(T)$, it holds that $T'(w) = T(T_{A}(w)) = T(w)$, for each $w \notin L(A)$, $T'(w) = T(\emptyset) = \emptyset$, and for each $w \in L(A) \setminus \Dom(T)$, it holds that $T'(w) = T(w) = \emptyset$. This proves that $\Dom(T') = \Dom(T) \cap L(A)$ and $T'(w) = T(w)$ for each $w \in \Dom(T) \cap L(A)$. The desired bound on the size of $T'$ follows directly from Proposition \ref{prop:nwtcomposition}.
\end{proof}

\begin{restate}{Lemma \ref{lemma:regularrange}}
	Let $T$ be a non-deleting NWT with $\Dom(T) \subs \wf$. Then $\Rng(T)$ is a regular language of nested words.
\end{restate}

\begin{proof}
	Let $T = (Q,P,P_\epsilon, \Sigma, \delta, q_0, F)$ be a non-deleting NWT in normal form. The basic idea behind constructing the $\epsilon$-NWA $A$ for $\Rng(T)$ is taking the input string for $A$ and verifying it against the output component of $T$. This way, it is easy to ensure that every string in $\Rng(T)$ is accepted by $A$; some care has to be taken, however, to make certain that for every string $w'$ accepted by $A$, there is a string $w \in \wf$ such that $w' \in T(w)$.
	
	Since $T$ is in normal form, all opening transitions in $\delta$ are of the form $(q,\op{a},q',p,\op{b})$ and all closing transitions of the form $(q,p,\cl{a},q',\cl{b})$ (for $a \in \Sigma \cup \{\epsilon\}$ and $b \in \Sigma$) while internal transitions are of the form $(q,\epsilon,q',\epsilon)$.
	
	The linear state set of $A$ is $Q$ and its set of hierarchical states is $P \times (\Sigma \cup \{\epsilon\})$; its starting and accepting states are those of $T$. Reading transitions in $A$ are constructed from those of $T$ by taking, for each opening transition $(q,\op{a},q',p,\op{b}) \in \delta$, a transition $(q,\op{b},q',(p,a))$, and for each closing transition $(q,p,\cl{a},q',\cl{b}) \in \delta$, a transition $(q,(p,a),\cl{b},q')$. Each internal transition $(q,\epsilon,q',\epsilon) \in \delta$ translates to an $\epsilon$-transition $(q,\epsilon,q')$ in $A$.
	
	To prove correctness of this construction, we need to show that for any string $w' \in \wf$ it holds that $w' \in L(A)$ if and only if $w' \in \Rng(T)$.
	
	For the ``if'' direction, if $w' \in \Rng(T)$, then there is some $w \in \wf$ such that $T$ has an accepting run $\rho$ on some $\epsilon$-extension of $w$ with output $w'$. Translating the transitions of $T$ taken in $\rho$ into transitions of $A$ as per the above construction naturally yields an accepting run of $A$ on $w'$.
	
	For the ``only if'' direction, assume that $w' \in L(A)$ for some $w' \in \wf$. Then, there is an accepting run $\rho$ of $A$ on $w'$. We can construct from $\rho$ an $\epsilon$-extension $\hat{w}$ of a string $w \in \wf$ by extracting the sequence of hierarchical components of transitions in $\rho$ -- every opening (closing) transition in $\rho$ with hierarchical component $(p,a)$ with some $p \in P$, $a \in \Sigma \cup \{\epsilon\}$ corresponds to a symbol $\op{a}$ ($\cl{a}$) in $\hat{w}$, and every $\epsilon$-transition taken in $\rho$ corresponds to a symbol $\epsilon$ in $\hat{w}$. It is clear to see that indeed $\hat{w}$ is an $\epsilon$-extension of some string $w \in \wf$ and that there is an accepting run of $T$ on $\hat{w}$ that outputs $w'$, which yields $w' \in \Rng(T)$ as was to be proven.
\end{proof}

\begin{restate}{Corollary \ref{cor:regulartransduction}}
	Regular nested word languages are closed under transduction by non-deleting NWT, i.e. if $L \subs \wf$ is regular and $T$ a NWT, then $T(L)$ is regular. 
\end{restate}

\begin{proof}
	Let $A$ be a NWA for $L$ and $T$ a NWT. From $A$, we can construct by Corollary \ref{cor:wellformeddomain} a NWT $T_A$ with $\Dom(T_A) = L \cap \Dom(T)$ and $T_A(w) = T(w)$ for each $w \in \Dom(T_A)$, which implies that $\Rng(T_A) = T(L)$. By Lemma \ref{lemma:regularrange}, this implies that $T(L)$ is regular as well.
\end{proof}

\begin{restate}{Theorem \ref{thm:nwtmembership}}
	The membership problem for non-deleting NWT (Given a non-deleting NWT $T$ and strings $w,u \in \wf$, is $u \in T(w)$?) is in $\PTIME$.
\end{restate}

\begin{proof}
	From $w$, we can easily compute a NWA $A$ of size $\bigO(w)$ with $L(A) = \{w\}$. By Corollary \ref{cor:wellformeddomain}, we can compute from $A$ and $T$ in polynomial time a non-deleting NWT $T_w$ with $\Dom(T_w) = \Dom(T) \cap \{w\}$ and $\Rng(T_w) = T(w)$. By Corollary \ref{cor:regulartransduction}, $T(w)$ is regular, and by Lemma \ref{lemma:regularrange} an $\epsilon$-NWA $A'$ for $T(w)$ can be computed in polynomial time. This $\epsilon$-NWA is of polynomial size, and checking $u$ for membership in $L(A)$ is possible in polynomial time by Theorem \ref{thm:epsnwacomplexity} (a).
\end{proof}

\begin{restate}{Theorem \ref{thm:nwtemptiness}}
	The nonemptiness problem for non-deleting NWT (Given a non-deleting NWT $T$, is there a string $w \in \wf$ with $T(w)\neq \emptyset$?) is $\PTIME$-complete with regard to logspace reductions.
\end{restate}

\begin{proof}
	By definition, it holds that there is some $w \in \wf$ with $T(w) \neq \emptyset$ if and only if $\Dom(T) \cap \wf \neq \emptyset$. By Corollary \ref{cor:wellformeddomain}, we can compute from $T$ a NWT $T'$ whose domain is exactly $\Dom(T) \cap \wf$. Clearly, it holds that $\Dom(T') \neq \emptyset$ if and only if $\Rng(T') \neq \emptyset$, so by Lemma \ref{lemma:regularrange} we can extract from $T'$ an $\epsilon$-NWA $A$ of polynomial size with $L(A) = \Rng(T')$. All of these transformations, as well as testing whether $L(A) \neq \emptyset$, are feasible in polynomial time. This proves the upper bound.
	
	The lower bound follows by reduction from the emptiness problem for DNWA (cf. Theorem \ref{thm:dnwaptime}). Let $A = (Q,P, \Sigma, \delta, q_0, F)$ be a DNWA to be tested for non-emptiness. We construct from $A$ a NWT $T$ with $T(\epsilon) = L(A)$ and $T(w) = \emptyset$ for all $w \neq \epsilon$. The idea behind the reduction is replacing, for each $a \in \Sigma$, each $\op{a}$ (resp. $\cl{a}$) transition of $A$ by an $\op{\epsilon}$ (resp. $\cl{\epsilon}$) transition in $T$ that writes $\op{a}$ (resp. $\cl{a}$) as an output. Formally, we set $T = (Q, P \times \Sigma, P \times \Sigma, \Sigma, \delta', q_0, F)$, where $\delta'$ contains a transition $(q,\op{\epsilon}, q', (p,a), \op{a})$ (resp. $(q, (p,a), \cl{\epsilon}, q', \cl{a})$) if and only if $\delta$ contains a transition $(q, \op{a}, q', p)$ (resp. $(q, p, \cl{a}, q')$). This construction is feasible using logarithmic space, and it holds that $T(\epsilon) = L(A)$ and $T(w) = \emptyset$ for all $w \neq \epsilon$, so $\Dom(T) \neq \emptyset$ if and only if $L(A) \neq \emptyset$.
\end{proof}

\begin{restate}{Theorem \ref{thm:nwttypechecking}}
	The type checking problem for non-deleting NWT (Given a non-deleting NWT $T$ and NWA $A_1, A_2$, is $T(L(A_1)) \subs L(A_2)$?) is 
	\begin{enumerate}[(a)]
		\item $\EXPTIME$-complete in general, and
		\item $\PTIME$-complete (w.r.t. logspace reductions) if $A_2$ is a DNWA.
	\end{enumerate}
\end{restate}

\begin{proof}
	The lower bound in (a) follows directly by a reduction from the inclusion problem for NWA, which is known to be $\EXPTIME$-complete \cite{AlurM09}. Let $B_1$ and $B_2$ be NWA to be checked for inclusion of $L(B_1)$ in $L(B_2)$. We can construct in polynomial time a NWT $T$ that simply reproduces each input symbol in the output and accepts any input (i.e. $T(w) = \{w\}$ for any $w \in \hat{\Sigma}^*$). Then, clearly, $L(B_1) \subseteq L(B_2)$ if and only if $T(L(B_1)) \subseteq L(B_2)$, as $T(L(B_1)) = L(B_1)$.
	
	The lower bound in (b) follows from $\PTIME$-hardness of the emptiness problem for DNWA by a similar reduction to that used in the proof of Theorem \ref{thm:epsnwacomplexity} (c). Let $A$ be a DNWA to be tested for emptiness. We can construct in logarithmic space an $\epsilon$-NWA $A_1$ deciding $\wf$, and (since $A$ is deterministic) a DNWA $A_2$ deciding the complement of $L(A)$. For $T$, we construct a NWT reproducing its input, i.e. with $\Dom(T) = \wf$ and $T(w) = \{w\}$ for each $w \in \wf$. It then holds that $T(L(A_1)) \subs L(A_2)$ if and only if $L(A_2) = \wf$, which is the case if and only if $L(A)$ is empty.
	
	The upper bounds in both (a) and (b) can easily be proven using prior results from this section. By Corollary \ref{cor:wellformeddomain} and Lemma \ref{lemma:regularrange}, we can construct in polynomial time an $\epsilon$-NWA deciding $T(L(A_1))$; checking this $\epsilon$-NWA for inclusion in $L(A_2)$ is generally possible in exponential time, and in polynomial time if $A_2$ is a DNWA by Theorem \ref{thm:epsnwacomplexity} (b) and (c).
\end{proof}
 \newpage
		\section*{Proofs for Section \ref{sec:unbounded}}

\begin{restate}{Theorem \ref{theo:unboundednotre}}
	For the class of games with NWT and \Call depth $k \geq 2$, $\Safelr$ is not recursively enumerable.
\end{restate}

\begin{proof}
	The proof is by reduction from the complement of the halting problem for Turing machines with an empty input. From a given Turing machine $M$ with working alphabet $\Sigma$ and state set $Q$, we construct a game $G$ and string $w$ such that \pone has a winning strategy on $w$ in $G$ if and only if $M$ does not halt on an empty input. The game $G$ uses two function symbols, $s$ and $t$, which are not in $\Sigma$ in order to avoid plays according to $G$ interfering with the workings of $M$.
	
	The idea behind the reduction is rather simple. The game begins on the input string $\op{s}\cl{s}$, where \pone is supposed to \Call $\cl{s}$ as her first move. For his reply, \ptwo then picks a number $r$ such that $M$ halts on the empty input after exactly $r$ steps (if such a number exists) and returns $\op{t}^r \nw{v}_0 \cl{t}^r$, with $\nw{v}_0$ being a string representing the initial configuration of $M$ on the empty input. \pone should then call each $\cl{t}$ in left-to-right order, making \ptwo simulate a step of $M$ on the current position of $M$ encoded in the current string. After \pone has called all $r$ closing $\cl{t}$ tags and thus rewritten $\nw{v}_0$ into a string $\nw{v}_r$ representing the configuration of $M$ after $r$ computation steps, \pone wins the game if $\nw{v}_r$ does not represent a halting configuration.

	We can represent configurations of $M$ as flat strings over the alphabet $\Sigma \times (Q \cup \{-\})$ in the standard fashion -- a string $(x_1,-) \cdots (x_{k-i},-) (x_k,q) (x_{k+1},-)\cdots (x_m,-)$ denotes that the content of $M$'s working tape is $x_1\cdots x_m$, with the head of $M$ being on the tape's $k$-th cell and $M$ being in state $q$. Each such flat string $v$ can then be represented as a nested string $\nw{v}$ using the standard nested word encoding defined in the preliminaries. Without loss of generality, we assume that $M$ always moves its head to the left-most used tape cell before halting, i.e. the flat string encoding of a halting configuration is of the form $(\Sigma \times \{h\}) (\Sigma \times \{-\})^*$ for the halting state $h$ of $M$.
	
	It is easy to construct a polynomial-sized NWT rewriting $\op{s}\cl{s}$ into a string of the form $\op{t}^r\nw{v}_0\cl{t}^r$ (for arbitrary $r \in \mathbb{N}$ and a nested string encoding $\nw{v}_0$ of $M$'s initial configuration) and every string $\op{t}\nw{v}_i\cl{t}$ where $\nw{v}_i$ is the nested word encoding of some configuration of $M$ into $\nw{v}_{i+1}$, where $\nw{v}_{i+1}$ encodes the successor configuration of the one encoded by $\nw{v}_i$; in fact, the latter rewriting is even functional (if nondeterministic for requiring a look-ahead when $M$ moves its head to the left). We can also easily construct a DNWA accepting all nested-word encodings of non-halting configurations of $M$.
	
	It only remains to be explained how we ensure that the game is played in the fashion sketched above, i.e. that \pone plays \Call first on $\cl{s}$ and then on each $\cl{t}$ in left-to-right order. To keep \pone from leaving $\cl{s}$ or some $\cl{t}$ uncalled, we simply set up the target DNWA such that it doesn't accept any nested strings containing tags with labels $s$ or $t$. Finally, to safeguard against \pone skipping some $\cl{t}$ before calling the next, we modify the replacement transducer $R$ in such a way that it rejects any input string containing two or more nested $t$ tags, i.e. strings of the form $\op{t}\op{t}v\cl{t}\cl{t}$ with $v \in \wf$.
	
	In the game $G$ thus constructed, \pone clearly has a winning strategy on $w=\op{s}\cl{s}$ if and only if there is no $r$ such that $M$ reaches a halting configuration within $r$ steps, i.e. if $M$ doesn't halt.
\end{proof}

We now prove the lower and upper bounds for Theorem \ref{thm:unboundediiexptime}. Since context-free games with general and with non-deleting NWT are polynomially equivalent by Lemma \ref{lemma:nondeleting}, we prove the lower bound for general NWT and the lower bound for non-deleting NWT.

\begin{proposition}\label{prop:unboundediiexptimelower}
	For the class of replay-free games with NWT, $\Safelr$ is $\iiEXPTIME$-hard.
\end{proposition}

\begin{proof}
	The proof is by reduction from the complement of the \algprobname{2-Player Exponential Corridor Tiling} problem. In this problem, we are given a tiling instance consisting of a tile set $U$, vertical and horizontal constraints $V$, $H$, initial and final tile $u_i, u_f$ and unary number $n$ and consider the following game: Player 1 and Player 2 place tiles from $U$ in an alternating fashion (starting with $u_i$ for Player 1); Player 1 wins the game if Player 2 places a tile violating some vertical or horizontal constraint, or some player places a tile completing a valid tiling of width $2^n-1$, while Player 2 wins if Player 1 places an invalid tile or if the game does not terminate with a valid tiling. \algprobname{2-Player Exponential Corridor Tiling} asks whether, given a tiling instance, Player 1 has a winning strategy in this game. This problem (and therefore also its complement problem of determining whether Player 2 has a winning strategy) is complete for $\AEXPSPACE = \iiEXPTIME$ \cite{Chlebus86}.
	
	Given a tiling instance $\calI = (U, V, H, u_i, u_f, n)$, we construct a game $G = (\Sigma, \funcsymb, R, T)$ and string $w$ such that \pone has a winning strategy in $G$ on $w$ if and only if Player 2 has a winning strategy in the exponential-width tiling game on instance $\calI$.
	
	The input string is $w = (\op{d}\op{e})^{n}\op{s}\cl{s}(\cl{e}\cl{d})^{n}$. The basic idea behind the game $G$ is that \pone is first supposed to play \Call on $\cl{s}$, to which \ptwo should give a witness for the existence of a winning strategy for Player 1 in the tiling game. \pone should then use calls to the rest of the input string to uncover a flaw in \ptwos witness and thus have a winning strategy in $G$ on $w$ if and only if \ptwo cannot prove the existence of a winning strategy for Player 1.
	
	More concretely, \ptwo should give as his witness (a slightly modified linearisation of) a \emph{strategy tree} for Player 1, in which nodes are labelled with tiles from two disjoint copies $U^1$, $U^2$ of $U$, representing moves by Player 1 and Player 2. Each node labelled with an element of $U^2$ has a single child (corresponding to a move of Player 1), each node labelled with an element of $U^1$ has $|U|$ children (corresponding to the possible moves of Player 2) and the sequence of labels on each path from the root to a leaf (called a \emph{tiling candidate}) either makes up a valid tiling or contains an invalid move by Player 2. Replacement strings for \ptwo on a \Call to $\cl{s}$ will be somewhat extended linearisations of strategy trees, with extensions that we will discuss later on.
	
	Proof trees given by \ptwo should be restricted (by construction of the replacement transducer $R$) to only represent tiling candidates that start with the initial tile $u^1_i$ and contain no horizontal errors. \pones task in trying to invalidate \ptwos proof tree therefore consists of uncovering vertical errors and incorrect final tiles (as well as encoding errors within the aforementioned extensions, to be discussed later). To this end, \pone uses \Call moves on $\cl{e}$ and $\cl{d}$ tags to select parts of the proof tree given by \ptwo that correspond to a fixed column in each tiling candidate represented in the proof tree; if this column contains an error in at least one of these tiling candidates (which should be detected by the target language DNWA), \pone wins the game.
	
	We now describe how strategy trees should be encoded. Throughout the rest of the proof, we will identify nested words and their forest representations to simplify presentation.
	
	The most intuitive approach to generate strategy trees in which all tiling candidates are horizontally correct would be to fix the replacement transducer in such a way that each node labelled with some $u^1 \in U^1$ has only children labelled by some $v^2 \in U^2$ with $(u,v) \in H$; additionally, to make certain that a strategy tree correctly represents all possible counter-strategies for Player 2, we would have to require that each node corresponding to a move by Player 2 has as children nodes with \emph{all} possible labels $v^2$ such that $(u,v) \in H$. This intuitive approach leads to two minor problems, though: 
	\begin{itemize}
	\item For some tiles $u$, there may be no tiles $v$ with $(u,v) \in H$. If this happens, we fix the replacement transducer such that it may only follow $u$ up with a special ``pseudo-tile'' $u_{\Err} \notin U$.
	\item If a tile $u$ is placed at the end of a line, the next tile $v$ does \emph{not} have to fulfil $(u,v) \in H$. In order to account for this fact and still have the replacement transducer produce horizontally correct tiling candidates, we introduce a special \emph{line divider} symbol $\# \notin U$ which may be placed by the replacement transducer at any time it could normally place a node for some player, regardless of its parent tile, and which may be followed up by either \emph{any} tile $u^1 \in U^1$ as a move for Player 1, or by \emph{all} tiles $u^2 \in U^2$ as possible moves for Player 2.
	\end{itemize}
	
	In a correct encoding of a strategy tree, the line divider symbol $\#$ should occur as a node label if and only if the corresponding node's depth is a multiple of $2^n$. However, the above construction of replacement transducer allows \ptwo to construct incorrect encodings of strategy trees. Dealing with attempts by \ptwo to cheat in this manner, as well as with vertical errors, is what we discuss next.
	
	We interpret a tree given by \ptwo as a tree whose paths from root to leaf encode tiling candidates with $2^n$ columns and an arbitrary number of lines (i.e. root-to-leaf paths should always have as their length a multiple of $2^n$), with the last column in a correctly encoded tiling candidate consisting only of line divider symbols, and no line divider symbol occurring in any other column.
	
	As mentioned above, to check for vertical errors and incorrect encodings, \pone selects some column number $\ell \in [2^n]$ and removes all nodes not corresponding to some tile in that column, i.e. reduces the original tree $t$ to one containing only the nodes at depths $k \cdot (2^n) + \ell$ for all $k \geq 0$. How this is done exactly will be explained later on. Due to the alternation between Players 1 and 2, if the column chosen by \pone is not the last column, each node with a label from $U^1$ should only have children with labels from $U^2$ in this reduced tree, or vice versa.
	
	At first glance, it seems clear that, with the considerations on encoding made above, \pone can already uncover attempts by \ptwo to cheat or give a non-winning strategy, namely
	\begin{itemize}
	\item by picking a column where $\ell \neq 2^n$ and in some path of the resulting tree, there are two directly subsequent labels from $U^1$ or from $U^2$ (\ptwo tried to cheat by disrespecting the correct player sequence after a $\#$), or 
	\item by picking a column $\ell \neq 2^n$ in which some path contains a $\#$, (\ptwo tried to cheat by placing an incorrect $\#$), or
	\item by picking a column where $\ell \neq 2^n$ and in some path of the resulting tree, there is a tile $u^2 \in U^2$ followed either by $u_{\Err}$ or some $v^1 \in U^1$ with $(u,v) \notin V$ (the strategy given by \ptwo is non-winning due to a horizontal or vertical error), or 
	\item by picking $\ell = 2^n-1$ when for some path in the resulting tree, its leaf is not $u_f$ (the strategy given by \ptwo is non-winning due to a wrong final tile), or
	\item by picking $\ell = 2^n$ if the last column contains some symbol that is not $\#$ in some path (\ptwo tried to cheat by not placing a correct $\#$).
	\end{itemize}
	All of these conditions can easily be checked by a polynomial-size DNWA.
	
	However, while these conditions are indeed sufficient for checking whether \ptwo has given an incorrect or non-winning strategy tree, our construction so far is still too restrictive in that it sometimes fails to recognise a correct winning strategy tree. Consider, for instance, the case where all paths of the strategy tree either contain a vertical error for Player 2 or end with a valid tiling. According to our construction so far, \pone could still win on a strategy tree of this shape, as our construction requires \ptwo to continue giving horizontally correct tiles even after a vertical error has occurred, and these ``irrelevant'' tiles might not always make it possible to end a line with $u_f$, thus allowing \pone to find some path in the strategy tree corresponding to column $2^n-1$ that doesn't end in $u_f$.
	
	To address this problem, we further modify the replacement transducer $R$ in such a way that, whenever it is supposed to output a node with label $u^2 \in U^2$ corresponding to a move by Player 2, it may instead nondeterministically choose to output a \emph{marked} version $\hat{u}$ of $u^2$ instead. This is supposed to indicate that choosing $u^2$ leads to a vertical error for Player 2, so below the node labelled $\hat{u}$, the replacement transducer may produce a sequence of arbitrary tiles from $U^1$ for Player 1 and from $U^2$ for Player 2 (as vertically correct ``pseudo-tiles''), terminating with $u_f$ at an odd depth followed by $\#$. If a node labelled $u_{\Err}$ ever occurs instead of a move of Player 2 due to a horizontal error, it is followed up by a similar path.
	
	In this way, we make it \ptwos responsibility to flag vertical errors for Player 2. This would allow \ptwo to cheat by claiming a vertical error when, in fact, there is none, but such attempts at cheating can be penalised by once again adapting the target language in such a way that \pone wins if she selects a column in which \ptwo has falsely flagged a vertical error (i.e. the selected column number $\ell$ is not $2^n$ and some path in the corresponding tree contains a marked tile $\hat{v}$ that is not preceded by some $u^1 \in U^1$ with $(u,v) \notin V$).
	
	Next, we examine how \pone selects a column in the strategy tree given by \ptwo. Recall that the input string is $w = (\op{d}\op{e})^{n}\op{s}\cl{s}(\cl{e}\cl{d})^{n}$, and that the first move by \pone is supposed to be a \Call on $\cl{s}$ to replace $\op{s}\cl{s}$ by an encoding of a strategy tree, in which each path is to be interpreted as a tiling candidate with $2^n-1$ columns and an additional column made up of separator symbols $\#$. \pone now selects one of these $2^n$ columns by incrementally  causing \ptwo to delete either all odd-numbered columns or all even-numbered columns (i.e. all nodes at odd or even depths in the strategy tree) until only a single column is left. This is the purpose of the nodes labelled $d$ and $e$.
	
	By a \Call to some node labelled $e$, \pone causes \ptwo to change its label to $o$ (without changing the tree nested below it). A \Call to a $d$-node with $e$-labelled child causes \ptwo to delete both of these nodes as well as all nodes at even depths in the strategy tree nested below them; similarly, a \Call to a $d$-node with $o$-child deletes these nodes and all odd-depth nodes in the strategy tree.
		
	To keep \pone from cheating in this selection process, we have to make certain that she calls \emph{all} $d$-nodes in left-to-right order. We can easily keep \pone from leaving uncalled $d$-nodes by fixing the target language to not contain any strings including the label $d$; to make certain that she does not skip any $d$-nodes, the replacement transducer rejects any strings containing a $d$-node as child of an $e$- or $o$-node. 
	
	To later check for the right kinds of vertical errors or inconsistencies, we have to keep track of the column selected by \pone, i.e. the sequence of \Call moves deleting odd or even columns. The following cases have to be distinguished:
	
	\begin{itemize}
	\item If \pone deletes odd columns on all $n$ calls, she selects column $2^n$;
	\item if \pone deletes even columns as her first call and then only odd columns on the following $n-1$ calls, she selects column $2^n-1$;
	\item otherwise, \pone selects a column with index $\ell \leq 2^n-2$.
	\end{itemize}
	
	To store this information, we use an additional node, initially labelled $0$, which \ptwo returns as the root of his chosen strategy tree after \pones \Call on $\cl{s}$. This node is then updated throughout the column-deletion process as follows:
	\begin{itemize}
	\item A label $0$ signifies that no deletions have been made so far; on a \Call to $d$ which deletes odd columns, $0$ is rewritten to $l$, on a \Call deleting even columns, it becomes $x$
	\item A label $l$ (``last'') signifies that the column \pone selects may possibly be the $2^n$-th line divider column. As long as only odd columns are deleted, the label $l$ remains, if even columns are deleted, $l$ is rewritten to $r$. 
	\item A label $x$ (``neXt-to-last'') signifies that the column \pone selects may possibly be the (final) column with number $2^n-1$. As long as only even columns are deleted, the label $x$ remains, if odd columns are deleted, $x$ is rewritten to $r$.
	\item A label $r$ (``standaRd'') signifies that the column \ptwo selects will definitely not be one of the last two columns, i.e. have a number at most $2^n-2$. The label $r$ is not rewritten by any deletion.
	\end{itemize}
	
	To summarise, the game $G=(\Sigma,\funcsymb,R,T)$ constructed from the tiling instance $(U,V,H,u_i,u_f,n)$ is as follows.
	
	The alphabet of $G$ is 
	$$\Sigma = U \cup \hat{U} \cup U^1 \cup U^2 \cup \{u_{\Err}, \#, s,d,e,o, 0, l,x,r\},$$
	(where $\hat{U} = \{\hat{u} \mid u \in U\}$, $U^1 = \{u^1 \mid u \in U\}$ and $U^2 = \{u^2 \mid u \in U\}$) with function symbols $\funcsymb = \{s,d,e\}$.
	
	The replacement transducer $R$ behaves as follows:
	\begin{itemize}
	\item $R$ rewrites $\op{s}\cl{s}$ into a string of the form $\op{0} w' \cl{0}$, where $w'$ is the linearisation of a strategy tree as described above, i.e. a tree $t$ with the following properties: 
		\begin{itemize}
		\item $t$ has a root node labelled $u_i^1$;
		\item Each node labelled with some $u^1 \in U^1$ that does not have a node with label from $\hat{U} \cup \{u_{\Err}\}$ as ancestor has as its children either a single node labelled $\#$, or nodes either labelled $v^2$ or $\hat{v}$ for each $v \in U$ with $(u,v) \in H$, or a single node labelled $u_{\Err}$ if no such $v$ exists.
		\item Each node labelled with some $u^2 \in U^2$ that does not have a node with label from $\hat{U} \cup \{u_{\Err}\}$ as ancestor has a single child labelled either $v^1$ for some $v \in U$ with $(u,v) \in H$, or $u_{\Err}$ if no such $v$ exists, or $\#$.
		\item Each node that has a node with label from $\hat{U} \cup \{u_{\Err}\}$ as its own label or as an ancestor has a single child labelled either by some $u^1 \in U^1$, or by some $u^2 \in U^2$, or by $\#$.
		\item Each node labelled with $\#$ has a single child labelled by some $u^1 \in U^1$, or by $u^2 \in U^2$, or no child at all.
		\end{itemize}
	\item $R$ rewrites strings of the form $\op{e}v\cl{e}$, for arbitrary $v \in \wf$, into $\op{o}v\cl{o}$.
	\item $R$ rewrites strings of the form $\op{d} \op{e} v_1 \cl{e} \cl{d}$ into strings $v_2$ as follows:
		\begin{itemize}
		\item If $v_1 = \op{0} v'_1 \cl{0}$ (with $v'_1 \in \wf$), then $v_2 = \op{x}v'_2\cl{x}$, where $v'_2$ is derived from $v_2$ by deleting all nodes at even depths;
		\item If $v_1 = \op{l} v'_1 \cl{l}$ or $v_1 = \op{x} v'_1 \cl{x}$ or $v_1 = \op{r} v'_1 \cl{r}$ (with $v'_1 \in \wf$), then $v_2 = \op{r}v'_2\cl{r}$, where $v'_2$ is derived from $v'_1$ by deleting all nodes at even depths.
		\end{itemize}
	\item $R$ rewrites strings of the form $\op{d} \op{o} v_1 \cl{o} \cl{d}$ into strings $v_2$ as follows:
		\begin{itemize}
		\item If $v_1 = \op{0} v'_1 \cl{0}$ or $v_1 = \op{l} v'_1 \cl{l}$ (with $v'_1 \in \wf$), then $v_2 = \op{l}v'_2\cl{l}$, where $v'_2$ is derived from $v'_1$ by deleting all nodes at odd depths;
		\item If $v_1 = \op{x} v'_1 \cl{x}$ (with $v'_1 \in \wf$), then $v_2 = \op{x}v'_2\cl{x}$, where $v'_2$ is derived from $v'_1$ by deleting all nodes at odd depths;
		\item If $v_1 = \op{r} v'_1 \cl{r}$ (with $v'_1 \in \wf$), then $v_2 = \op{r}v'_2\cl{r}$, where $v'_2$ is derived from $v'_1$ by deleting all nodes at odd depths.
		\end{itemize}
	\item All other strings are rejected by $R$.
	\end{itemize}

	The target language $T$ contains all strings $v$ of the following kinds:
	\begin{itemize}
	\item $v = \op{l}v'\cl{l}$, where the tree represented by $v'$ has some path from root to leaf containing a label different from $\#$.
	\item $v = \op{r}v'\cl{r}$, where the tree represented by $v'$ has some path from root to leaf
		\begin{itemize}
		\item containing a label $\#$, or
		\item containing two subsequent labels from $U^1$ or two subsequent labels from $U^2$, or
		\item containing some label $u^2$ from $U^2$ followed either by $u_{\Err}$ or by a $u'^1$ with $(u,u') \notin V$, or
		\item containing some label from $\hat{U}$ that is \emph{not} part of a vertical error.
		\end{itemize}	
	\item $v = \op{x}v'\cl{x}$, where the tree represented by $v'$ has some path from root to leaf
		\begin{itemize}
		\item containing a label $\#$, or
		\item containing two subsequent labels from $U^1$ or two subsequent labels from $U^2$, or
		\item containing some label $u^2$ from $U^2$ followed either by $u_{\Err}$ or by a $u'^1$ with $(u,u') \notin V$, or
		\item containing some label from $\hat{U}$ that is \emph{not} part of a vertical error
		\item containing \emph{no} vertical error, \emph{no} label $u_{\Err}$ and ending with a label different from $u^1_f, u^2_f$.
		\end{itemize}	
	\end{itemize}

	It is relatively easy to see (but tedious to prove formally) that this construction is possible in polynomial time and that \pone indeed has a winning strategy in $G$ on $w$ if and only if Player 1 has no winning strategy on the tiling instance $\calI$.
\end{proof}

For the upper bound, we first prove decidability in $\iiEXPTIME$ for a ``purely NWT-based'' problem, to which we will later reduce \Safelr.

\begin{definition}
	The \emph{alternating iterated transduction} problem for non-deleting NWT is defined as the following decision problem:\\
	\begin{centering}
  \algproblem{\textsc{AIT}(NWT)}
  					{A string $w \in \wf$, number $k$ (given in unary), DNWA $A$, and pairs of non-deleting NWTs $(T_{1,0},T_{1,1}), \ldots, (T_{k,0},T_{k,1})$}
  					{Is there an $i_1 \in \{0,1\}$ such that for every $w_1 \in T_{1,i_1}(w)$ there exists $i_2 \in \{0,1\}$ such that for every $w_2 \in T_{2,i_2}(w_1)$ \ldots there exists $i_k \in \{0,1\}$ such that for every $w_k \in T_{k,i_k}(w_{k-1})$ it holds that $w_k \in L(A)$?}
	\end{centering}
\end{definition}

\begin{proposition}\label{prop:aitupper}
	$\algprobname{AIT(NWT)} \in \iiEXPTIME$
\end{proposition}

\begin{proof}
We assume, without loss of generality, that each transducer $T_{j,i}$ (for $j \in [k], i \in \{0,1\}$) is in normal form, i.e. that each (reading or $\epsilon$-) transition of each transducer produces exactly one output tag (cf Lemma \ref{lemma:normalform}).

The idea behind this proof is to eliminate the existential quantification in the problem setting by constructing an NWT $T$ that simulates in parallel \emph{both} transducers $T_{j,1}$ and $T_{j,2}$ at each level $j \in [k]$. More concretely, $T$ takes as input a string $w$ and outputs a $2^k$-tuple of strings, each component of which corresponds to a sequence of existential choices $i_1 \ldots i_k \in \{0,1\}^k$ of transducers, while the non-determinism in $T$ simulates universal choice. It then holds (as we will prove after the construction of $T$) that the condition of $\algprobname{AIT(NWT)}$ is fulfilled if and only if for each possible transduct $w'$ of $w$ by $T$, at least one of the $2^k$ component strings is contained in (an appropriate modification of) the target language $L(A)$.

To construct $T$, we construct from $T_{j,0}$ and $T_{j,1}$ for each level $j \in [k]$ a transducer $T_j$ that takes as input a nested string $w_{j-1}$ over an alphabet of $2^{j-1}$-tuples of alphabet symbols and outputs a nested string $w_j$ over $2^j$-tuples, with the intuition being that $T_j$ simulates one run of both $T_{j,0}$ and $T_{j,1}$ on each of the $2^{j-1}$ input strings encoded in $w_{j-1}$ to produce a total of $2 \cdot 2^{j-1} = 2^j$ output strings, which are encoded in $w_j$. From all of the transducers $T_j$ for all $j \in [k]$, we then use Proposition \ref{prop:nwtcomposition} to construct $T$ as the transducer for $T_k \circ \ldots \circ T_1$.

The main difficulty in the construction of each $T_j$ is the fact that, while both $T_{j,0}$ and $T_{j,1}$ should be simulated for exactly one run on each of the components on the input string, all of these runs should be independent of each other. For instance, if $T_2$ wanted to simulate some run $\rho$ of $T_{2,0}$ on the first and a different run $\rho'$ of $T_{2,0}$ on the second input string component, $\rho$ might start by reading the first tag from the input string while $\rho'$ starts with an $\epsilon$-transition and reads the first input tag afterwards. In such a situation, $\rho'$ would produce an output before the output of $\rho$ even starts. We therefore need to construct each $T_j$ in such a manner that runs are \emph{synchronised}, producing an output for \emph{all} components of the output string, even if only \emph{one} run of some transducer $T_{j,i_j}$ calls for producing an output.

To address this problem, we introduce a special \emph{blank} symbol $\blank \notin \Sigma$ and construct each $T_j$ in such a way that reading transitions of $T_j$ simulate synchronous reading transitions of $T_{j,0}$ and $T_{j,1}$ on \emph{all} symbols of the next input tuple; $\epsilon$-transitions of $T_j$, on the other hand, simulate only \emph{one} of $T_{j,0}$ or $T_{j,1}$ on only \emph{one} of the $2^{j-1}$ input components and output tuples consisting of exactly one symbol from $\Sigma$ in the corresponding component and $\blank$ symbols in all the others. To keep this construction consistent and allow for later transductions, this also means that all transducers $T_{j,i}$ ($j \in [k]$, $i \in \{0,1\}$) have to be modified in such a way that, on reading an opening (closing) $\blank$ tag, they always output an opening (closing) $\blank$ tag and do not change states.

More formally, if (for each $j \in [k]$ and $i \in \{0,1\}$) $T_{j,i} = (Q_{j,i}, P_{j,i}, P^\epsilon_{j,i}, \Sigma, \delta_{j,i}, q^0_{j,i}, F_{j,i})$, then let first $T'_{j,i}$ be an extension of $T_{j,i}$ to strings containing $\blank$ tags that ignores and reproduces these tags as described above, i.e. $T'_{j,i} = (Q_{j,i}, P_{j,i}\cup \{p_\blank\}, P^\epsilon_{j,i}, \Sigma \cup \{\blank\}, \delta'_{j,i}, q^0_{j,i}, F_{j,i})$, where $p_\blank \notin P_{j,i}$ and $\delta'_{j,i}$ consists of $\delta_{j,i}$ extended by transitions $(q,\op{\blank},q,p_\blank,\op{\blank})$ and $(q,p_\blank,\cl{\blank},q,\cl{\blank})$ for each $q \in Q_{j,i}$. For any $w \in \wf[\Sigma\cup\{\blank\}]$, let $\Strip(w) \in \wf$ denote the nested string obtained from $w$ by deleting all $\blank$-labelled tags, and for any set $S \subs \wf[\Sigma\cup\{\blank\}]$, let $\Strip(S) \mydef \{\Strip(w) \mid w \in S\}$. Then, it is clear that for any $w \in \wf[\Sigma\cup\{\blank\}]$, it holds that $\Strip(T'_{j,i}(w)) = T_{j,i}(\Strip(w))$.

Now we describe in detail the construction of $T_j$ from $T'_{j,0}$ and $T'_{j,1}$. Let 
$$T_j = (Q_j, P_j, P^\epsilon_j, \Sigma_j^{\text{in}} \cup \Sigma_j^{\text{out}}, \delta_j, q^0_j, F_j),$$
 where the linear and hierarchical state sets as well as initial and final states simply derive from a $2^{j-1}$-fold product construction of $T'_{j,0}$ and $T'_{j,1}$, i.e. 
 $$Q_j = (Q_{j,0} \times Q_{j,1})^{2^{j-1}},$$
 $$P_j = ((P_{j,0} \cup \{p_\blank\}) \times (P_{j,1} \cup \{p_\blank\}))^{2^{j-1}},$$
 $$q^0_j = (q^0_{j,0} \times q^0_{j,1})^{2^{j-1}},$$ and 
 $$F_j = (F_{j,0} \times F_{j,1})^{2^{j-1}}.$$ 
 The working alphabet of $T_j$ consists of input alphabet $\Sigma_j^{\text{in}} = (\Sigma \cup \{\blank\})^{2^{j-1}}$ and output alphabet $\Sigma_j^{\text{out}} = (\Sigma \cup \{\blank\})^{2^{j}}$. The construction of $P_j^\epsilon$ is similar, but has to be adjusted slightly; since $\epsilon$-transitions of $T_j$ are supposed to simulate $T_{j,0}$ or $T_{j,1}$ on only \emph{one} component of the output string and the other components have to be filled in with $\blank$ symbols, we add a new hierarchical $\epsilon$-state $p^\epsilon_\blank \notin P_{j,0} \cup P_{j,1}$ and set $P_j^\epsilon = ((P_{j,0}^\epsilon \cup \{p_\blank^\epsilon\}) \times (P_{j,1}\epsilon \cup \{p_\blank^\epsilon\}))^{2^{j-1}}$.

In accordance with the above intuition, we construct the transition relation $\delta_j$ of $T_j$. Reading transitions in $\delta_j$ (i.e. transitions that read one input symbol and produce one output symbol) are pretty much products of transitions from $\delta'_{j,0}$ and $\delta'_{j,1}$. That is, $\delta_j$ contains an opening transition which, for $\ell \in [2^j]$, starts with linear state $q$, ends in state $q'$ and produces hierarchical state $q$ in the $\ell$-th position of the corresponding state tuples while reading $\op{a}$ in the $\lfloor\frac{\ell}{2}\rfloor$-th position of the input tuple and writing $\op{b}$ in the $\ell$-th position of the output tuple if the transition $(q,\op{a},q',p,\op{b})$ is in $\delta'_{j,0}$ (for odd $\ell$) or in $\delta'_{j,1}$ (for even $\ell$).\footnote{At first glance, it may seem more intuitive to associate even positions with $T_{j,0}$ and odd positions with $T_{j,1}$, but seeing as $\ell$ is a number between 1 and $2^j$ which will later be encoded as $\ell-1$ by a binary sequence of length $j$, the association described here is indeed the more useful one.} Closing reading transitions in $\delta_j$ are constructed accordingly.

As for $\epsilon$-transitions, $\delta_j$ contains an opening transition $(\ol{q}, \op{\epsilon}, \ol{q}', \ol{p}, \ol{a})$ if there is an $\ell \in [2^j]$ such that (with $\ol{x}_i$ denoting the $i$-th component of a tuple $\ol{x}$)
\begin{itemize}
	\item for all $\ell' \neq \ell$, it holds that $\ol{q}_{\ell'} = \ol{q}'_{\ell'}$, $\ol{p}_{\ell'} = p^\epsilon_\blank$ and $\ol{a}_{\ell'} = \blank$, and
	\item for $\ol{q}_{\ell} = q$, $\ol{q}'_{\ell} = q'$, $\ol{p}_{\ell} = p$ and $\ol{a}_{\ell} = a$, there is a transition $(q, \op{\epsilon}, q', p, \op{a})$ in $\delta'_{j,0}$ (for odd $\ell$) or in $\delta'_{j,1}$ (for even $\ell$),
\end{itemize}
and accordingly for closing $\epsilon$-transitions.

With this construction, $T_j$ is indeed an NWT in normal form (i.e. fulfils the $\epsilon$-consistency, well-formedness and synchronisation properties). This directly implies that $T_j$ transduces well-nested strings over $\Sigma_j^{\text{in}}$ into well-nested strings over $\Sigma_j^{\text{out}}$. We now identify nested strings over tuples of alphabet symbols with tuples of  nested strings as follows: For some $j$, let $\ol{w} = \ol{w}^1 \cdots \ol{w}^n$ be a well-nested string over $(\Sigma \cup \{\blank\})^{2^j}$, and for each $\ol{w}^i \in \op{(\Sigma \cup \{\blank\})^{2^j}}$, let $\ol{w}^i_\ell$ denote its $\ell$-th component interpreted as an opening tag (and analogously for closing tags). Then the nested string $\ol{w}_\ell$ is defined as $\ol{w}_\ell = \ol{w}^1_\ell \cdots \ol{w}^n_\ell$. It is clear that, if $\ol{w} \in \wf[(\Sigma \cup \{\blank\})^{2^j}]$, then $\ol{w}_\ell \in \wf[\Sigma \cup \{\blank\}]$ for each $\ell$.

Identifying strings of tuples with tuples of strings in this way, it is easy (if tedious) to prove that, for each $j$ and each string $\ol{w} \in \wf[\Sigma_j^\text{in}]$, interpreted as a $2^{j-1}$-tuple $(\ol{w}_1, \ldots, \ol{w}_{2^{j-1}})$ of nested strings, it holds that $\Strip(T_j(w)) = \prod_{\ell=1}^{2^{j-1}} T_{j,0}(\Strip(\ol{w}_\ell)) \times T_{j,1}(\Strip(\ol{w}_\ell))$, where the $\Strip$ operator is applied component-wise, i.e. $\Strip((\ol{w}_1, \ldots, \ol{w}_n)) = (\Strip(\ol{w}_1), \ldots, \Strip(\ol{w}_n))$. In other words, disregarding $\blank$ tags, every run of $T_j$ on a $2^{j-1}$-tuple of nested words simulates one run each of $T_{j,0}$ and $T_{j,1}$ on each of its component strings, and all combinations of such component runs can be simulated by a run of $T_j$.

By a simple induction argument, it follows that each transducer of the form $T_j \circ \cdots \circ T_1$ completely describes all possible series of existential choices up to the $j$-th level. More precisely, denoting by $S_\ell = \{\ol{w}_\ell \mid \ol{w} \in S\}$ the set of all $\ell$-th component strings in a set $S$ of tuples of nested strings, we get that for every $j \leq k$, $w \in \wf$ and $\ell \in [2^j]$, it holds that $\Strip((T_j \circ \cdots \circ T_1(w))_\ell) = T_{j,i_j} \circ \cdots T_{1,i_1}(w)$, where $i_1 \cdots i_j \in \{0,1\}^j$ is the binary representation of the number $\ell-1$.

To use the transducer $T = T_k \circ \cdots \circ T_1$ for solving the alternating iterated transduction problem as described initially, we now need to prove that the defining property of $\algprobname{AIT(NWT)}$ (i.e. ``There is an $i_1 \in \{0,1\}$ such that for every $w_1 \in T_{1,i_1}(w)$...'') is equivalent to the following: For every $\ol{w} \in T(w)$, there is an $\ell \in [2^k]$ such that $\Strip(\ol{w}_\ell) \in L(A)$. We denote this property by (*). As we have already seen, the existence of $\ell$ in property (*) is equivalent to the existence of an index sequence $i_1 \cdots i_k \in \{0,1\}^k$ representing $\ell-1$ in binary, so we basically need to prove that the alternation between existential and universal choices in the defining property of $\algprobname{AIT(NWT)}$ is equivalent to a single universal choice of iterated transducts for each possible index sequence of existential choices, followed by a single existential choice of index sequence.

The proof of this equivalence is by induction; the central step is proving, for any set $S$ of nested strings, the equivalence of
 $$\forall \ol{w} \in T_{j-1} \circ \cdots \circ T_1(w) \exists \ell \in [2^{j-1}] \exists i_j \in \{0,1\} \forall w' \in T_{j,i_j}(\Strip(\ol{w}_\ell)): w' \in S$$
 and
 $$\forall \ol{w}' \in T_j \circ T_{j-1} \circ \ldots \circ T_1(w) \exists \ell' \in [2^j]: \Strip(\ol{w}'_{\ell'}) \in S,$$
which follows primarily from the fact that $T_j$ simulates both $T_{j,0}$ and $T_{j,1}$ as was shown before.

To sum up the proof thus far, we have constructed from an $\algprobname{AIT(NWT)}$ instance a transducer $T$ such that the original instance is a positive one if and only if for the input string $w$ it holds that all transducts in $T(w)$ have a component string that is in $L(A)$ when stripped of all $\blank$-labelled tags. It remains to be seen how we can check for this property in doubly exponential time.

To that end, let $A'$ be a modification of $A$ that ignores $\blank$-labelled tags, constructed from $A$ using a similar construction to that one for each $T_{j,i_j}$. From $A'$, we can construct a DNWA $B$ that gets as input a nested string of tags over $2^k$-tuples of symbols from $\Sigma \cup \{\blank\}$, simulates a copy of $A'$ on each of the $2^k$ components and accepts if and only if at least one of its components accepts. It then holds that the original $\algprobname{AIT(NWT)}$ instance is a positive one if and only if $T(w) \subseteq L(B)$, which yields an instance of a type checking problem. Each level transducer $T_j$ is of size at most $\bigO((|T_{j,0}|\cdot |T_{j,1}|)^{2^k})$, therefore $T$ is of doubly exponential size, and so is $B$. Since the type checking problem with target DNWA is decidable in $\PTIME$ by Theorem \ref{thm:nwttypechecking}(b), this yields a $\iiEXPTIME$ algorithm for $\algprobname{AIT(NWT)}$, as was to be proven.
\end{proof}

A reduction to $\AIT$, along with Proposition \ref{prop:unboundediiexptimelower} now proves Theorem \ref{thm:unboundediiexptime}.

\begin{restate}{Theorem \ref{thm:unboundediiexptime}}
	For the class of replay-free games with NWT, $\Safelr$ is $\iiEXPTIME$-complete.
\end{restate}

\begin{proof}
	The lower bound was proven as Proposition \ref{prop:unboundediiexptimelower}. We prove a matching upper bound by reduction to $\AIT$, which is in $\iiEXPTIME$ according to Proposition \ref{prop:aitupper}.
	
	Let $G=(\Sigma,\funcsymb,R,T)$ be a game with NWT replacement, and let $w$ be an input string for $G$. Let furthermore $k \leq \frac{|w|}{2}$ be the number of occurrences of closing tags from $\cl{\funcsymb}$ in $w$. 
	
	The idea behind the reduction to $\AIT$ is taking $k$ rounds of alternating transduction, where the $j$-th round (with transducers $T_{j,0}$, $T_{j,1}$) corresponds to the replay-free subgame on the $j$-th function symbol in $w$ (in left-to-right order). The choice between transducers $T_{j,0}$ and $T_{j,1}$ models \pones choice between \Read and \Call; to that end, $T_{j,0}$ basically does not change its input string at all, while $T_{j,1}$ simulates the replacement transducer $R$ on the substring that \pone chose to be replaced. The only minor technical difficulty in this construction is the fact that, in the game $G$, the transducer $R$ only rewrites the called substring, while each $T_{j,i}$ rewrites the entirety of the current string. This difficulty can be solved by some minor modifications, which we will now examine.
	
	The input string $w'$ for $\AIT$ is derived from $w$ by replacing, in left-to-right order, each substring $\op{f}v\cl{f}$ of $w$ by $\op{j}\op{f}v\cl{f}\cl{j}$, where $\cl{f}$ is the $j$-th closing function tag in $w$ and $j \notin \Sigma$ for each $j \in [k]$. In other words, the substring on which \pone has to make her $j$-th strategy decision is encapsulated in $j$-tags.
	
	For each $j \in [k]$, the transducer $T_{j,0}$ simply deletes the $\op{j}$ and $\cl{j}$ tags from its input, leaving it otherwise unchanged. The transducer $T_{j,1}$, on the other hand, also directly outputs its input until it reaches the $\op{j}$ tag. It deletes this tag and then starts simulating the replacement transducer $R$. Once $T_{j,1}$ reaches the $\cl{j}$ tag, it deletes that tag as well and stops its simulation, rejecting its input if $R$ has not reached an accepting state. Afterwards, $T_{j,1}$ simply outputs its input again. Note that the simulation of $R$ in $T_{j,1}$ will never receive as input any tags with labels not in $\Sigma$, as all such tags have a label strictly less than $j$ and have therefore already been removed by earlier transductions.
	
	The target DNWA $A$ for $\AIT$ is simply the target DNWA $A(T)$ of $G$.
	
	It is easy to see that $k$, $w'$ and each $T_{j,0}$ for $j \in [k]$ can be computed from $w$ in polynomial time, as can each $T_{j,1}$ from $R$. As the alternating transduction simulates the replay-free game, it is also clear that \pone has a replay-free winning strategy on $w$ in $G$ if and only if the constructed instance for $\AIT$ is a positive one, which concludes the reduction.
\end{proof}

\begin{restate}{Theorem \ref{thm:depthboundedexpspace}}
	For the class of replay-free games with depth-bounded NWT, $\Safelr$ is $\EXPSPACE$-complete.
\end{restate}

We prove the lower and upper bounds of this theorem as separate propositions.

\begin{proposition}\label{prop:depthboundedlower}
	For the class of replay-free games with depth-bounded NWT, $\Safelr$ is $\EXPSPACE$-hard.
\end{proposition}

\begin{proof}(sketch)
	This lower bound proof uses the problem \algprobname{Exponential Corridor Tiling}: Given a tiling instance consisting of a tile set $U$, vertical and horizontal constraints $V$, $H$, initial and final tile $u_i, u_f$ and unary number $n$, is there a tiling of width $2^n-1$ and arbitrary height?
	
	From any input tiling instance, we construct a game $G = (\Sigma, \funcsymb, R, T)$ with depth-bounded NWT $R$ and an input word $w$ such that \pone has a winning strategy on $w$ in $G$ if and only if there exists \emph{no} valid tiling of width $2^n-1$, i.e. we reduce from the complement of \algprobname{Exponential Corridor Tiling}. Since \algprobname{Exponential Corridor Tiling} is $\EXPSPACE$-complete, so is its complement.
	
	The basic idea behind the construction of $G$ is similar to that used in the proof of Proposition \ref{prop:unboundediiexptimelower}. The input string is $w = \op{r}\op{o_n}\op{e_n}\cdots\op{o_1}\op{e_1}\op{s}\cl{s}\cl{e_1}\cl{o_1}\cdots \cl{e_n}\cl{o_n}\cl{r}$. \pone is first supposed to \Call $\cl{s}$, allowing \ptwo to respond with the standard nested string encoding $\nw{v}_0$ of a flat string $v_0 \in (U \cup \{\#\})^*$ which is supposed to encode a valid tiling $t$ of width $2^n-1$ in the standard way (i.e. lines of $t$ are concatenated and separated by $\#$ symbols). We can fix the transducer $R$ in such a way that $v_0$ is always a concatenation of horizontally correct substrings (not necessarily of length $2^n-1$) separated by $\#$ symbols. After $\nw{v}_0$ is given, \pone then plays \Call on either $e_1$ or $o_1$, prompting \ptwo to delete either all even-numbered or all odd-numbered positions in $v_0$ and yielding the standard nested string encoding $\nw{v}_1$ of the resulting flat string $v_1$. Continuing this process further (i.e. sequentially calling exactly one of $e_i$ or $o_i$ for each $i \in [n]$) eventually yields a nested string $\nw{v}_n$ which encodes what is supposed to be a single column in the tiling given by \ptwo. This string is then checked for vertical correctness by the target DNWA. In this way, \pone has a winning strategy on $w$ in $G$ if and only if every tiling candidate that can be provided by \ptwo contains some (horizontal or encoding) error, i.e. if there is no valid tiling of the desired width.
\end{proof}

Similar to the proof of Theorem \ref{thm:unboundediiexptime}, the upper bound in Theorem \ref{thm:depthboundedexpspace} is proven by a reduction to (a variant of) the alternating iterated transduction problem: Given a string $w \in \wf$, number $k$ (given in unary), DNWA $A$, and pairs of depth-bounded NWTs $(T_{1,0},T_{1,1}), \ldots, (T_{k,0},T_{k,1})$, is there an $i_1 \in \{0,1\}$ such that for every $w_1 \in T_{1,i_1}(w)$ there exists $i_2 \in \{0,1\}$ such that for every $w_2 \in T_{2,i_2}(w_1)$ \ldots there exists $i_k \in \{0,1\}$ such that for every $w_k \in T_{k,i_k}(w_{k-1})$ it holds that $w_k \in L(A)$?

\begin{proposition}\label{prop:depthboundedaitupper}
	The alternating iterated transduction problem for depth-bounded NWT is in $\EXPSPACE$.
\end{proposition}

\begin{proof}
	Basically, the exponential-space algorithm solving this problem follows the same idea as the $\iiEXPTIME$ algorithm from the proof of Proposition \ref{prop:aitupper}; here, the bounded depth of input transducers ensures that the construction and type checking of transducers $T_j$ (for each $j$) and $T$ from that proof may be simulated on-the-fly by a co-nondeterministic exponential-space algorithm $\calA$.
	
	More specifically, it can be easily proven by induction that, in any run $\rho$ of an NWT receiving as input a string of depth $d_{\text{in}}$ and outputting a string of depth at most $d_{\text{out}}$, the sequence of hierarchical states in any configuration occurring in $\rho$ has length at most $d_{\text{in}} + d_{\text{out}}$.
	
	Since all input transducers are depth-bounded, we can assume without loss of generality that they have a common upper bound $d$ on their output depth. Then, for each $j >1$, each $T_{j,i}$ has an input \emph{and} output depth bounded by $d$, since $T_{j,i}$ receives as its input the output of some $T_{j-1, i'}$. On the other hand, $T_{1,0}$ and $T_{1,1}$ also have an output depth bounded by $d$ and an input depth bounded by $|w|$, which is also fixed for fixed input strings $w$.	Similarly, the DNWA $A$ receives as inputs only outputs of $T_{k,0}$ or $T_{k,1}$, so we may restrict our attention to configurations of $A$ with a hierarchical state sequence of length at most $d$.
	
	The idea behind the algorithm $\calA$, then, is to traverse the input string $w$ from left to right and, on each input tag, co-nondeterministically guess transducts for each combination of follow-up transducers and check that at least one resulting final transduct is in $L(A)$, thus simulating the transducer $T$ from the proof of Proposition \ref{prop:aitupper} and checking that each transduct from $T(w)$ is accepted by the NWA $B$ from that proof.
	
	The algorithm $\calA$ is organised into $k$ layers, with the $j$-th layer (for $j \in [k]$) simulating the workings of transducer $T_j$ from the proof of Proposition \ref{prop:aitupper}, which takes as input $2^{j-1}$-tuples of symbols from $\Sigma\cup \{\blank\}$ and outputs $2^k$-tuples of such symbols by guessing transitions of both $T_{j,0}$ and $T_{j,1}$ on each component of the input tuple. The algorithm $\calA$ does this simulation based on a single tuple of symbols for each layer (called that layer's \emph{active input tuple}) -- once layer $j-1$ produces an output tuple, this tuple becomes the active input for layer $j$, and $\calA$ continues its simulation with layer $j$, producing an output to become the active input tuple for layer $j+1$ by either consuming the active input tuple of layer $j$ or by an $\epsilon$-transition (which leaves the active input tuple for layer $j$ unchanged). Once the active input tuple of some layer $j$ has been consumed, control passes back to layer $j-1$. The bottom layer $k$, instead of producing further active input tuples, directly simulates the effect of its output on $2^k$ modified copies of the DNWA $A$ (i.e. simulates the NWA $B$ from the proof of Proposition \ref{prop:aitupper}).
	
	Some extra care has to be taken regarding deletions. The proof sketch so far (as well as the construction for Proposition \ref{prop:aitupper}) assumes that simulated transitions are synchronised in such a way that, when an input tuple is consumed, each component of that tuple produces two components of the output tuple. This assumption obviously does not hold in the presence of deleting transitions. For this reason, we allow active input tuples to contain \emph{null positions} labelled $\Null$, which do not induce any transition on the corresponding transducers once the tuple is consumed.
	
	In the sequel, as in the proof of Proposition \ref{prop:aitupper}, assume that all input transducers are in normal form as per Lemma \ref{lemma:normalform}, and for each $i \in \{0,1\}$ and $j \in [k]$, let $T'_{j,i}$ be an extension of $T_{j,i}$ to strings containing $\blank$ tags that ignores and reproduces these tags, and let $A'$ be an analogous extension of $A$. We now describe the algorithm $\calA$ in more detail.
	
	For each $j$, the following information is stored for layer $j$:
	\begin{itemize}
	\item A $2^j$-tuple of configurations of $T'_{j,0}$ (in odd positions) and $T'_{j,1}$ (in even positions), each consisting of a linear state and a sequence of at most $d+|w|$ (for $j=1$) or $2d$ (for $j>1$) hierarchical states. These configurations are initialised with the starting configuration of $T'_{j,0}$ or $T'_{j,1}$.
	\item A $2^{j-1}$-tuple of symbols from $\op{\Sigma \cup \{\blank\}} \cup \cl{\Sigma \cup \{\blank\}} \cup \{\Null\}$, where either all non-$\Null$ components are opening tags or all are closing tags (called the \emph{active input tuple} of layer $j$). The non-$\Null$ components of each active input tuple gives the next input symbols to be consumed by ($2^{j-1}$ copies of) $T_{j,0}$ and $T_{j,1}$. The active input tuple of layer $j$ is initialised to $\Null^{2^{j-1}}$
	\end{itemize}
	
	Additionally, $\calA$ stores the following global information:
	\begin{itemize}
	\item A counter $\ell \in \{0, \ldots, |w|\}$ denoting the last position of the input string $w$ that has been read, initialised to $\ell = 0$.
	\item A $2^k$-tuple of configurations of $A'$, each consisting of a single linear state and a sequence of up to $d$ hierarchical states. These configurations are initialised with the starting configuration of $A'$.
	\item A layer counter $j \in \{0, \ldots, k\}$ denoting the current \emph{active layer} being processed (with layer 0 corresponding to the input string itself). This counter is initialised to $j=0$.
	\end{itemize}

	All of this information can obviously be stored in exponential space.
	
	The algorithm $\calA$ now proceeds as follows:
	\begin{enumerate}
	\item If $j=0$...
		\begin{enumerate}
		\item ...and $\ell < |w|$, then the input has not yet been completely read. In this case, $\calA$ sets the active input ($2^0$-)tuple of layer 1 to the $(l+1)$-th symbol of $w$, $\ell := \ell +1$, and $j := 1$.
		\item ...and $\ell = |w|$, then the input has been processed completely, and $\calA$ halts.
		\end{enumerate}
	\item If $0<j<k$...
		\begin{enumerate}
		\item ...and the active input tuple of layer $j$ is \emph{not} equal to $\Null^{2^{j-1}}$, then layer $j$ still has some input to be processed. In this case, $\calA$ guesses co-nondeterministically which of the following two steps to perform:
			\begin{enumerate}
			\item $\calA$ processes the input of layer $j$, i.e. for each non-$\Null$ position $i \in [2^{j-1}]$ of the active input tuple, $\calA$ guesses a transition of $T'_{j,0}$ with the input symbol from that position starting at the configuration in position $2i-1$ of the configuration tuple for layer $j$, and updates that configuration accordingly; similarly, $\calA$ guesses a transition for $T'_{j,1}$ with position $2i$ of the configuration tuple. The outputs of these transitions are written to positions $2i-1$ and $2i$ of a new $2^j$-tuple $t$. For any $\Null$-position $i \in [2^{j-1}]$, positions $2i-1$ and $2i$ of $t$ are then filled with $\Null$ markers. Afterwards, the active input tuple of layer $j+1$ is set to $t$, the active input tuple of layer $j$ is reset to $\Null^{2^{j-1}}$, and $j := j+1$.
			\item $\calA$ performs an $\epsilon$-transition. To that end, $\calA$ guesses an index $i \in [2^j]$. If $i$ is odd, $\calA$ guesses an opening (resp. closing) $\epsilon$-transition for $T'_{j,0}$ available in the configuration at position $(i+1)/2$ of the configuration tuple and updates that configuration accordingly, while configurations at all other positions remain unchanged. It then writes the output of that transition to position $(i+1)/2$ of a new $2^j$-tuple $t$ and fills all other positions of $t$ with $\op{\blank}$ (resp. $\cl{\blank}$) symbols. If $i$ is even, $\calA$ proceeds analogously with $T'_{j,1}$ instead of $T'_{j,0}$ and $i/2$ instead of $(i+1)/2$. Finally, $\calA$ sets the active input tuple of layer $j+1$ to $t$ and $j := j+1$ \emph{without} resetting the active input tuple of layer $j$.
			\end{enumerate}
		\item ...and the active input tuple of layer $j$ is $\Null^{2^{j-1}}$, i.e. layer $j$ currently has no input waiting to be processed. In this case, $\calA$ guesses co-nondeterministically which of the following two steps to perform:
			\begin{enumerate}
			\item $\calA$ performs an $\epsilon$-transition, as described under 2.a.ii.
			\item $\calA$ ends its processing of layer $j$ and sets $j:=j-1$.
			\end{enumerate}
		\end{enumerate}
	\item If $j = k$, then $\calA$ basically proceeds as described under 2., with the only difference being that any time $\calA$ has created an output tuple $t$ of size $2^k$, instead of setting $j:=j+1$, $\calA$ directly aggregates $t$ onto the configuration $2^k$-tuple for $A'$, simulating, for each $i \in [2^k]$, a transition of $A'$ starting from the $i$-th component of the configuration tuple, consuming the $i$-th component of $t$ and storing the resulting configuration in the $i$-th position of the configuration tuple for $A'$.
	\end{enumerate}
	
	Finally, once $\calA$ halts, it accepts if at least one of the following conditions is met:
	\begin{itemize}
	\item One of the $2^k$ stored configurations for $A'$ is accepting. In this case, there is some sequence of existential choices of transducers such that, for the universal choices made co-nondeterministically by $\calA$, the final transduct is in $L(A)$.
	\item For some $j \in [k]$, the configuration tuple for layer $j$ contains some non-accepting configuration. In this case, the co-nondeterministic choices taken by $\calA$ have lead to an incorrect transduction being performed, i.e. the corresponding run of $\calA$ should not be counted against the acceptance condition.
	\end{itemize}
	
	To show correctness of $\calA$, it needs to be proven that there exists a non-accepting run of $\calA$ if and only if for all $i_1 \in \{0,1\}$ there exists $w_1 \in T_{1,i_1}(w)$ such that for all $i_2 \in \{0,1\}$ \ldots for all $i_k \in \{0,1\}$ there exists $w_k \in T_{k, i_k}(w_{k-1})$ with $w_k \notin L(A)$. 
	
	The proof for this is rather technical, but its basic idea is as follows: For the ``only if'' direction, we extract from an accepting run of $\calA$ each ``witness string'' $w_j$ inductively based on the sequence $i_1, \ldots, i_j$ of universal choices and prior witness strings $w_1, \ldots, w_{j-1}$ by taking the output string produced in component $i+1$ of layer $j$, where $i \in \{0, \ldots, 2^j-1$ is the number represented by the binary encoding $i_1 \cdots i_j$. Using the construction of $\calA$, it is then easy to see that $w_j \in T_{j,i_j}(w_{j-1})$, and that $w_k \notin L(A)$. Similarly, for the ``if'' direction we can construct a run of $\calA$ from witness strings given universal choices of indices in $\{0,1\}$, which is non-accepting because all transductions are performed correctly and, regardless of universal choices, the resulting final string $w_k$ is not in $L(A)$.
	
	Finally, as stated above, $\calA$ requires only co-nondeterministic exponential space and can therefore be simulated by an $\EXPSPACE$ algorithm.
\end{proof}

\begin{proofof}{Theorem \ref{thm:depthboundedexpspace}}
	The lower bound was proven as Proposition \ref{prop:depthboundedlower}. For the upper bound, we reduce $\Safelr$ to the alternating iterated transduction problem for depth-bounded transducers using the same reduction as in the proof of Theorem \ref{thm:unboundediiexptime}.
	
	Some care needs to be taken with this reduction, as the reduction from Theorem \ref{thm:unboundediiexptime} constructs, from an input string $w$ and a game $G$ whose replacement transducer has output depth bounded by some constant $d$, an instance for the alternating iterated transduction problem whose transducers have a depth bound $d+|w|$; that is, the depth bound for the alternating iterated transduction problem actually depends on the size of the input. However, the proof of Proposition \ref{prop:depthboundedaitupper} shows that the algorithm given there still only takes exponential space, even if the depth bound of transducers depends on the input, since in this case, the algorithm keeps track of (exponentially many) configurations whose size is still at most polynomial in the input size. In this way, Proposition \ref{prop:depthboundedaitupper} still yields an $\EXPSPACE$ upper bound for $\Safelr$.
\end{proofof}\newpage
		\section*{Proofs for Section \ref{sec:bounded}}

This section gives proofs for the results of Section \ref{sec:bounded}. Since Lemma \ref{lemma:nondeleting} applies in this setting, we generally prove lower bounds for deleting $\epsilon$-free NWT and upper bounds for non-deleting $\epsilon$-free NWT.

\begin{restate}{Theorem \ref{thm:boundedundecidable}}
	For the class of games with $\epsilon$-free NWT and unbounded replay, $\Safelr$ is undecidable.
\end{restate}

\begin{proof}
	The proof is by reduction from the halting problem for Turing machines with an empty input. From a given Turing machine $M$ with working alphabet $\Sigma$ and state set $Q$, we construct a game $G$ and string $w$ such that \pone has a winning strategy on $w$ in $G$ if and only if $M$ halts on an empty input. Without loss of generality, we assume that $M$ always moves its head to the left-most used tape cell before halting.
	
	The basic idea behind the reduction is encoding configurations of $M$ by strings nested below a root $r \notin \Sigma$ and using \Call moves by \pone to $\cl{r}$ to simulate moves of $M$. The input string is $w = \op{r}\nw{v}_0\cl{r}$, where $\nw{v}_0$ represents the initial configuration of $M$, and any time \pone plays \Call on $\cl{r}$ in some string $\op{r}\nw{v}_i\cl{r}$ (with $\nw{v}_i$ representing some configuration of $M$), that string gets replaced by $\op{r}\nw{v}_{i+1}\cl{r}$, where $\nw{v}_{i+1}$ represents the successor configuration of the one represented by $\nw{v}_i$. The target language of $G$ is constructed to contain all strings of the form $\op{r}\nw{v}_h\cl{r}$, where $\nw{v}_h$ represents a halting configuration of $M$. This way, \pone has a winning strategy of \Call depth $k$ on the input string if and only if $M$ halts on the empty input within at most $k$ steps.
	
	We represent configurations of $M$ in the same way as in the proof of Theorem \ref{theo:unboundednotre}: a flat string $(x_1,-) \cdots (x_{k-1},-) (x_k,q) (x_{k+1},-)\cdots (x_m,-)$ over the alphabet $\Sigma \times (Q \cup \{-\})$ denotes that the content of $M$'s working tape is $x_1\cdots x_m$, with the head of $M$ being on the tape's $k$-th cell and $M$ being in state $q$, and these flat strings are represented as nested strings using the standard nested string representation. Again, by our assumption on the shape of $M$'s halting configurations, the flat string encoding of a halting configuration is of the form $(\Sigma \times \{h\}) (\Sigma \times \{-\})^*$ for the halting state $h$ of $M$.
	
	From $M$, we can easily construct a $\epsilon$-free NWT implementing $M$'s transition function on configurations represented in this way.\footnote{Note that this $\epsilon$-free NWT is functional, so $G$ is basically a ``solitaire'' game for \pone, where \ptwo does not get to make any choices. In fact, the replacement relation can even be implemented using a deterministic finite-state transducer that inserts at most two symbols with each transduction.} A DNWA accepting all strings $\op{r}v\cl{r}$ in which $v$ represents a halting configuration for $M$ is similarly easy to construct. Finally, it is clear that \pone has a winning strategy on the input string $w$ using at most $k$ \Call moves if and only if $M$ reaches a halting configuration from its initial configuration within at most $k$ steps, which completes the proof.
\end{proof}

Throughout the rest of this section, let $\Exp(k,n)$ be the $k$-fold exponential tower function in $n$, defined recursively by $\Exp(0,n) = n$ and $\Exp(k,n) = 2^{\Exp(k-1,n)}$ for all integers $k>0$ and $n \geq 0$.

\begin{lemma}\label{lemma:nonelementarysize}
	An input string of length $2 \cdot (n+2k-1)$ can be transformed into a string of length $2 \cdot \Exp(k,n)$ by a game of \Call depth $2$ with deterministic $\epsilon$-free NWT replacement.
\end{lemma}

\begin{proof}
	Choose as input a string of the form $\op{k-1}\op{k-2} \cdots \op{1} \op{c_0}^n \op{c_1} \cdots \op{c_k} \cl{c_k} \cdots \cl{k-1}$. The tree represented by this string is obviously a path of length $n+2k-2$. Play proceeds in $k$ \emph{rounds} as follows: In round $i \in [k]$, \pone plays \Call on each node labelled $c_{i-1}$ in bottom-up (i.e. left-to-right) order. Each such \Call move deletes the called $c_{i-1}$ node and doubles the number of $c_{i}$ nodes below it (i.e. replaces each $\op{c_{i}}$ by $\op{c_{i}}\op{c_{i}}$ and $\cl{c_{i}}$ by $\cl{c_{i}}\cl{c_{i}}$). Afterwards, if $i < k$, \pone plays \Call on the node labelled $i$, which deletes that node, attaches its child path to its $i+1$-labelled parent and allows \pone to play again on its child path, since this is the first time \pone has played \Call on the node labelled $i$. By an induction argument, it is easy to show that at the conclusion of round $i$, the current string contains exactly $2 \cdot \Exp(i,n)$ tags labelled $c_i$ (half of them opening, the other half closing tags).
\end{proof}

Lemma \ref{lemma:nonelementarysize} is an integral part in proving the following lower bound.

\begin{proposition}\label{prop:coknexptimelower}
	For each $k \geq 1$, it holds that for the class of games with $\epsilon$-free NWT, and \Call depth bounded by $2$,  $\Safelr$ is hard for $\cokNEXPTIME$.
\end{proposition}

\begin{proof}
	Let $k \geq 1$. We show $\cokNEXPTIME$-hardness by reduction from the complement of the $\kNEXPTIME$-complete problem \algprobname{$k$-ExpTiling} of, given a tile set $U$, vertical and horizontal constraints $V$, $H$, initial and final tile $u_i, u_f$ and unary number $n$, determining whether there exists a valid tiling of height $\Exp(k,n)$ and width $\Exp(k,n)-1$.
	
	We construct from an instance of \algprobname{$k$-ExpTiling} a game $G$ and a string $w$ such that \pone has a winning strategy in $G$ starting at $w$ if and only if there exists \emph{no} valid tiling of the given size. We first give a rough overview over the game before describing the construction of $G$ and $w$ in detail.
	
	Play proceeds in three phases:
	\begin{itemize}
	\item In the first phase, the construction from Lemma \ref{lemma:nonelementarysize} is used to transform the polynomial-sized input string into a \emph{seed string} of size $\bigO(\Exp(k-1,n))$. This phase is deterministic in the sense that we will design the replacement transducer and target language in a way such that neither \pone nor \ptwo get to make any choices in phase 1. 
	\item In phase 2, \ptwo constructs from this seed string a \emph{tiling candidate}, i.e. a (representation of a) potential tiling of size $\Exp(k,n) \times (\Exp(k,n)-1)$. In this phase, \pone still doesn't get to make any choices; all she is supposed to do is to \Call certain nodes in order to allow \ptwo to construct the tiling candidate. By fixing the replacement transducer in the proper way, we can ensure that this tiling candidate has no horizontal errors and starts with the initial tile.
	\item In the final phase, \pone tries to show that the tiling candidate \ptwo constructed in phase 2 does not represent a valid tiling. To this end, she repeatedly forces \ptwo to delete either all even or all odd columns in his tiling candidate until only a single column is left; this column is then checked for vertical errors by the target language automaton.
	\end{itemize}

	Note that, during all of the three phases, we require \pone to stick to a certain ``game plan'', in which she only calls certain nodes in a pre-determined order. This can be enforced by constructing the replacement transducer and target language automaton in an appropriate way. We will first examine phase-by-phase how the game $G$ is constructed, assuming \pones compliance, and later describe how $G$ has to be modified to prevent deviations from the game plan.
	
	Let $v_{k,n}$ be the sequence of opening tags from the proof of Lemma \ref{lemma:nonelementarysize}, i.e. $v_{k,n} = \op{k-1}\op{k-2} \cdots \op{1} \op{c_0}^n \op{c_1} \cdots \op{c_k} $, and let $\overline{v_{n,k}}$ be the complementary sequence of closing tags. The input string $w$, then, is of the form 
	$$w = \op{\text{start2}}\op{\text{mv}} \op{\text{dbl}} \op{\text{cp1}} v_{k-1,n} \op{r}\op{p}\cl{p}\cl{r} \overline{v_{k-1,n}} \cl{\text{cp1}} \cl{\text{dbl}} \cl{\text{mv}} \cl{\text{start2}}.$$ The purpose of each tag will be explained when it is first used in the game.
	
	At the start of phase 1, \pone uses the procedure described in the proof of Lemma \ref{lemma:nonelementarysize} to rewrite $w$ into the following string of $(k-1)$-fold exponential size:
		$$\op{\text{start2}}\op{\text{mv}} \op{\text{dbl}} \op{\text{cp1}} \op{c_{k-1}}^\ell \op{r}\op{p}\cl{p}\cl{r} \cl{c_{k-1}}^\ell \cl{\text{cp1}} \cl{\text{dbl}} \cl{\text{mv}} \cl{\text{start2}},$$
		where $\ell = \Exp(k-1,n)$. Afterwards, \pone plays her first \Call on the node labelled cp1 (where ``cp'' stands for ``copy''), allowing her a replay on $\op{c_{k-1}}^e \op{r}\cl{r} \cl{c_{k-1}}^e$. In this replay, she calls every $c_{k-1}$ in left-to-right order; each such \Call replaces the called $c_{k-1}$ by $c'_{k-1}$, replaces $\op{r}$ by $\op{r}\op{d}\op{e}$ and replaces $\cl{r}$ by $\cl{e}\cl{d}\cl{r}$. After this rewriting, the $r$-labelled node has a child path consisting of $\ell$ alternating $d$- and $e$-nodes.
		
	Next, \pone plays \Call on the node labelled dbl (for ``double''), which rewrites the $c'_{k-1}$-labelled path into a $c'_{k-1}$-labelled path of double length, i.e. the rewriting transducer replaces each $\op{c'_{k-1}}$ by $\op{c'_{k-1}}\op{c'_{k-1}}$ and each $\cl{c'_{k-1}}$ by $\cl{c'_{k-1}}\cl{c'_{k-1}}$. With her next call to mv (for ``move''), \pone gets another replay on this path of length $2 \ell$, calling each $c'_{k-1}$ node in left-to-right order, which causes \ptwo to delete that node and insert a single $c$ (for ``create'') node  below the bottom $p$-labelled node. Finally, \pone plays \Call on the root node labelled start2 to start the second phase of the game; as his reply to this \Call, \ptwo deletes the start2-node and inserts a node labelled $u_i$ as a leaf at the very bottom of the path. This ends phase 1 and leaves \pone to play a replay-free game on the string 
		$$w_1 = \op{r} (\op{d} \op{e})^\ell \op{p} \op{c}^{2 \ell} \op{u_i}\cl{u_i} \cl{c}^{2\ell} \cl{p} (\cl{e}\cl{d})^\ell \cl{r}.$$
		
	In the second phase, \ptwo is supposed to set up a tiling candidate. This candidate will be encoded within the leaves of the current string, i.e. in siblings of the initial $u_i$-labelled leaf. To this end, \pone calls each of the $c$-nodes in bottom-up order. On each such call, \ptwo deletes the called node and doubles the number of leaves in the current string. More precisely, the replacement transducer may replace each $\cl{u}$ (for some  $u \in U$) either by $\cl{u}\op{u'}\cl{u'}$ for some $u' \in U$ with $(u, u') \in H$, or by $\cl{u}\op{\#}\cl{\#}$ for a divider symbol $\# \notin U$. Since each \Call to a $c$-node doubles the length of the current tiling candidate, it is clear that when all $c$-nodes have been called (and phase 2 ends), the current string is of the form
	$$w_2 = \op{r} (\op{d} \op{e})^\ell \op{p} v_t \cl{p} (\cl{e}\cl{d})^\ell \cl{r},$$
	where $v_t$, the final tiling candidate, is a string of nesting depth zero consisting a total of $2^{2\ell} = (2^\ell)^2 = \Exp(k,n)^2$ pairs of corresponding opening and closing tags with labels from $U \cup \{\#\}$, beginning with $\op{u_i}\cl{u_i}$.
		
	We interpret the tiling candidate $v_t$ as the concatenation of $\Exp(k,n)$ lines of length $\Exp(k,n)$ each. To encode a valid tiling of size $\Exp(k,n) \times (\Exp(k,n)-1)$, we expect $v_t$ to be of the form $((\op{U}\cl{U})^{\Exp(k,n)-1}\#)^{\Exp(k,n)}$, i.e. we expect \ptwo to use the symbol $\#$ as a line separator only. 
		
	It should be clear from the construction that the string $v_t$ represents a concatenation of several (possibly empty) horizontally correct substrings of tiles, separated by $\#$, the first of which starts with $u_i$. The tiling candidate encoded by $v_t$ may, however, still contain one or more of the following types of errors:
	\begin{itemize}
		\item \emph{Vertical error}: Two vertically adjacent tiles $u, u'$ with $(u,u') \notin V$;
		\item \emph{Incorrect line lengths}: Strictly more or less than $\Exp(k,n)-1$ symbols from $U$ between two subsequent $\#$;
		\item \emph{Incorrect final tile}: The last symbol from $U$ in $v_t$ is not $u_f$.
	\end{itemize}

	The main observation needed for this part of the reduction is the following: if a tiling candidate does not represent a correct tiling, then at least one of these errors can be found by examining just a single column of the tiling candidate. If a tiling candidate contains a vertical error, then there is a column containing two subsequent tiles $u,u'$ with $(u,u') \notin V$; if some line length is incorrect, then the last column contains some symbol other than $\#$, or there is a $\#$ in a column that is \emph{not} the last column; and finally, if the tiling candidate does not end with $u_f$, then the next-to-last column (i.e. column number $\Exp(k,n)-1$) does not end with $u_f$. It is easy to see that, once a single column of $v_t$ has been isolated, all three of these conditions can easily be checked using a polynomial-sized DNWA.
	
	\pones task on the string $w_2$ therefore consists of isolating a single column containing an error. If and only if she manages to do so, she wins the game (i.e. the target DNWA checking for the existence of one of the above kinds of errors accepts). 
	
	To isolate a single column, \pone plays \Call moves on all nodes labelled $d$ (for ``destroy'') in bottom-up order. Each such \Call removes the called $d$-node and its child, and forces the replacement transducer to either delete all even-numbered columns or all odd-numbered columns by deleting every other node in $v_t$. \pone makes the choice of whether to delete all even-numbered or all odd-numbered columns by playing either \Read or \Call on the $e$-labelled node just below the $d$-node she is to call next; a \Read move leaved the label $e$ (``even'') intact, causing the replacement transducer to delete all even-numbered columns on \pones \Call to the $d$-node above, while a \Call move relabels $e$ into $o$ (``odd''), causing all odd-numbered columns to be deleted analogously. Each such deletion step halves the number of remaining columns, which means that after $\ell$ deletion moves, only a single column of length $\Exp(k,n)$ remains of $v_t$.
	
	From the above considerations concerning error types, it is clear that some additional information needs to be tracked through the deletion phase, to determine whether the column chosen by \pone is the last, next-to-last, or some other column. This is the purpose of the node labelled $p$ (``position''), which is rewritten depending on the sequence of \pones choices of even and odd columns:
	
	\begin{itemize}
		\item To reach the last column, \pone must successively remove only odd-numbered columns. Therefore, on the first \Call to a $d$-node with $o$-child (and $p$-grandchild), $p$ is rewritten into $l$ (``last''). Any \Call to a $d$-node with $e$-child and $l$-grandchild rewrites $l$ into $s$ (``standard'', i.e. the column to be checked is not a special case), while calling a $d$-node with $o$-child leaves the label $l$ as is.
		\item To reach the next-to-last column, \pone must remove all even-numbered columns in the first step and successively remove only odd-numbered columns after that. Therefore, on the first \Call to a $d$-node with $e$-child (and $p$-grandchild), $p$ is rewritten into $x$ (``neXt-to-last''). Any \Call to a $d$-node with $o$-child and $x$-grandchild leaves the label $x$ intact, while a \Call to a $d$-node with $e$-child and $x$-grandchild relabels $x$ to $s$.
	\end{itemize}

	Using the label of the rewritten $p$-node as an indicator, a polynomial-sized DNWA can now easily check whether a column chosen by \pone contains some error. Assuming that \pone is restricted to strategies that follow the described order of \Call moves, it is straightforward (if tedious) to prove that \pone has a winning strategy if and only if there does not exist a valid tiling of size $\Exp(k,n) \times (\Exp(k,n)-1)$. If such a tiling exists, \ptwo can give its encoding as $v_t$ and deny \pone the opportunity to find any errors, no matter which column she isolates, and if no such tiling exists, then any tiling candidate given by \ptwo necessarily contains at least one error, which \pone can then point out.
	
	Finally, we examine how \pone may be restricted to play only according to the game plan described above. The basic idea behind this is to construct the replacement transducer and target DNWA in such a way that any deviation from the game plan causes \pone to immediately and irrevocably lose the game, thus ensuring that any winning strategy, should one exist, sticks to the game plan.
	
	The construction of the target language so far already ensures that \pone does not leave any undesired uncalled nodes behind, i.e. since target strings may only have tags from $U \cup \{\#, r, l, x, s\}$, \pone loses automatically if, for instance, there are any uncalled $d$-nodes left behind in the final string.
	
	The only thing that still needs to be ensured is that \pone does not skip any calls. For instance, we could imagine \pone trying to cheat in phase 1 by playing \Call on some $i$-labelled node without having called all nodes labelled $c_{i-1}$ below it, or in phase 3 by calling some $d$-labelled node while there are still uncalled $c$-labelled nodes. There are numerous further situations like these, but they are all handled in the same way: on any \Call to a function symbol, the nested substrings supposed to be rewritten due to that \Call are required to be of a specific form; for instance, on a \Call to some $i$-labelled node in phase 1, the substring nested below it has to be of the form 
	$$\op{c_i}^* \op{c_{i+1}} \cdots \op{c_k} \op{r}\op{p}\cl{p}\cl{r} \cl{c_k} \cdots \cl{c_{i+1}} \cl{c_i}^*.$$
	If the substring below the called function node is not of the desired form (for instance due to remaining $c_{i-1}$-nodes in the previous example), the replacement transducer constructed according to the ideas laid out above will not have an accepting run on this incorrect substring, therefore reject that substring, and, according to the semantics of NWT games, \ptwo will win the game immediately. This shows that we can safely assume \pone to be restricted to the game plan laid out above.
	
	Finally, it is relatively easy (but again rather tedious) to prove that the replacement transducer described above can be constructed to be of polynomial size in $|U|$, $n$ and $k$.
\end{proof}

\begin{restate}{Theorem \ref{thm:boundednonelementary}}
	For the class of games with $\epsilon$-free NWT and \Call depth bounded by $d \geq 2$, $\Safelr$ is decidable, but not decidable in elementary time.
\end{restate}

\begin{proof}
	Decidability follows from the fact that, due to the restriction to $\epsilon$-free NWT, there are only finitely many possible replacements for each substring that \pone plays \Call on. This, combined with the finite \Call depth, means that all strategies for \pone and all possible plays for each strategy can be enumerated in finite time.
	
	The non-elementary lower bound on complexity is implied by Proposition \ref{prop:coknexptimelower}: Assume that, for some $d \geq 2$, there exists a $k$ such that $\Safelr(\calG_d)$ is decidable in $k$-fold exponential time, where $\calG_d$ denotes the class of games with $\epsilon$-free NWT and \Call depth $d$. It then follows (by a trivial reduction) that $\Safelr(\calG_2)$ is also in $\kEXPTIME$, and therefore in $\cokNEXPTIME$. However, by Proposition \ref{prop:coknexptimelower}, $\Safelr(\calG_2)$ is hard for $\cokpNEXPTIME$, which yields a contradiction to the nondeterministic time hierarchy theorem \cite{Cook73}.
\end{proof}

\begin{restate}{Theorem \ref{thm:conexptime}}
	For the class of replay-free games with $\epsilon$-free NWT, $\Safelr$ is $\coNEXPTIME$-complete.
\end{restate}

\begin{proof}
	The lower bound 
	follows as in the proof for Proposition \ref{prop:coknexptimelower}, omitting the first phase of the game constructed there and setting $k=1$ in the second and third phase.
	
	The co-nondeterministic exponential-time algorithm yielding a matching upper bound is conceptually very straightforward: It moves through the input string from left to right, recursively trying out all possible strategy decisions for \pone while guessing universally \ptwos strategy decisions.
	
	To formalise this algorithm, we use the shorthand notation $[u,v]$ for game positions, where $uv \in \wf$, with $u$ denoting the substring that has already been processed and $v$ the substring that is yet to be played on (including the closing tag on which \pone is to move next). For a string $u\cl{a} \in \hat{\Sigma}^*$, let $\Last(u\cl{a})$ denote the (unique) rooted substring of $u\cl{a}$ ending at $\cl{a}$. The following algorithm \textsc{CheckWin}$(G,[u,v])$ then tests whether \pone has a replay-free winning strategy in game $G = (\Sigma, \funcsymb, R, T)$ starting at position $[u,v]$.

 	\begin{algorithm}[h]
   \caption{\textsc{CheckWin}$(G,[u,v])$}
   \label{alg:checkwin}
   \begin{algorithmic}[1]
 	\IF{$v = \epsilon$}
 		\IF{$u \in T$}
 			\STATE Accept
 		\ELSE
 			\STATE Reject
 		\ENDIF
 	\ENDIF
 	\IF{$v = tv'$ for $t \in \op{\Sigma} \cup (\cl{\Sigma} \setminus \cl{\funcsymb})$}
 		\STATE // \pone may not make a strategy choice on $t$; move to the right.
 		\STATE Return \textsc{CheckWin}$(G, [ut,v])$
 	\ENDIF
 	\IF{$v = \cl{f}v'$ for $f \in \funcsymb$}
 		\STATE // Try out both strategy options for \pone{}; \Read first.
 		\IF{\textsc{CheckWin}$(G, [u\cl{f},v])$ accepts}
 			\STATE Accept
 		\ELSE
 			\STATE Guess universally a transduct $u_t \in R(\Last(u\cl{f}))$
 			\STATE $u' \gets u\cl{f}$ with $\Last(u\cl{f})$ replaced by $u_t$
 			\IF{\textsc{CheckWin}$(G, [u',v'])$ accepts}
 				\STATE Accept
 			\ELSE
 				\STATE Reject
 			\ENDIF
 		\ENDIF
 	\ENDIF
   \end{algorithmic}
 \end{algorithm}	

	As the algorithm \textsc{CheckWin}$(G, [ut,v])$ directly mimics the gameplay according to $G$ from position $[u,v]$, a simple induction argument suffices to prove that \textsc{CheckWin}$(G, [ut,v])$ accepts if and only if \pone has a replay-free winning strategy in $G$ from $[u,v]$; therefore, \textsc{CheckWin}$(G, [\epsilon,w])$ accepts if and only if $w \in \safelr[1](G)$. It remains to be shown that \textsc{CheckWin} indeed runs in co-nondeterministic exponential time.
	
	To that end, we first examine the maximum size of positions used as inputs for recursive calls of \textsc{CheckWin} (where the size of a position $[u,v]$ is defined as $|uv|$). On some input $[u,v]$, \textsc{CheckWin} can only increase the size of $[u,v]$ through the transduction in lines 14 and 15. Since $R$ is a NWT without $\epsilon$-transitions, any of its outputs on some input $x$ may only have size $c \cdot |x|$ for some constant $c$ depending only on $R$. This in turn means that the size of positions may only increase by a factor of $c$ for each recursive call to \textsc{CheckWin}.
	
	Obviously, as each recursive call to \textsc{CheckWin} removes one symbol from the right side of a position and insertions only occur on the left side of positions, the recursion depth of \textsc{CheckWin} on an input position $[\epsilon,w]$ is at most $w$. This in turn means that input positions for recursive calls to \textsc{CheckWin} may be of size at most $|w| \cdot c^{|w|}$.
	
	Finally, since each recursive call branches into at most two further calls of \textsc{CheckWin}, the algorithm's recursion tree has at most exponentially many nodes in the size of $w$. This implies that the membership test in line 2 and the co-nondeterministic choice in line 14 are executed at most exponentially many times on at most exponentially long strings. Since membership testing both for NWA and for NWT is in $\PTIME$, this yields an exponential upper bound on the running time of \textsc{CheckWin}.
\end{proof}

\begin{restate}{Theorem \ref{thm:boundedcallwidth}}
	For the class of games with $\epsilon$-free NWT, \Call depth bounded by $d \geq 1$ and \Call width bounded by $k \geq 1$, $\Safelr$ is $\coNEXPTIME$-complete.
\end{restate}

\begin{proof}
	As replay-free games always have bounded \Call width, the lower bound carries over directly from Theorem \ref{thm:conexptime}.
	
	The algorithm for the upper bound is almost the same as the one for Theorem \ref{thm:conexptime}, simply going through the input string from left to right, trying out all possible strategies for \pone and guessing co-nondeterministically replacement strings for \ptwo. The only difference is that, in this case, it also needs to track how much of the allowed \Call width and depth in the current substring has been used up already.
	
	As in the proof of Theorem \ref{thm:conexptime}, the correctness of this algorithm is obvious as it basically just simulates gameplay.
	
	To prove that the running time and amount of non-determinism required by the algorithm is at most exponential, we again examine the maximum size of occurring strings and of the decision tree for \pone for any fixed counter-strategy of \ptwo, i.e. the recursion tree in any run of the algorithm.
	
	For string sizes, we again note that, for any replacement NWT $R$ and string $w$, all strings in $R(w)$ are of size at most $s\cdot |w|$ for some constant $s$ only depending on $R$. Furthermore, any replacement substring $v$ resulting from a \Call move of depth $r<d$ allows for at most $k$ further \Call moves of depth $r+1$ (each of which may be further played on with \Call width $k$ if $r+1 < d$). Each of these depth-$(r+1)$ substrings may be of size at most $s^k \cdot |v|$ (in case all $k$ allowed calls go into re-transducing $v$ in some fashion). In total, this means that increasing \Call depth by 1 increases the size of occurring strings by a multiplicative factor exponential in $k$. Therefore, with \Call depth $d$ and \Call width $k$, the final string after a play on some input string $w$ is at most of size $|w| \cdot 2^{\bigO(dk)}$, and this size is also an upper bound on the size of each single transduct (and therefore also on the amount of non-determinism required any time a transduct is guessed).
	
	Since each \Call move of depth $r$ enables at most $k$ further calls of depth $r+1$, and there are at most $|w|$ possible \Call positions of depth 0 in any input string $w$, \pone may play at most $c_{\text{max}} \mydef |w| \cdot k^d$ \Call moves in any play on $w$.
	
	To reach an upper bound on the number of positions \pone has available to place calls on during an entire play, we use the above upper bound of $|w| \cdot 2^{\bigO(dk)}$ on the size of working strings and assume (as a generous estimate) that the \emph{entire} string is replaced by a different string of length $|w| \cdot 2^{\bigO(dk)}$ every time \pone plays \Call on some position. Then, the total number $\ell_{\text{max}}$ of positions \pone can place calls on is bounded by the total number of positions that ever become available during the entire game, which is at most the maximum string length times the maximum number of available \Call moves, i.e. $\ell_{\text{max}} = (|w| \cdot 2^{\bigO(dk)}) \cdot (|w| \cdot k^d) = |w| \cdot 2^{\bigO(dk)}$.
	
	We can now give an upper bound on the size of the decision tree for \pone for a given counter-strategy by \ptwo. Assume nodes of this tree to be labelled with positions that \pone has to make a strategy decision on, and outgoing edges to be labelled either \Read or \Call. Then, each root-to-leaf path has length equal to at most $\ell_{\text{max}}$ and contains at most $c_{\text{max}}$ edges labelled \Call, with all other edges being labelled \Read. Since the sequence of \Read{}- and \Call{}-edges taken from the root forms a unique address in this tree, the number of different paths in the tree is at most the number of strings of length $\ell_{\text{max}}$ over $\{\Read, \Call\}$ containing at most $c_{\text{max}}$ occurrences of \Call, which is in $\ell_{\text{max}}^{\bigO(c_{\text{max}})}$. Together with $\ell_{\text{max}}$ being the maximum length of paths, this implies that the decision tree has at most $\ell_{\text{max}}^{\bigO(c_{\text{max}})} = |w| \cdot 2^{\bigO(|w| \cdot dk^{d+1})}$ nodes. 
	
	Since $d$ and $k$ are fixed, this means that the size of the recursion tree for the co-nondeterministic algorithm simulating play is at most exponential, and since each node of this recursion tree requires at most exponential computation time and non-determinism, the algorithm indeed runs in co-nondeterministic exponential time, as was to be proven.
\end{proof}

\begin{restate}{Theorem \ref{thm:boundedinputcallwidth}}
	For the class of games with $\epsilon$-free NWT, \Call depth bounded by $d$ and \Call width including input bounded by $k$, $\Safelr$ is 
	\begin{enumerate}[(a)]
	\item $\coNP$-complete for $d \geq 1$ and $k \geq 2$,
	\item $\coNP$-complete for $d \geq 2$ and $k \geq 1$, and
	\item in $\PTIME$ for $d=k=1$.
	\end{enumerate}
\end{restate}

\begin{proof}
	The $\coNP$ upper bound uses almost exactly the same algorithm as Theorem \ref{thm:boundedcallwidth}, going through the input string from left to right, trying out all possible strategies for \pone, guessing co-nondeterministically replacement strings for \ptwo and tracking \Call depth and width along the way. The only difference to that algorithm lies in the fact that \Call width is also tracked for the input string, which changes the analysis from the proof of Theorem \ref{thm:boundedcallwidth} as follows.
	
	With fixed bounds $d$ on \Call depth and $k$ on \Call width, the maximum number of \Call moves that \pone can play is $c_{\text{max}} = k^d$, which is constant. With each \Call move increasing the size of the current string by at most a multiplicative constant $s$ (depending only on the replacement transducer), this means that the maximum size of strings is $\ell_{\text{max}} = |w| \cdot s^{c_{\text{max}}}$ for any input string $w$, which is linear in $|w|$. In analogy to the arguments from the proof of Theorem \ref{thm:boundedcallwidth}, this means that \pone may perform calls on at most $c_{\text{max}}$ out of $c_{\text{max}} \cdot \ell_{\text{max}} = \bigO(|w|)$ positions, implying that the decision tree has at most $\ell_{\text{max}}^{\bigO(c_{\text{max}})}$ paths of length $\ell_{\text{max}}$ each, and therefore size $\bigO(|w|^c)$ for some constant $c$. Since all strings that need to be processed or nondeterministically guessed at each node of the recursion tree are of polynomial length, this implies a co-nondeterministic polynomial time complexity.
	
	The lower bounds in (a) and (b) follow by reductions from the problem $\coiiiSAT$: Given a propositional formula $\varphi$ in conjunctive normal form with three literals per clause, is $\varphi$ unsatisfiable? This problem is the complement of the well-known $\NP$-complete problem $\algprobname{3-SAT}$, and therefore complete for $\coNP$.
	
	The reductions work as follows. Let $\varphi = C_1 \land \ldots \land C_m$ be a 3-CNF formula over variables $x_1, \ldots, x_n$ with clauses $C_1, \ldots, C_m$. We construct from $\varphi$ a  game $G = (\Sigma, \funcsymb, R, T)$ and input string $w$ such that \ptwo has a winning strategy on $w$ in $G$ if and only if $\varphi$ is satisfiable; the basic idea behind the reduction is that \pones first \Call to the input string allows \ptwo to pick a variable assignment $\alpha$ for $\varphi$, while \pones second \Call is supposed to mark a clause that is not satisfied by $\alpha$ (if one exists). In this way, \pone has a winning strategy if and only if every possible assignment that \ptwo can chose does not satisfy all of the clauses in $\varphi$. Using this idea, \pone needs to perform exactly two \Call moves; where and how these calls should be performed depends on whether the game is replay-free with \Call width 2 (for (a)) or has \Call depth 2 and \Call width 1 (for (b)).
	
	For both reductions, the input string $w$ is of the form 
	$$w = \op{r}\op{C_1} \cdots \op{C_m}\op{V}\op{x_1}\cl{x_1} \cdots \op{x_n}\cl{x_n}\cl{V}\cl{C_m} \cdots \cl{C_1}\cl{r},$$
	i.e. its tree representation is a path with labels $r, C_1, \ldots, C_m, V$ (in top-down order), with the $V$-labelled node having as children leaves labelled $x_1, \ldots, x_n$.
	
	For the reduction in (a), the set of function symbols are $\funcsymb = \{V, C_1, \ldots, C_m\}$. \pones first \Call is supposed to be on $\cl{V}$, which relabels $V$ to $V'$ and each $x_i$ to either $0_i$ or $1_i$ (for $i \in [n]$), at \ptwos choice; afterwards, she calls exactly one $\cl{C_j}$ (with $j \in [m]$), which relabels the called node to $\hat{C}_j$ and leaves the rest of the string unchanged. If \pone plays \Call on some $\cl{C_j}$ without having called $\cl{V}$ first, the replacement transducer rejects its input string and \pone loses the game immediately. The target language accepts all strings containing exactly one node labelled with a symbol $\hat{C}_j \in \{\hat{C}_1, \ldots, \hat{C}_m\}$ such that the clause $C_j$ is not satisfied by the variable assignment encoded below $V$.
	
	For (b), the set of function symbols is $\funcsymb = \{r, C'_1, \ldots C'_m\}$. \pones first \Call (the only one she can perform in the input string) is to $\cl{r}$, which relabels $r$ to $r'$, each $C_j$ to $C'_j$ (for $j \in [m]$) and each $x_i$ to either $0_i$ or $1_i$ (for $i \in [n]$), at \ptwos choice; afterwards, \pone is supposed to call exactly one $\cl{C'_j}$, relabelling it to $\hat{C}_j$. The target language accepts all strings rooted with $r'$ that contain exactly one node labelled with a symbol $\hat{C}_j \in \{\hat{C}_1, \ldots, \hat{C}_m\}$ such that the clause $C_j$ is not satisfied by the variable assignment encoded below $V$.
	
	It is easy to see that both reductions can be computed in polynomial time and that in both cases, \pone has a winning strategy in $G$ if and only if $\varphi$ is not satisfiable.
	
	The $\PTIME$ upper bound in (c) is quite simple. Let $G=(\Sigma,\funcsymb,R,T)$ be a game and $w$ an input string; we now need to find out whether exactly one of the function symbols occurring in $w$ can be called to yield a string in the target language. For any fixed occurrence of a function symbol in $w$, it is easy to construct a NWT $R'$ that simulates $R$ on the substring rooted at the chosen function symbol and otherwise leaves $w$ unchanged. A type checking test against the target DNWA can then determine in polynomial time (cf. Theorem \ref{thm:nwttypechecking} (b)) whether calling that function symbol leads to a string in the target language. Since there are less than $|w|$ occurrences of function symbols in $w$, polynomially many such type checking tests suffice to determine whether \pone has a winning strategy in $G$ on $w$ that plays at most one \Call move.\end{proof}

\newpage
		\section*{Proofs for Section \ref{sec:relabel}}

\begin{restate}{Theorem \ref{thm:relabelexptime}}
	For the class of games with relabelling transducers and unbounded replay, $\Safelr$ is $\EXPTIME$-complete.
\end{restate}

\begin{proof}
	The upper bound uses a trivial alternating polynomial-space algorithm that moves through the input string in a left-to-right order (resetting its focus as necessary after \Call moves), guesses existentially for each closing tag a move for \pone and, in case of a \Call, guesses universally a relabelling chosen by \ptwo and applies it to the input string. As relabellings are generally of linear size, they can easily be guessed on polynomial space, and the verification whether a guessed relabelling is indeed consistent with the replacement transducer is feasible in polynomial time by Theorem \ref{thm:nwtmembership}.
	
	The lower bound follows by a reduction from the membership problem for linear bounded alternating Turing machines, which is complete for $\APSPACE = \EXPTIME$ \cite{ChandraKS81}.
	
	Let $M$ be a linear bounded ATM with state set $Q$ and working alphabet $\Sigma$ (i.e. on input $w$, $M$ uses at most $|w|$ tape cells in any computation). Assume without loss of generality that any non-halting state of $M$ is either existential or universal, that the transition relation for $M$ has exactly two transitions for each state and tape symbol (i.e. each non-halting configuration of $M$ has exactly two successor configurations), that the initial state of $M$ is universal and that $M$ always moves its head to the left-most tape position before halting.
	
	We construct from $M$ and a given input string $w$ a game $G$ and input string $w'$ such that \pone has a winning strategy on $w'$ in $G$ if and only if $M$ accepts $w$. To this end, we use a similar technique as for Theorem \ref{thm:boundedundecidable}: We represent configurations of $M$ by strings rooted at a function symbol $r \notin \Sigma$ and have each \Call by \pone to $\cl{r}$ initiate a transition of $M$, updating $M$'s configuration by the replacement transduction.
	
	The main difference to the proof of Theorem \ref{thm:boundedundecidable} is that we have to take alternation into account. We simulate the alternation in $M$ by strategy choices of \pone and \ptwo, with \pone choosing existential and \ptwo choosing universal transitions. Universal choice can simply be encoded into the replacement transducer, such that when \pone initiates a transition of $M$, \ptwo chooses which of the two possible successor configurations to rewrite the current string to. A little more care has to be taken with existential choice, as \pone may not select any rewritings but can only choose whether or not a substring should be rewritten.
	
	To allow \pone to choose between transitions, we extend the nested word representation of $M$'s current configuration by a special \emph{flag substring}, which may be either $\op{0}\cl{0}$ or $\op{1}\cl{1}$ or $\op{2}\cl{2}$. This flag should be $0$ if the current configuration is universal or halting, and initially set to $1$ once an existential configuration is reached. If it is $1$, \pone has the option to have \ptwo rewrite it to $2$ by way of a \Call move, or leave it as is with a \Read move. Once \pone then initiates a transition of $M$, the flag indicates which of the two possible successor configuration is reached by the rewriting, i.e. \ptwo does not get any choice but rewrites the current string based on the flag's value.
	
	Once again, we represent configurations of $M$ as in the proof of Theorem \ref{theo:unboundednotre}: a flat string $(x_1,-) \cdots (x_{k-i},-) (x_k,q) (x_{k+1},-)\cdots (x_m,-)$ over the alphabet $\Sigma \times (Q \cup \{-\})$ denotes that the content of $M$'s working tape is $x_1\cdots x_m$, with the head of $M$ being on the tape's $k$-th cell and $M$ being in state $q$, and these flat strings are represented as nested strings using the standard nested string encoding. By our assumption on the shape of $M$'s halting configurations, the flat string encoding of an accepting configuration is of the form $(\Sigma \times \{q_+\}) (\Sigma \times \{-\})^*$ for the accepting state $q_+$ of $M$.
	
	The input string $w'$ for the game $G$ is constructed as $w' = \op{r}\op{0}\cl{0}\nw{v_0}\cl{r}$, where $\nw{v_0}$ is the nested string representation of $M$'s initial configuration with input $w$ (which is universal by our assumption above, hence the $0$ flag). The game $G$ is over the alphabet $(\Sigma \times (Q \cup \{-\})) \cup \{r,0,1,2\}$ with function symbols $\funcsymb = \{r,1\}$.	
	
	The replacement transducer simulates a transition of $M$ at a \Call on $\cl{r}$ as described above. Such a transducer can be computed from $M$'s transition relation in polynomial time. It is also easy to construct in polynomial time a target DNWA accepting all strings of the form $\op{r}\op{0}\cl{0}v\cl{r}$ where $v$ represents an accepting configuration of $M$. It follows from the above considerations that \pone has a winning strategy on $w'$ in $G$ if and only if $M$ accepts $w$.
\end{proof}

\begin{restate}{Theorem \ref{thm:relabelpspace}}
	For any $k \geq 1$, for the class of games with relabelling transducers and bounded \Call depth $k$, $\Safelr$ is $\PSPACE$-complete.
\end{restate}

\begin{proof}
	The upper bound again uses the trivial alternating algorithm that moves through the input string in a left-to-right order, guesses existentially for each closing tag a move for \pone and, in case of a \Call, guesses universally a relabelling chosen by \ptwo and applies it to the input string. As in the proof of Theorem \ref{thm:relabelexptime}, relabellings can be guessed  and verified in polynomial time; additionally, since the game has bounded \Call depth $k$, strategy decisions and relabellings are guessed at most $k \cdot |w|$ times, which yields a polynomial time bound for the alternating algorithm.
	
	The lower bound for $k=1$ follows directly from the $\PSPACE$ lower bound proof for games with fixed replacement languages without transducers (\cite{SchusterS15}	, Proposition 10(b)), as that proof only requires relabelling.
\end{proof}

\begin{restate}{Theorem \ref{thm:functionalpspace}}
	For the class of games with functional relabelling transducers and unbounded replay, $\Safelr$ is $\PSPACE$-complete.
\end{restate}

\begin{proof}
	As for Theorems \ref{thm:relabelexptime} and \ref{thm:relabelpspace}, the upper bound uses the trivial algorithm simulating the game with functional relabelling transducers by moving through the input string in a left-to-right order (resetting its focus as necessary after \Call moves), guessing existentially for each closing tag a move for \pone and, in case of a \Call, also guessing the relabelling chosen by \ptwo and applying it to the input string. Since there is only at most a single possible rewriting for each string input into the relabelling transducer and relabellings are of linear size in the input, this is a nondeterministic polynomial-space algorithm witnessing membership of $\Safelr$ in $\NPSPACE = \PSPACE$.
	
	The lower bound for games with deterministic relabelling transducers is proven similarly to the one in Theorem \ref{thm:relabelexptime} by reduction from the membership problem for linear bounded \emph{deterministic} Turing machines. The simulation of a TM by a game works as in the proof of Theorem \ref{thm:boundedundecidable}, as that proof also just requires a deterministic transducer. For linear bounded TMs, a relabelling transducer suffices, since no additional tape cells beyond those provided by the input are ever inserted.
	
	Since all deterministic NWTs are also functional, the upper bound also applies to deterministic relabelling NWTs and the lower bound also to functional relabelling NWTs, which proves the claim.
\end{proof}

\begin{restate}{Theorem \ref{thm:functionalnp}}
	For any $k \geq 1$, for the class of games with functional relabelling transducers and bounded \Call depth $k$, $\Safelr$ is $\NP$-complete.
\end{restate}

\begin{proof}
	The upper bound for functional relabelling transducers uses the same nondeterministic algorithm as the one used for the upper bound of Theorem \ref{thm:functionalpspace}; the only difference is that, since the input game has \Call depth $k$, the algorithm has to guess nondeterministically and verify at most $k \cdot |w|$ strategy choices and relabellings, each of which may be done nondeterminstically in polynomial time. This yields the desired $\NP$ upper bound.
	
	We prove the lower bound for replay-free games with deterministic relabelling transducers by a reduction from the \algprobname{3SAT} problem: Given a propositional formula $\varphi = C_1 \land \ldots \land C_m$ over variables $x_1, \ldots, x_n$ where each clause $C_j$ is a disjunction of exactly three literals, is there an assignment $\alpha: \{x_1, \ldots, x_n\} \rightarrow \{0,1\}$ such that $\varphi$ evaluates to 1 under $\alpha$?
	
	We construct from $\varphi$ a game $G = (\Sigma, \funcsymb, R, T)$ and string $w \in \wf$ such that \pone has a winning strategy on $w$ in $G$ if and only if $\varphi$ is satisfiable.The basic idea behind the reduction is that the string $w$ contains \emph{variable substrings}, on which \pone chooses an assignment $\alpha$ by using \Read and \Call moves, and a \emph{clause substring}, which models the structure of $\varphi$ and which gets rewritten after \pones choice of variable assignment into a form that allows a target DNWA to check whether $\alpha$ satisfies $\varphi$.
	
	More concretely, the clause substring is 
	$$w_\varphi = \op{C_1}\op{\ell_{1,1}}\cl{\ell_{1,1}}\op{\ell_{1,2}}\cl{\ell_{1,2}}\op{\ell_{1,3}}\cl{\ell_{1,3}}\cl{C_1} \cdots \op{C_m}\op{\ell_{m,1}}\cl{\ell_{m,1}}\op{\ell_{m,2}}\cl{\ell_{m,2}}\op{\ell_{m,3}}\cl{\ell_{m,3}}\cl{C_m},$$ 
	where for each $i \in [m]$, $j \in [3]$ and $k \in [n]$,
	\begin{displaymath}
		\ell_{i,j} =
			\left\{
			\begin{array}{ll}
				x_k, & \text{ if $x_k$ is the $j$-th literal of $C_i$} \\
				\overline{x_k}, & \text{ if $\neg x_k$ is the $j$-th literal of $C_i$.} \\
			\end{array}
			\right.
	\end{displaymath}
	
	The input string $w$ is constructed inductively; we set $w = w_1$, where for each $i \in [n+1]$,
	\begin{displaymath}
		w_i =
			\left\{
			\begin{array}{ll}
				\op{r}\op{y_i^0}\cl{y_i^0} w_{i+1}\cl{r}, & \text{ for $i \in [n]$ } \\
				w_\varphi, & \text{ for $i=n+1$.} \\
			\end{array}
			\right.
	\end{displaymath}
	
	Play on $w$ is supposed to proceed as follows: \pone first goes through the variable substrings of the form $\op{y_i^0}\cl{y_i^0}$, choosing whether or not to call each $\cl{y_i^0}$ in turn. On a \Call, $y_i^0$ gets relabelled to $y_i^1$. In this way, \pone chooses a variable assignment $\alpha$, with a remaining $y_i^0$ interpreted as $\alpha(x_i) = 0$ and a rewritten $y_i^1$ interpreted as $\alpha(x_i) = 1$.
	
	Afterwards, \pone is supposed to call each $\cl{r}$ in turn. On each such \Call, the relabelling transducer $T$ takes the valuation $y_i^0$ (or $y_i^1$ respectively) of the variable immediately following the corresponding $\op{r}$ and rewrites all occurrences of $x_i$ in $w_\varphi$ into $0$ and all $\overline{x_i}$ into $1$ (or $x_i$ into $1$ and $\overline{x_i}$ into $0$, respectively). After all literals in $w_\varphi$ have been rewritten into $0$ or $1$, the target DNWA simply needs to check whether, for each $j \in [m]$, the substring enclosed in $C_j$ tags contains at least one substring $\op{1}\cl{1}$.
	
	More formally, the game $G$ uses the alphabet $\Sigma = \{x_i, \overline{x_i}, y_i^0, y_i^1 \mid i \in [n]\} \cup \{C_j \mid j \in [m]\} \cup \{r\}$ with function symbols $\funcsymb = \{r\} \cup \{y_i^0 \mid i \in [n]\}$.
	
	The relabelling transducer $R$ behaves as follows:
	\begin{itemize}
	\item $R$ rewrites each input of the form $\op{y_i^0}\cl{y_i^0}$ into $\op{y_i^1}\cl{y_i^1}$.
	\item On an input of the form $\op{r}\op{\ell}\cl{\ell} v \cl{r}$ (for $\ell \in \{y_i^0, y_i^1 \mid i \in [n]\}$), $R$ memorises $\ell$ in its state and performs a relabelling on $v$:
		\begin{itemize}
		\item If $\ell = y_i^0$ for some $i \in [n]$, then $R$ relabels each $x_i$ in $v$ to $0$ and each $\overline{x_i}$ to $1$, and
		\item If $\ell = y_i^1$ for some $i \in [n]$, then $R$ relabels each $x_i$ in $v$ to $1$ and each $\overline{x_i}$ to $0$.
		\end{itemize}
	\end{itemize}

	The target language $T$ contains all strings $w'$ of the form of $w_1$ as defined above, with the following modifications:
	\begin{itemize}
	\item $w'$ does not contain any labels from $\{x_i, \overline{x_i} \mid i \in [n]\}$, and
	\item each $\op{C_j}$ tag (for $j \in [m]$) is followed by at least one $\op{1}\cl{1}$ substring before the corresponding $\cl{C_j}$ tag.
	\end{itemize}

	The construction of $G$ and $w$ from $\varphi$ is obviously possible in polynomial time. Furthermore, by the above consideration on the construction and verification of assignments for $\varphi$, it is easy to see that \pone has a winning strategy on $w$ in $G$ if and only if $\varphi$ is satisfiable.
\end{proof}

\begin{restate}{Theorem \ref{thm:writeonce}}
	For the class of write-once games with functional relabelling transducers, $\Safelr$ is in $\PTIME$.
\end{restate}

\begin{proof}
	The crucial insight for this proof is the fact that games with functional replacement transducers are essentially solitaire games for \pone -- the result of any \Call on some substring is uniquely determined by that substring, with no choice for \ptwo.
	
	We utilise this fact by constructing from a given game $G = (\Sigma, \funcsymb, R, T)$ a (generally non-functional) relabelling NWT $R_\ponea$ such that for each input string $w \in \wf$, the image $R_\ponea(w)$ is exactly the set of all strings that $w$ can be rewritten into by some left-to-right sequence of \Read and \Call moves by \pone. Checking for the existence of a winning strategy for \pone on $w$ then simply amounts to checking whether $R_\ponea(w) \cap T \neq \emptyset$.
	
	The NWT $R_\ponea$ is constructed from $R$ by a simple modification: In its standard mode of operation, $R_\ponea$ simply reproduces its input. Before any opening tag, however, $R_\ponea$ may choose nondeterministically to start simulating $R$ beginning with the next opening tag, rewrite the input substring up until the corresponding closing tag and then return to simply reproducing its input (or starting another simulation of $R$).
	
	Now, for any $w,w' \in \wf$ it holds that $w' \in T(w)$ if and only if there is a write-once strategy for \pone rewriting $w$ into $w'$, which can be proven by a somewhat involved induction on the structure of $w$ and $w'$. The two main insights required for this proof are the facts that (i) due to \pone having complete information and $R$ being functional, \pone can make her decision whether to \Read or \Call already on \emph{opening} tags instead of the corresponding closing tags, and (ii) since $R_\ponea$ is a relabelling transducer, the relabelling of a closing tag according to $R_\ponea$ is already determined at the corresponding opening tag, which enforces the corresponding rewriting for \pone to be write-once.
	
	The construction of $R_\ponea$ is obviously possible in polynomial time. Furthermore, given $R_\ponea$ and an input string $w$, by Corollary \ref{cor:wellformeddomain} and Lemma \ref{lemma:regularrange} we can construct in polynomial time a polynomial-size $\epsilon$-NWA deciding $R_\ponea(w)$. Intersection nonemptiness with the target language $T$ can then be checked for in polynomial time due to Lemma \ref{lemma:epsnwabool} and Theorem \ref{thm:epsnwacomplexity}(a).
\end{proof}

\end{document}